\documentclass[aps,
	prd, 
	onecolumn, 
	a4paper, 
	10pt, 
	preprintnumbers,  
	tightenlines, 
	notitlepage,
	floatfix,
	nofootinbib,
	longbibliographyprd,
	superscriptaddress]{revtex4-1}

\input{CQGformat.input}
\newcommand{\mr}{\ms\hline\ms}
\usepackage{etoolbox}
 
\makeatletter
\def\@mkboth#1#2{}
\newlength\appendixwidth
\preto\appendix{\addtocontents{toc}{\protect\patchl@section}}
\newcommand{\patchl@section}{%
  \settowidth{\appendixwidth}{\textbf{Appendix }}%
  \addtolength{\appendixwidth}{1.5em}%
  \patchcmd{\l@section}{1.5em}{\appendixwidth}{}{\ddt}%
}
\makeatother

\usepackage{IEEEtrantools}
\usepackage{microtype}
\usepackage[toc,title,page,titletoc]{appendix}
\usepackage{amsmath}
\usepackage{amssymb}
\usepackage{amsthm}
\usepackage{mathrsfs}
\usepackage{enumitem}
\usepackage[T1]{fontenc}
\usepackage{xfrac}
\usepackage{lmodern}
\usepackage[pdftex,breaklinks,colorlinks,
linkcolor=Blue,
citecolor=teal,
anchorcolor=red,
urlcolor=cyan]{hyperref}
\usepackage{mathbbol}
\usepackage[dvipsnames]{xcolor}
\usepackage[graphicx]{realboxes}
\usepackage{adjustbox}

\newcommand{\GHPw}[2]{\left\{ #1, #2 \right\}}

\renewcommand{\S}{{\mathcal S}}
\newcommand{\E}{{\mathcal E}}
\newcommand{\T}{{\mathcal T}}

\renewcommand{\O}{{\mathcal O}}

\newcommand{\A}{{\mathcal A}}

\newcommand{\sM}{{\mathscr M}}
\newcommand{\sS}{{\mathscr S}}
\newcommand{\sI}{{\mathscr I}}
\newcommand{\sH}{{\mathscr H}}
\newcommand{\sP}{{\mathscr P}}
\newcommand{\sL}{{\mathscr L}}
\newcommand{\sV}{{\mathscr V}}
\newcommand{\sW}{{\mathscr W}}

\usepackage[normalem]{ulem}

\font\ec=ecrm0800 at 12pt
\def\th{\hbox{\ec\char'336}}

\def\mb{\bar m}

\def\pheq{\phantom{=}}

\DeclareMathOperator{\thorn}{\text{\rm \th}}
\let\eth\relax
\DeclareMathOperator{\eth}{\text{\rm \dh}}

\newtheorem{theorem}{Theorem}
\newtheorem*{corollary}{Corollary}
\newtheorem{lemma}[theorem]{Lemma}
\newtheorem{definition}{Definition}
\newtheorem{proposition}{Proposition}

\renewcommand{\Re}{\operatorname{Re}}

\newcommand{\half}{\dfrac{1}{2}}

\newcommand{\intd}[1]{\textrm{d} #1}
\newcommand{\lie}[1]{\mathcal{L}_#1}

\newcommand{\dd}{{\rm d}}

\def\ben{\begin{equation}}
\def\een{\end{equation}}
\def\bena{\begin{eqnarray}}
\def\eena{\end{eqnarray}}

\usepackage{tikz}
\usetikzlibrary{patterns}
\usetikzlibrary{shapes.geometric}
\usetikzlibrary{decorations.pathmorphing}
\usetikzlibrary{arrows.meta}
\tikzset{snake it/.style={decorate, decoration=snake}}
\usetikzlibrary{positioning}
\usetikzlibrary{arrows,shapes,positioning}
\usetikzlibrary{decorations.markings}
\tikzstyle arrowstyle=[scale=1]
\tikzstyle directed=[postaction={decorate,decoration={markings,mark=at position .65 with {\arrow[arrowstyle]{stealth}}}}]
\tikzstyle reverse directed=[postaction={decorate,decoration={markings,mark=at position .65 with {\arrowreversed[arrowstyle]{stealth};}}}]

\begin{document}

\title{Metric Reconstruction in Kerr Spacetime}

\author{Stefan Hollands$^{1,2}$, Vahid Toomani}

\address{$^1$ Institut f\"ur Theoretische Physik, Universit\"at Leipzig, Br\"uderstra{\ss}e 16, 04103 Leipzig, Germany.}

\address{$^2$ Max-Planck-Institute MiS, Inselstrasse 22, 04103 Leipzig, Germany.}



\ead{
stefan.hollands@uni-leipzig.de, 
vahidtoomani2002@gmail.com}

\date{\today}

\begin{abstract}
Metric reconstruction is the general problem of parameterizing GR in terms of its two ``true degrees  
of freedom'', e.g., by a complex scalar ``potential''---in  
practice mostly with the aim of simplifying the Einstein equation (EE) within perturbative approaches.  
In this paper, we re-analyze the metric reconstruction procedure by Green, Hollands, and Zimmerman (GHZ) [Class. Quant. Grav. \textbf{37}, 075001 (2020)], which is a 
generalization of the Chrzanowski-Cohen-Kegeles (CCK) approach. Contrary to the CCK method, that by GHZ is applicable not only to the vacuum, but also to 
the sourced linearized Einstein equation (EE). Our main innovation is a version of the GHZ method 
giving an efficient integration scheme for the initial value problem of the sourced linear EE. By iteration, our scheme gives the metric to as high an order in perturbation theory around Kerr as one might wish, in principle. At each order, the metric perturbation is a sum of a corrector, obtained by solving a triangular system of transport equations, a reconstructed piece, obtained from a Hertz potential as in the CCK approach, and an algebraically special perturbation, determined by the ADM quantities. As a by-product, we determine the precise relations between the asymptotic tail of the Hertz potential in the GHZ and CCK schemes, and the quantities 
relevant for gravitational radiation, namely, the energy flux, news- and memory tensors, and their associated BMS-supertranslations. We also discuss ways of transforming the metric perturbation to Lorenz gauge.
\end{abstract}

\maketitle
\tableofcontents

\section{Introduction}

\subsection{Metric reconstruction in Kerr}

When analyzing a disturbance of the Kerr spacetime caused, e.g., by a small body orbiting the hole, or by a gravitational wave impinging on it, it is natural to start with the linearized Einstein equation (EE), 
\begin{equation}\label{eq:Eh=T}
(\E h)_{ab}=T_{ab},
\end{equation}
where $\E$ is the linearized Einstein operator on the Kerr background \eqref{eq:linearE}. The stress energy tensor, $T_{ab}$, represents any matter, or, in an order-by-order 
perturbative approach to the non-linear EE, the driver of the metric perturbation at next order caused by the lower orders. 
The linearized---let alone non-linear\footnote{The stability of at the  non-linear level has been shown by \cite{Klainerman:2021qzy} for slowly rotating Kerr black holes, by \cite{Hintz2} for slowly rotating Kerr-deSitter black holes and by \cite{Dafermos:2021cbw} for Schwarzschild by an independent analysis.}---EE equations are extremely difficult to analyze, either from a mathematical viewpoint, such as when estimating the decay of solutions
(see, e.g., \cite{Holzegel} for a review), or from a practical viewpoint, such as
when studying the dynamics of the small body and its gravitational radiation in the self-force approach (see, e.g., \cite{Poisson:2011nh,Pound:Wardell,Barack} for reviews). 

One of the main difficulties associated with \eqref{eq:Eh=T} is that, on a Kerr background, the components of $h_{ab}$ are coupled in an intricate way.
A second difficulty arises from the invariance of \eqref{eq:Eh=T} under gauge transformations, implying that, in practice, one often has to adopt different gauges 
in different parts of spacetime in order to analyze the solution. A third---mostly practical---difficulty, strongly inhibiting the use of frequency domain techniques, 
is that there is no known\footnote{To the authors.} basis of mode functions in which \eqref{eq:Eh=T} directly becomes amenable to the separation of variables method. 
From either of these points of view, it is much more convenient to work with the perturbed extreme Weyl scalars, $\psi_0$ and $\psi_4$, instead of $h_{ab}$: These are not only scalar, but also gauge invariant and obey 
separately Teukolsky's equation \cite{Teukolsky:1972my}, which can be solved by a separation of variables ansatz \cite{Teukolsky:1973ha}. At linear order, $\psi_0$ and $\psi_4$
are entirely sufficient to analyze gravitational radiation interacting with the Kerr black hole. For recent mathematical studies of Teukolsky's equation
directly relevant for this paper, see \cite{DHR,Ma,Hintz,Angelopoulos,Shlapentokh-Rothman:2023bwo,Andersson}.

However, if one wants go beyond linear order, one requires 
the metric perturbation $h_{ab}$ {\it itself}, and not just $\psi_0$ and/or $\psi_4$, 
to compute its second or higher order Einstein tensor sourcing the second or higher order perturbation. 
Similarly, $h_{ab}$ is required, e.g., in order to evaluate the self-force on a small body in orbit around the hole to first or higher order \cite{van2015metric,meent2,Keidl1,Keidl2,Pound:2012dk,Spiers:2023mor,Spiers:2023cip,Wardell:2021fyy}. The problem of re-obtaining $h_{ab}$---to the extent possible---from $\psi_0$, $\psi_4$, or perhaps some other scalar field(s) obeying separable wave equations, is called ``metric reconstruction''. 

For the {\it homogeneous} version of the EE \eqref{eq:Eh=T}, i.e., when $T_{ab}=0$, the reconstruction problem was addressed by \cite{Chrzanowski:1975wv,Kegeles:1979an}, 
who showed that a metric perturbation of the form
\begin{equation}
\label{eq:reconstr}
h_{ab} = \Re \S^\dagger_{ab} \Phi
\end{equation}
is a solution to the homogeneous EE if $\Phi$ satisfies the (adjoint) Teukolsky equation, where $\S^\dagger_{ab}$ is a certain second order partial differential operator 
that we recall below in \eqref{eq:Sdag}. Furthermore, it is generally taken for granted, and will be justified by a complete proof for the first time in corollary \ref{GHZcor}, that, up to possibly an infinitesimal perturbation corresponding to a variation of the parameters $(M,a)$ of Kerr, and up to a gauge transformation, essentially {\it any} physically relevant linear perturbation solving the homogeneous EE can be written in this way. Given the significant advantages of Teukolsky's equation over the full system of EE, it would clearly be highly desirable if solutions to the {\it inhomogenous}
EE could similarly be written in ``reconstructed form'' \eqref{eq:reconstr}. 

Unfortunately, as we recall in Sec. \ref{sec:GHZ decomposition}, for general sources this hope is illusory for simple algebraic reasons, and so a more sophisticated theory of metric reconstruction than that 
by \cite{Chrzanowski:1975wv,Kegeles:1979an} must be sought. Such a scheme would naturally be deemed satisfactory if (a) {\it any} physically reasonable perturbation solving the inhomogeneous EE for any physically reasonable $T_{ab}$ (especially ones coming up in perturbative approaches)
can be reconstructed, (b) the ingredients in the reconstruction obey differential equations that are less complicated than the full EE, such as e.g., Teukolsky or simple transport equations, (c) there exist numerically efficient ways of solving these differential equations. Current proposals include\footnote{
We do not include in this list other integration schemes such as \cite{Benomio:2022tfe} or \cite{Klainerman:2021qzy}, \cite{Dafermos:2021cbw}, aimed more broadly at a stability/decay analysis of the full non-linear EE, rather than at reducing the EE, order-by-order in perturbation theory, to a Teukolsky equation. At the linearized level, \cite{Benomio:2022tfe} has some similarities with \cite{Andersson:2019dwi,Loutrel:2020wbw} described below, though it is better optimized for good control near $\sH^+$. 
}:

\begin{itemize}
\item 
The method initiated by Keidl et al. \cite{Keidl1,Keidl2} applies to the important special case of a point particle source \eqref{eq:pp}. It works by gluing two vacuum solutions in reconstructed form \eqref{eq:reconstr} along an interface at the particle's instantaneous radial position, and underlies many first order self-force calculations. Prominently, it has been used by van de Meent \cite{meent2} to fully compute the self-force for general bound orbits in Kerr.\footnote{See \cite{Barack:2017oir} for an implementation of this method in the time domain.} While very reasonable and efficient in this setting, this method is by its very construction not applicable beyond first order perturbations, nor to extended sources even at first order. 

\item Loutrel et al. \cite{Loutrel:2020wbw} have described a method for reconstructing the metric perturbation from Weyl scalars at first and second order in perturbation theory, but 
assuming a non-generic algebraic condition on the stress tensor (closely related to the outgoing radiation gauge, ORG). Thus, while useful in special cases where this condition is satisfied, or approximately so, the method falls short of satisfying criterion (a).

\item Andersson et al. \cite{Andersson:2019dwi} have used an integration scheme quite similar to \cite{Loutrel:2020wbw} to prove that one can, in a sense, transfer decay results for the 
Weyl scalars such as \cite{DHR, Ma} to $h_{ab}$. Though different from the reconstruction method by \cite{Chrzanowski:1975wv,Kegeles:1979an}
described above, their scheme is also formulated only for the homogeneous EE ($T_{ab}=0$), and thus not what we are looking for either.

\item Andersson et al. \cite{Andersson:2021eqc} have established a non-linear version of the so-called ingoing radiation gauge (IRG), $h_{ab}l^b=0$, where $l^a$ is the outgoing principal null direction, see Fig. \ref{fig:2}.
In this sense, their scheme is quite similar to the reconstruction method by \cite{Chrzanowski:1975wv,Kegeles:1979an} described above, which also yields a perturbation in (traceless) IRG automatically. Furthermore, 
if one were to expand their method perturbatively, their IRG construction would apply order-by-order to the higher perturbations  sourced by the $T_{ab}$ corresponding to the lower orders. However, \cite{Andersson:2021eqc} does not show that the IRG metric is generated by one or more scalars satisfying ``simple'' differential equations in the sense of (b).  
Thus, their scheme is not really a reconstruction method and for this reason alone\footnote{Another difference is that \cite{Andersson:2021eqc} employs several tetrad frames, whereas we will stick to one and only such frame (up to equivalence), tied to the principal null directions of Kerr.} quite different from what we are looking for in this paper. 

\item The AAB method \cite{Aksteiner:2016pjt} (after Aksteiner, Andersson and B\" ackdahl) uses {\it two} potentials (obtained as anti-time-derivatives of $\psi_0$ and $\psi_4$), 
as well as a correction piece which is an anti-time-derivative of certain other derivatives of $T_{ab}$. We give a more conceptual derivation of their decomposition in \ref{sec:TScov}, see \eqref{AAB}, based on an interesting  ``Nullstellensatz'' for the EE on Kerr (theorem \ref{lem:null} in \ref{sec:NSsatz}). The AAB method has, perhaps, an intrinsic lack of economy, because one complex scalar should, in principle, suffice to describe the two real propagating degrees of freedom of General Relativity.
Nevertheless, the AAB scheme is a contender in the sense of (a)--(c) above. 

\item The GHZ method \cite{Green:2019nam} (after Green, Hollands and Zimmerman) is closer to the metric reconstruction approach for the homogeneous EE
by \cite{Chrzanowski:1975wv,Kegeles:1979an} than the AAB method. In addition to the Hertz potential, $\Phi$, it requires a correction piece obeying certain decoupled transport equations involving a general source $T_{ab}$. Thus, the GHZ scheme also is compatible with (a)--(c) above.
\end{itemize}

The AAB and GHZ methods therefore seem to be the most generally applicable metric reconstruction schemes, and both deserve to be studied further. Both the AAB and GHZ approaches require further gauge changes to be practically useful, see Sec. \ref{sec:gaugeissues}.

In this work, we will focus on the GHZ method (apart from \ref{sec:TScov}). Our main advance is an extension of that method capable of reducing the Cauchy problem for the sourced linearized EE to a Teukolsky equation and transport equations, along with a precise characterization relating the asymptotic expansion of the GHZ decomposition to the observables at future null infinity, $\sI^+$. 

\subsection{Outline of GHZ method}
\label{sec:GHZ decomposition}

As we have said, the GHZ method \cite{Green:2019nam} is a possible generalization of the metric reconstruction procedure by \cite{Chrzanowski:1975wv,Kegeles:1979an} for 
solutions $h_{ab}$ to the {\it inhomogeneous} EE \eqref{eq:Eh=T}. The method considers the following decomposition, 
\begin{equation}\label{eq:decompi}
h_{ab} =  \Re\, \S^\dagger_{ab} \Phi + \dot g_{ab} + {\mathcal L}_X g_{ab} + x_{ab}, 
\end{equation}
into a reconstructed piece, $\Re\,\S^\dagger_{ab} \Phi$ [see \eqref{eq:Sdag}], a perturbation $\dot g_{ab}$ corresponding to varying the hole's mass and spin (called ``zero mode'', see \ref{app:D}), a pure gauge 
perturbation, ${\mathcal L}_X g_{ab}$, and a so-called ``corrector'', $x_{ab}$. A corrector is necessary for sufficiently general $T_{ab}$, 
as can be seen as follows \cite{Price,Price2}: Both $\dot g_{ab}$ and ${\mathcal L}_X g_{ab}$ are solutions to the homogeneous EE, i.e., they are annihilated 
by $\E$. $\Re\,\S^\dagger_{ab} \Phi$ automatically is in traceless ingoing radiation gauge (TIRG), meaning $h_{ab}l^a = 0 = g^{ab} h_{ab}$. Inspection of the $ll$ NP component of the linearized Einstein operator in GHP form given in \ref{app:LinEinGHP} shows that, for any metric perturbation in TIRG, $(\E h)_{ll} = 0$. Thus, if $(\E h)_{ll} \neq 0$, or said differently, if $T_{ll} \neq 0$, then it is simply impossible to write $h_{ab}$ in reconstructed form up to a gauge- and algebraically special perturbation, i.e., we need a suitable corrector, $x_{ab}$.

It turns out \cite{Green:2019nam} that the $l$ NP components of the stress energy tensor $T_{ab}$ in the EE \eqref{eq:Eh=T} are the essential obstacle to write $h_{ab}$ in reconstructed form---thus in particular in TIRG. The corrector is thus constructed in such a way that $h_{ab}-x_{ab}$ has a stress tensor 
without any $l$ NP component, i.e., following \cite{Green:2019nam} we seek an $x_{ab}$ such that
\begin{equation}\label{xeq}
(\mathcal{E} x)_{ab} l^b = (\E h)_{ab} l^b,
\end{equation}
so that $h_{ab} - x_{ab}$ satisfies the linearized EE with a new source $S_{ab} = T_{ab} - (\mathcal{E} x)_{ab}$ such that $S_{ab}l^b = 0$. 
\cite{Green:2019nam} chose a $x_{ab}$ having only $nn, nm, m\bar{m}$ NP components, and observed that, with such a choice, all three NP components of \eqref{xeq} become {\em ordinary} differential equations for these components. More precisely, consider an $x_{ab}$ with the NP components
\begin{equation}
\label{eq:xdef}
x_{ab} = 2m_{(a} \bar{m}_{b)} x_{m\bar{m}} -2l_{(a} \bar{m}_{b)} x_{nm} -2l_{(a} m_{b)} x_{n\bar{m}} + l_a l_b x_{nn}.
\end{equation}
Then one sees \cite{Green:2019nam}, using formulas of \ref{app:LinEinGHP}, that the NP components of  \eqref{xeq} are
schematically of the form
%
\begin{equation}
\label{eq:xmmb}
\rho^2 \thorn \left[ \frac{\bar{\rho}}{\rho^3} \thorn \left(\frac{\rho}{\bar{\rho}} x_{m\bar{m}} \right) \right] = T_{ll}
\end{equation}
for $x_{m\bar{m}}$ 
\begin{equation}
\label{eq:xnm1}
\frac{\rho}{2(\rho+\bar\rho)} \thorn
\left\{ (\rho+\bar{\rho})^2
\thorn \left[ \frac{1}{\rho(\rho+\bar\rho)} x_{nm} \right] \right\} = T_{lm}
\text{$+$ (terms involving $x_{m\bar{m}}$)}
\end{equation}
for $x_{nm}$, and
\begin{equation}
\label{eq:xnn1}
\half (\rho+\bar{\rho})^2\thorn \left( \frac{1}{\rho+\bar\rho}x_{nn} \right) = T_{ln}
\text{$+$ (terms involving $x_{m\bar{m}},x_{nm}$)}
\end{equation}
for $x_{nn}$. This system, called ``GHZ transport equations'', for $x_{nn}, x_{nm}, x_{m\bar{m}}$ is triangular because the integration of the equation for  $x_{m\bar{m}}$ does not involve  either $x_{nn}, x_{nm}$, and the integration for $x_{nm}$ does not involve $x_{nn}$. Furthermore, 
\begin{equation}
\thorn = \partial/\partial r , \quad \rho = -(r-ia\cos \theta)^{-1},
\end{equation}
 in retarded Kerr-Newman coordinates $(x^\mu)=(u,r,\theta,\varphi_*)$ and the Kinnersley frame (see Sec. \ref{sec:Kpert}), so solving the GHZ transport equations  involves $r$-integrations at constant $(u,\varphi_*,\theta)$, see Figs. \ref{fig:3}, \ref{fig:4}.
Eqs. \eqref{eq:xnm1}, \eqref{eq:xnn1} are written out more fully in \eqref{eq:xnm}, \eqref{eq:xnn} below.

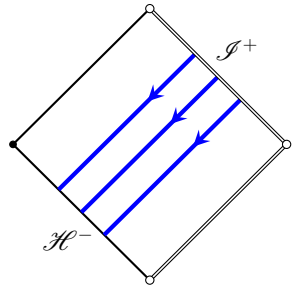
\begin{figure} 
\begin{center}
\begin{tikzpicture}[scale=0.6, transform shape]
\draw[thick](3,-3)--(0,0)--(3,3);
\draw[blue, reverse directed, ultra thick](1,-1)--(4,2);
\draw[blue, reverse directed, ultra thick](1.5,-1.5)--(4.5,1.5);
\draw[blue, reverse directed, ultra thick](2,-2)--(5,1.0);

\draw[double](3,3)--(6,0)--(3,-3);

\node[anchor=west] at(4.3,2.1) {{\Large ${\mathscr I}^+$}};
\draw (0,0) node[draw,shape=circle,scale=0.4,fill=black]{};
\node[anchor=east] at(1.9,-2.1) {{\Large ${\mathscr H}^-$}};


\draw (3,3) node[draw,shape=circle,scale=0.5,fill=white]{};
\draw (6,0) node[draw,shape=circle,scale=0.5,fill=white]{};
\draw (3,-3) node[draw,shape=circle,scale=0.5,fill=white]{};

\end{tikzpicture}
\end{center}
\caption{\label{fig:3} Backward integration contour for GHZ transport equations.}
\end{figure}

Once the corrector $x_{ab}$ has been determined, the second step in the GHZ approach is to show that the remaining part of the metric, that is $h_{ab}-x_{ab}$, can be written in reconstructed form ${\rm Re}\, \S^\dagger_{ab} \Phi$ for some $\Phi$, plus a  suitable pure gauge perturbation plus a suitable zero mode.  In practice, once we {\it know} that the GHZ decomposition exists, the only remaining task is to have a convenient {\it algorithm} for actually 
{\it finding} the Hertz potential $\Phi$  and the corrector $x_{ab}$ in the GHZ decomposition. 
It seems that the following sequence of steps is a particularly convenient integration scheme:

\medskip
\noindent
{\bf Step 1: GHZ transport equations.}
Find $x_{ab}$, which only involves solving the GHZ transport equations with boundary conditions concretely determined by the problem at hand, e.g., retarded or advanced solutions. In \cite{Casals:2024ynr} explicit Green's functions for the GHZ transport equations were provided which may be employed for this purpose.

\medskip
\noindent
{\bf Step 2: Teukolsky equation.}
Finding $\Phi$ is easiest solving first the initial value problem for the usual sourced Teukolsky equation for the Weyl scalar $\psi_0$, e.g., by standard frequency domain techniques. 

\medskip
\noindent
{\bf Step 3: Inversion.}
With $\psi_0$ at hand, one obtains $\Phi$ as the appropriate solution to the transport equation $-\frac{1}{4} \thorn^4 \bar \Phi = \psi_0$. For retarded or advanced solutions, this inversion can be implemented e.g., using radial 
Teukolsky-Starobinsky (TS) identities in the frequency domain \cite{Ori:2002uv}.  On the other hand, as we demonstrate in Sec. \ref{sec:Cauchy}, if we want to solve an initial value problem for the sourced EE, 
this inversion step involves angular TS identities.

\medskip
\noindent
{\bf Step 4: Iteration.} This procedure can be iterated to obtain the metric $g_{ab}(\epsilon) = g_{ab} + \epsilon h_{ab}^{(1)} + \epsilon^2 h_{ab}^{(2)} + \dots,$ solving the EE with source $T_{ab}(\epsilon) = 
\epsilon T_{ab}^{(1)} + \epsilon^2 T_{ab}^{(2)} + \dots$ to increasing accuracy. For the extreme mass ratio inspiral problem, the expansion parameter $\epsilon = \mu/M$ would be the ratio of the particle's mass and that of the black hole, $T^{(1)}_{ab}$ would be the stress tensor for a point particle  on a bound geodesic $\gamma$ orbiting the hole,
\begin{equation}
\label{eq:pp}
     T^{(1)}_{ab}(x) \propto \int \dd \tau \, \dot \gamma_a \dot \gamma_b \delta^4(x-\gamma(\tau))/\sqrt{-g}, 
\end{equation}
$T^{(2)}_{ab}$ would e.g., account for the particle's spin
\cite{Mathews}. We would run the GHZ procedure to first get $h^{(1)}_{ab}$
solving \eqref{eq:Eh=T} with source $T^{(1)}_{ab}$. Then we would suitably adjust the gauge of $h^{(1)}_{ab}$, see Sec. \ref{sec:gaugeissues}, and next solve \eqref{eq:Eh=T} for $h^{(2)}_{ab}$ with source $T^{(2)}_{ab} - G^{(2)}_{ab}(h^{(1)},h^{(1)})$, 
where $G^{(2)}_{ab}$ is the second order Einstein tensor (see \cite{Spiers:2023cip} for its GHP form), and so on.

\medskip
\noindent
Of course, to justify this scheme, we must prove that the 
perturbation actually {\it has} a GHZ decomposition \eqref{eq:decompi} with the desired properties in the first place! The difficult step is showing that, after the corrector $x_{ab}$ has been determined, $h_{ab}-x_{ab}$ can, after all, be written in reconstructed form ${\rm Re}\, \S^\dagger_{ab} \Phi$ for some $\Phi$, plus a  suitable pure gauge perturbation plus a suitable zero mode. To demonstrate this, we require certain decay conditions on $h_{ab}$ at $\sI^+$ or certain boundary conditions at $\sH^-$. In \cite{Green:2019nam}, the authors required that $h_{ab}$ be smooth and vanish in an open neighborhood of $\sH^-$ in some gauge, as would be the case for a retarded solution associated with a $T_{ab}$ of compact support in the asymptotic past. A trivial modification of their argument would also establish the GHZ decomposition if, e.g., $h_{ab}$ were to vanish in an open neighborhood of $\sI^+$ in some gauge, as would be the case for an advanced solution associated with a $T_{ab}$ of compact support in the asymptotic future.

While the condition that $h_{ab}$ vanish (in some gauge) in an open neighborhood of $\sH^-$ or $\sI^+$ is appropriate for retarded or advanced solutions to the linearized EE, such conditions clearly are unsuitable e.g., if one would like to use the GHZ scheme to solve the Cauchy problem with initial conditions posed on some Cauchy surface of the hole's exterior. The main aim of the present paper is to suitably generalize the GHZ decomposition theorem \cite{Green:2019nam} so as to incorporate such situations, and thereby be able to 
make use of the GHZ integration scheme for solving the Cauchy problem for the sourced linearized EE. 

More precisely, we will prove the GHZ decomposition in the following situations:
\begin{itemize}
\item[1)] metric perturbations $h_{ab}$ and stress tensors $T_{ab}$ having appropriate ``Bondi-like'' decay at $\sI^+$ (see \ref{BondiIRG}), or
\item[2)] metric perturbations $h_{ab}$ and stress tensors $T_{ab}$ satisfying appropriate ``no outgoing radiation'' type conditions at $\sH^-$. 
\end{itemize}
While the general strategy is essentially the same as in \cite{Green:2019nam}, very considerable extra work is required to carry them out in the more general setting of the present paper.  In case 1), we will consider a particular boundary condition for the various transport equations, effectively at $\sI^+$, corresponding to a ``backward integration'' scheme, 
see Fig. \ref{fig:3} for the integration contours. To this end, we will go through the following steps in Sec. \ref{GHZproof}:

\begin{itemize}
\item
We state our global decay assumptions at $\sI^+$ on $h_{ab}$ and $T_{ab}$ and extract from $h_{ab}$ the zero mode, $\dot g_{ab}$, 
see Sec. \ref{sec:gdot}.

\item
We perform a gauge transformation with gauge vector field $\xi^a$, eliminating any $l$ NP component of $h_{ab}$, see 
Sec. \ref{sec:gaugexi}.

\item
We obtain the corrector tensor $x_{ab}$, see Sec. \ref{sec:xabdef}.

\item
We perform a residual gauge transformation with gauge vector field
$\zeta^a$, allowing us to find the Hertz potential $\Phi$ in 
Secs. \ref{Hertzbackward} and \ref{sec:zetaproof}. The combination $X^a = \xi^a + \zeta^a$ 
is the gauge vector field appearing in the GHZ decomposition.
\end{itemize}

We will also give a GHZ decomposition corresponding to case 2). The resulting ``forward integration'' scheme goes through similar steps and is described in Sec. \ref{GHZproofH-}, 
see Fig. \ref{fig:4} for the integration contours. Many auxiliary results required in Secs. \ref{GHZproof}, \ref{GHZproofH-} are relegated to various appendices.

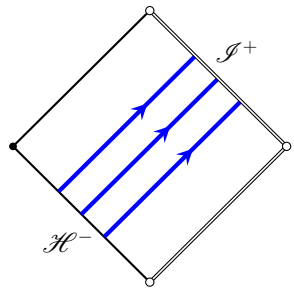
\begin{figure} 
\begin{center}
\begin{tikzpicture}[scale=0.6, transform shape]
\draw[thick](3,-3)--(0,0)--(3,3);
\draw[blue, directed, ultra thick](1,-1)--(4,2);
\draw[blue, directed, ultra thick](1.5,-1.5)--(4.5,1.5);
\draw[blue, directed, ultra thick](2,-2)--(5,1.0);

\draw[double](3,3)--(6,0)--(3,-3);

\node[anchor=west] at(4.3,2.1) {{\Large ${\mathscr I}^+$}};
\draw (0,0) node[draw,shape=circle,scale=0.4,fill=black]{};
\node[anchor=east] at(1.9,-2.1) {{\Large ${\mathscr H}^-$}};


\draw (3,3) node[draw,shape=circle,scale=0.5,fill=white]{};
\draw (6,0) node[draw,shape=circle,scale=0.5,fill=white]{};
\draw (3,-3) node[draw,shape=circle,scale=0.5,fill=white]{};

\end{tikzpicture}
\end{center}
\caption{\label{fig:4} Forward integration contours for GHZ transport equations.}
\end{figure}

\subsection{Gauge issues}
\label{sec:gaugeissues}

The metric reconstruction approach \cite{Chrzanowski:1975wv,Kegeles:1979an}, and its possible generalization, the GHZ method, yields a metric 
perturbation in TIRG respectively IRG [or ORG, after changing $(\sI^+, \sH^-) \leftrightarrow (\sI^-, \sH^+)$ and applying the GHP priming operation \cite{GHP} to all our relations]. This has certain disadvantages: As one may see e.g., from theorem \ref{thm:2}, some metric components in IRG may fail to 
decay towards $\sI^+$. Another issue, analyzed from the point of view of ``wave front sets'' \cite{hormander1} in \cite{Casals:2024ynr}, 
is that singularities of the source $T_{ab}$ in the linearized EE
may be propagated---as gauge singularities---off the source. A manifestation of this phenomenon in the case of a source describing a point-particle on a worldline, 
is the appearance of  ``string singularities'' along outgoing principal null geodesics emanating from the particle in IRG, see e.g., \cite{Pound:2013faa,HT2}. 

To remedy such shortcomings, a possible strategy is to construct explicitly a gauge vector field putting the linear perturbation into a more regular gauge. 
One possibility is the so-called ``no-string gauge'', treated in \cite{HT2}, see \cite{Bourg:2024vre} for a recent application to self-force problems. Another possibility is the Lorenz gauge, in which the linearized EE becomes manifestly hyperbolic, implying, e.g., the ``normal'' propagation of singularities \cite{Casals:2024ynr}.

As is well-known, in order to put a linear perturbation into the Lorenz gauge, one needs to solve a wave-equation for the prerequisite gauge vector field. Unfortunately, solving this wave equation on Kerr is roughly at the same level of difficulty as solving the linearized EE itself. Thus, one needs yet another reconstruction procedure, now for the gauge vector field, to make this task feasible in practice. 

\begin{itemize}
    \item If the metric perturbation is in reconstructed form, 
    $\Re\,\S^\dagger_{ab} \Phi$, then the prerequisite gauge VF can be obtained by a very efficient trick \cite{Dolan}, which in effect only requires taking anti-time-derivatives and derivatives of a perturbed Weyl scalar---a trivial task in the frequency domain.
    \item If the metric perturbation satisfies the EE with source, then we need to supplement the reconstructed part with a corrector, as e.g. in the GHZ or AAB prescriptions outlined above for spin-2. In this case, one can still reduce the problem of finding the gauge vector field to that of solving sourced massless scalar wave equations and sourced Maxwell equations, as we describe in Sec. \ref{sec:RedMax}. However, the task becomes more complicated.
\end{itemize}

The massless scalar wave equation already is a Teukolsky equation, whereas methods for reducing the sourced Maxwell equation to Teukolsky equations by a sort of reconstruction method (actually requiring two Hertz-type potentials) were given by \cite{Green2,Green3,Dolan2}.  
The prescription by \cite{Green2,Green3} has so far not appeared in the literature, so we will give a detailed account in somewhat modified form below
in Sec. \ref{SourcedMax}.  This same problem is also solved, in a different way, by the 
spin-1 analogue of the GHZ construction \cite{Hollands:2020vjg}, which we will recall in Sec. \ref{sec:MaxGHZ}, and which requires solving only one Teukolsky equation plus one transport equation, similarly to the GHZ method for spin-2 as outlined above.


\medskip
\noindent
{\bf Conventions:} We generally follow the conventions of 
\cite{Waldbook}, except for the following: the metric signature is $(+,-,-,-)$ in this paper, in order to be consistent with most of the literature about the Geroch-Held-Penrose (GHP) formalism \cite{GHP}. We use units in which $8\pi G=c=1$. 
$\lie{\xi}$ is the ordinary Lie-derivative when acting on tensors, but the GHP covariant Lie-derivative \cite{Edgar} when acting on quantities with non-trivial GHP weights.

\section{Kerr perturbation theory}

\subsection{Kerr metric, NP frames, and GHP formalism}
\label{sec:Kpert}


Many developments in the perturbation theory of the Kerr metric rely on the fact that it is a Petrov type D spacetime (see e.g., \cite{Waldbook}). This is also the case for the developments in this paper. A type D spacetime by definition has two repeated principal null directions, $n^a$ and $l^a$. In Boyer-Lindquist (BL) coordinates $(x^\mu) = (t,r,\theta,\varphi)$ in Kerr, they may be chosen as
\begin{subequations}\label{eq:Kintet}
  \begin{align}
    l^a &= \frac{1}{\Delta} \left[ (r^2+a^2) \left( \frac{\partial}{\partial t} \right)^a + \Delta \left( \frac{\partial}{\partial r} \right)^a + a \left( \frac{\partial}{\partial \varphi} \right)^a \right], \\
    n^a &= \frac{1}{2 \Sigma} \left[ (r^2+a^2) \left( \frac{\partial}{\partial t} \right)^a - \Delta \left( \frac{\partial}{\partial r} \right)^a + a \left( \frac{\partial}{\partial \varphi} \right)^a \right], \\
    m^a &= \frac{1}{\sqrt{2} (r+ia\cos\theta)} \left[ ia\sin\theta \left( \frac{\partial}{\partial t} \right)^a +
    \left( \frac{\partial}{\partial \theta} \right)^a + i\csc\theta \left( \frac{\partial}{\partial \varphi} \right)^a \right],
  \end{align}
\end{subequations}
where $\Delta = r^2 -2Mr+a^2$ and $\Sigma = r^2 + a^2 \cos^2 \theta$. We assume throughout that $|a| \le M$, and denote the outer horizon radius by $r_+=M+\sqrt{M^2-a^2}$. We restrict attention in this paper to the black hole's exterior, $r>r_+$.
As is commonly done, we have completed $n^a, l^a$ to a complex null (Newman-Penrose, NP) tetrad with the complex null vector $m^a$.
Eqs.  \eqref{eq:Kintet} are called the Kinnersley tetrad. 

Our formalism will be fully covariant with respect to local rescalings of the NP tetrad (see below), so all of our formulas can be trivally transformed from one NP frame to any other. Nevertheless, it will sometimes be useful to refer to a specific frame, e.g. when making statements about the decay of certain quantities at $\sI^+$, or their properties at $\sH^-$. For such purposes, the Kinnersley frame is preferable because it is regular at $\sH^-$ and well-behaved at $\sI^+$. 

The Kerr metric itself is
\begin{equation}
g^{ab} = 2(l^{(a}n^{b)}-\bar m^{(a} m^{b)}).
\end{equation}
This corresponds to the normalizations $l^a n_b = 1, m^a \bar m_a = -1$, with all other contractions equal to zero, and to the signature $(+,-,-,-)$. With the standard choice of time-orientation ($\nabla^a t$ future-directed), $l^a$ and $n^a$ are future directed, and the orbits of $l^a$ are outgoing to $\sI^+$ while the orbits of $n^a$ are ingoing at $\sH^+$, see Fig. \ref{fig:2}.

\usetikzlibrary{decorations.pathmorphing}
\tikzset{zigzag/.style={decorate, decoration=zigzag}}
\begin{figure} 
\begin{center}
\begin{tikzpicture}[scale=0.6, transform shape]
\draw[thick](3,-3)--(0,0)--(3,3);
\draw[blue,directed, ultra thick](1,-1)--(4,2);
\draw[red, reverse directed, ultra thick](1.0,1.0)--(4.0,-2.0);
\node[anchor=east] at(2.7,1) {{\Large \color{blue} $l^a$}};
\node[anchor=west] at(3.1,-.5) {{\Large \color{red} $n^a$}};

\draw[double](3,3)--(6,0)--(3,-3);

\node[anchor=west] at(4.3,2.1) {{\Large ${\mathscr I}^+$}};
\node[anchor=west] at(4.3,-2.1) {{\Large ${\mathscr I}^-$}};
\draw (0,0) node[draw,shape=circle,scale=0.4,fill=black]{};
\node[anchor=east] at(1.9,2.1) {{\Large ${\mathscr H}^+$}};
\node[anchor=east] at(1.9,-2.1) {{\Large ${\mathscr H}^-$}};


\draw (3,3) node[draw,shape=circle,scale=0.5,fill=white]{};
\draw (6,0) node[draw,shape=circle,scale=0.5,fill=white]{};
\draw (3,-3) node[draw,shape=circle,scale=0.5,fill=white]{};

\end{tikzpicture}
\end{center}
\caption{\label{fig:2} Principal null vector fields $l^a$ and of $n^a$ in the exterior of the Kerr black hole.}
\end{figure}
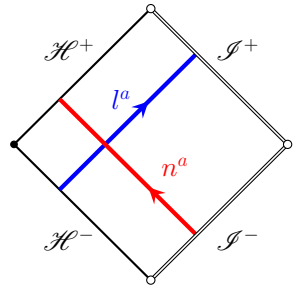

For calculations near $(\sH^+, \sI^-)$ respectively $(\sH^-,\sI^+)$, advanced respectively retarded Kerr-Newman (KN) coordinates are useful. They are defined by
\begin{equation}
\label{eq:KNdef}
(x^\mu) = (v=t+r_*(r), r, \theta, \varphi^*), \quad 
(x^\mu) = (u=t-r_*(r), r, \theta, \varphi_*),
\end{equation}
where
\begin{subequations}
\begin{align}
\dd \varphi^* = \ & \dd \varphi + a \frac{\dd r}{\Delta}, \\
\dd \varphi_* = \ & \dd \varphi - a \frac{\dd r}{\Delta},
\end{align}
\end{subequations}
and where 
\begin{equation}
\dd r_*=\frac{r^2+a^2}{\Delta}\dd r\, 
\end{equation}
defines the Kerr tortoise coordinate $r_*$.
In retarded KN coordinates $(x^\mu)=(u,r,\theta,\varphi_*)$, the Kinnersley tetrad \eqref{eq:Kintet}  is given by:
\begin{subequations}
\label{eq:Kintet up}
\begin{align}
l^a &= \left( \frac{\partial}{\partial r} \right)^a , \\
n^a &=\frac{1}{\Sigma} \left[ (r^2+a^2) \left( \frac{\partial}{\partial u} \right)^a -(\Delta/2)
\left( \frac{\partial}{\partial r} \right)^a + a\left( \frac{\partial}{\partial \varphi_*} \right)^a \right] ,\\
m^a &= \frac{1}{\sqrt{2} (r+ia\cos\theta)} \left[ ia\sin\theta \left( \frac{\partial}{\partial u} \right)^a +
    \left( \frac{\partial}{\partial \theta} \right)^a + i\csc\theta \left( \frac{\partial}{\partial \varphi_*} \right)^a \right].
\end{align}
\end{subequations}
The isometry group Kerr spacetime  includes a discrete part ${\mathbb Z}_2 \times {\mathbb Z}_2$ corresponding to 
the $t-\varphi$-reflection and the reflection across the equatorial plane $\theta = \pi/2$. If we apply the $t-\varphi$ reflection 
to the Kinnersley frame, we would get a frame that is regular instead at $\sI^-$ and $\sH^+$.


The existence of geometrically preferred directions at any point in the Kerr spacetime $\sM$ makes a formalism based on NP tetrads aligned with these directions generally preferable. 
In a tetrad formalism, general relativity closely resembles a gauge theory, see e.g., \cite{Waldbook}. In fact, in the ordinary tetrad formalism based on real, oriented and time oriented 
orthonormal frames, the collection of such frames over each spacetime point defines a so-called principal fibre bundle $\mathscr F$ over $\sM$, called the frame-bundle (see e.g., 
\cite{kobayashi}). The bundle $\mathscr F$ is acted upon by the proper orthochronous Lorentz group at every point, tensor fields are sections of suitable associated vector bundles, and the Levi-Civita covariant derivatives corresponds to a special gauge connection 1-form of $\mathscr F$. 

The Geroch-Held-Penrose (GHP) formalism~\cite{GHP} can be introduced from the perspective of fibre bundles \cite{Ehlers}: One considers a reduction $\sP$ of $\mathscr F$ to complex null (NP) frames aligned with the principal null directions and the time- and spacetime orientations\footnote{The GHP formalism can in fact be set up in generic spacetimes, i.e., with respect to any smooth distribution of pairs of null directions, not necessarily geometrically preferred.}. 
Correspondingly, the structure group $G$ of the principal fibre bundle $\sP$ is that subgroup of the proper orthochronous Lorentz group consisting of boosts preserving the pair of null directions, and rotations in their orthogonal complement. More precisely, the remaining gauge transformations consist of (a) a local boost, sending $l^a \to \lambda \bar{\lambda} l^a$, $n^a \to \lambda^{-1} \bar{\lambda}^{-1} n^a$, or (b) a local rotation $m^a \to \lambda \bar{\lambda}^{-1} m^a$, where $\lambda$ is a nonzero complex number depending on the spacetime point. In other words, the local gauge group can be identified with the multiplicative group $G = {\mathbb C}^\times$ of non-zero complex numbers.\footnote{The GHP formalism is {\it not} covariant under the null rotations which change the null directions $l^a$ or $n^a$.} In a sense, the GHP formalism is analogous to writing the equations of a spontaneously broken local gauge theory in terms of covariant derivatives, curvatures, etc. associated with the unbroken part of the gauge group (here the local gauge transformations (a) and (b) parameterized by $\lambda \in G$.)

Just as matter fields in ordinary gauge theories are classified by a representation of the gauge group that they transform in, and are sections in a corresponding associated vector bundle (see e.g., \cite{kobayashi}), so are the so-called GHP-scalars. The representations 
of the 1-dimensional gauge group $G = {\mathbb C}^\times$ in question are the 1-dimensional representations on the complex vector space $\mathbb C$ given by $\pi_{p,q}(\lambda) = \lambda^p \bar \lambda^q$, where $p,q$ are real numbers called the GHP weights. A GHP scalar $\eta$ of weights $(p,q)$, then, is a (smooth) section, i.e., $\eta \in \Gamma^\infty(\sM, \sL_{p,q})$, in the associated 1-dimensional complex vector bundle $\sL_{p,q} := \sP \ltimes_{\pi_{p,q}} {\mathbb C}$; we write this as 
\begin{equation}
\eta \circeq \GHPw{p}{q} :\Longleftrightarrow \eta \in \Gamma^\infty(\sM, \sL_{p,q}).
\end{equation} 
GHP weights are additive under multiplication---abstractly, we have $\sL_{p,q} \otimes \sL_{p',q'} = \sL_{p+p',q+q'}$.
From this it follows in particular that the dual bundle of complex linear maps is $(\sL_{p,q})' = \sL_{-p,-q}$, because the product of GHP scalars with opposite weights is an ordinary complex valued scalar field on $\sM$.

A specification of an NP frame $\{n^a, l^a, m^a\}$ aligned with the fixed null directions in some open neighborhood, is, according to the general theory of principal fibre bundles, a local section of $\sP$. Such a section corresponds to a trivialization of $\sP$, i.e., local identification of $\sP$ with $\sM \times \mathbb{C}^\times$. A trivialization allows us to locally identify any section $\eta \circeq \GHPw{p}{q}$ with an honest to God scalar, i.e., $\mathbb{C}$-valued function, called the ``gauge representative'' of $\eta$. If we change from one such frame to another, i.e., 
\begin{equation}
 \tilde l^a = \lambda \bar{\lambda} l^a, \quad 
 \tilde n^a = \lambda^{-1} \bar{\lambda}^{-1} n^a, \quad 
 \tilde m^a = \lambda \bar{\lambda}^{-1} m^a,  
\end{equation}
then the gauge representative changes to $\tilde \eta$, where 
\begin{equation}\label{GHPtrafo}
\tilde \eta = \lambda^p \bar{\lambda}^q \eta .
\end{equation}
In the following, we will, for ease of notation, not distinguish a GHP scalar and its local gauge representative, and the above transformation law, which also applies to any gauge-invariant equation between GHP scalars, will be understood implicitly. 

An easy way to produce GHP scalars from an ordinary tensor field is to contract its indices into the legs of an NP tetrad. The resulting scalar function, $\eta$, on $\sM$ can then be seen as the gauge representative, in the NP tetrad  utilized, of a corresponding 
section in the appropriate line bundle $\sL_{p,q}$. 
For example, one may consider the Weyl components,
\begin{subequations}\label{eq:Weylcomp}
\begin{align}
\Psi_0  =& -C_{abcd} \; l^a m^b l^c m^d \circeq \GHPw{4}{0}, \\
\Psi_1  =& - C_{abcd} l^a n^b l^c m^d \circeq \GHPw{2}{0}, \\
\Psi_2  =& -\tfrac{1}{2} C_{abcd}( l^a n^b l^c n^d + l^a n^bm^c \bar{m}^d) \circeq \GHPw{0}{0},\\
\Psi_3  =& -C_{abcd} l^a n^b \bar{m}^c n^d \circeq \GHPw{-2}{0}, \\
\Psi_4  =& -C_{abcd} \; n^a \bar{m}^b n^c \bar{m}^d \circeq \GHPw{-4}{0} ,
\end{align} 
\end{subequations}
which are gauge representatives of invariantly defined GHP scalars (i.e. sections of $\sL_{p,q}$) of the indicated weights. 
More non-trivially, the quantities 
\begin{subequations}\label{GHPscalars}
\begin{align}
\kappa &=m^a l^b \nabla_b l_a \circeq \GHPw{3}{1} ,\\
\tau &= m^a n^b \nabla_b l_a \circeq \GHPw{1}{-1},\\
\sigma &= m^a m^b \nabla_b l_a \circeq \GHPw{3}{-1}, \\
\rho &= m^a  \bar{m}^b \nabla_b l_a  \circeq \GHPw{1}{1},
\end{align}
\end{subequations}
are also GHP scalars of the indicated weights, i.e., they have the homogeneous scaling \eqref{GHPtrafo} despite the presence of the derivative operators in the definitions.
One may also consider the primed counterparts $\kappa',\tau',\sigma',\rho'$, where the GHP priming operation is in general defined by exchanging $l^a \leftrightarrow n^a, m^a \leftrightarrow \bar{m}^a$. If $\eta \circeq \GHPw{p}{q}$, then $\eta' \circeq \GHPw{-p}{-q}$. Furthermore, complex conjugation gives a GHP scalar of weights $\bar \eta \circeq \GHPw{q}{p}$, 
so it makes sense to talk about real GHP scalars when $p=q$, though not otherwise. The difference, $s=(p-q)/2$ is called the ``spin'', and so changes its sign under complex conjugation.

The GHP scalars \eqref{GHPscalars} are sometimes referred to as the ``optical scalars'' because they are directly related to the expansion, shear, and twist in case the null leg(s) of the NP tetrad are tangent to geodesics, as is automatically the case for our choice of NP tetrad aligned with the principal null directions $l^a$ or $n^a$ in Kerr, see table \ref{Opticalscalars}. %
\begin{table}[t]\label{Opticalscalars}
\renewcommand{\arraystretch}{1.5}
\begin{indented}
\item[]
\begin{tabular}{ c | c }
\br
\ GHP quantity \ & interpretation\\
\hline
$\kappa$ & \ failure of $l^a$ to be tangent to geodesic \ \\
${\rm Re} \; \rho$ & expansion of $l^a$\\
${\rm Im} \; \rho$ & twist of $l^a$\\
$\sigma $ & shear of $l^a$\\
$\tau$ & 
H\'aji\v{c}ek twist
\\
\br
\end{tabular}
\end{indented}
\caption{Optical scalars in the GHP formalsm; $\kappa=\kappa'=\sigma=\sigma'=0$ in Kerr. Note that the real and imaginary parts of $\rho$ are well-defined GHP quantities because the GHP weight $\GHPw{1}{1}$ of $\rho$ equals that of $\bar{\rho}$. The case of $n^a$ is analogous and would correspond to the GHP primed optical scalars.}
\end{table}
In our NP tetrad aligned with the principal null directions of Kerr, 
we have $\Psi_i = 0, i \neq 2$ and $\kappa=\sigma=\kappa'=\sigma'=0$. 
We refer e.g., to \cite{Price} for the values of the remaining GHP quantities in various frames.

For GHP calculus we 
require a GHP covariant derivative, denoted by $\Theta_a$, on the bundles $\sL_{p,q}$, i.e., its action on a GHP scalar $\eta$ with scaling weights 
$\eta \circeq \GHPw{p}{q}$ gives another such scalar\footnote{More precisely, $\Theta_a \eta$ is a section in $T^*M \otimes \sL_{p,q}$. The covariant derivative $\Theta_a$ should be understood to be a mapping from sections of $T^{k,l}\sM \otimes \sL_{p,q}$ to sections in 
$T^{k+1,l}\sM \otimes \sL_{p,q}$, where $T^{k,l}\sM$ are the tensor bundles of $k$-times covariant and $l$-times contravariant tensors.}, and it satisfies the usual axioms (see e.g., \cite{kobayashi}) of a covariant derivative. Such a derivative is defined by 
\begin{align}\label{nabladef}
\Theta_a \eta &= \left[\nabla_a - \half(p+q) n^b \nabla_a l_b + \half(p-q) \bar{m}^b \nabla_a m_b\right]\eta \nonumber \\
 & \equiv \left[\nabla_a + l_a(p\epsilon'+q\bar{\epsilon}') + n_a(-p\epsilon-q\bar{\epsilon}) - m_a(p \beta' -q \bar{\beta}) - \bar{m}_a(-p \beta +q \bar{\beta}')\right]\eta. 
\end{align}
The GHP covariant directional derivatives along the NP tetrad legs are denoted traditionally by
\begin{equation}
\thorn = l^a \Theta_a, \quad 
\thorn' = n^a \Theta_a, \quad 
\eth = m^a \Theta_a, \quad 
\eth' = \bar{m}^a \Theta_a . 
\end{equation}
Due to the presence of the NP tetrad legs $l^a, n^a, m^a, \bar m^a$ these operators shift the GHP weights by the amounts 
$\thorn: \GHPw{1}{1}, \thorn': \GHPw{-1}{-1}, \eth: \GHPw{1}{-1}, \eth': \GHPw{-1}{1}$.
In order to write the GHP covariant directional derivatives operators explicitly, one requires the remaining 
4 complex spin coefficients $\epsilon$, $\epsilon'$, $\beta$, $\beta'$. These 
may, in principle, be read off from \eqref{nabladef}, and are given by
\begin{subequations}\label{GHPscalars1}
\begin{align}
\beta &=\half ( m^a n^b \nabla_a l_b - m^a \bar{m}^b \nabla_a m_b)  ,\\
\epsilon &= \half ( l^a n^b \nabla_a l_b - l^a \bar{m}^b \nabla_a m_b) .
\end{align}
\end{subequations}
%
Contrary to the optical scalars, these remaining spin coefficients do not have definite GHP weight i.e., they do not transform as GHP scalars \eqref{GHPtrafo}. Instead, they form part of the definition of the GHP covariant derivative $\Theta_a$ \eqref{nabladef} and consequently of $\thorn$, $\thorn'$, $\eth$, $\eth'$. The spin coefficients $\epsilon$, $\epsilon'$, $\beta$, $\beta'$ may never appear explicitly in any GHP covariant equation nor in any GHP covariant calculation. The full list of
commutators between the operators $\thorn$, $\thorn'$, $\eth$, $\eth'$ as well as their action on the non-zero optical scalars 
$\rho$, $\rho'$, $\tau$, $\tau'$ in Kerr, as used extensively throughout this paper, is recalled in \ref{sec:GHPformulas}. 

In type D spacetimes such as Kerr, the GHP formalism based on a pair $l^a, n^a$ of repeated principal null directions is symmetric under the GHP priming operation (i.e., $n^a \leftrightarrow l^a, m^a \leftrightarrow \bar m^a$), in the sense that if some GHP covariant geometric equation holds, then so does its primed version. Thus e.g., $\sigma = 0$ automatically implies $\sigma'=0$ because $\sigma$ is a GHP scalar.

\subsection{Teukolsky formalism and metric reconstruction}

Teukolsky's approach \cite{Teukolsky:1973ha,Teukolsky:1972my} to perturbation theory on Kerr may be elegantly recast into an operator identity \cite{Wald}, which intertwines the 
linear operator, $\E$, acting on $h_{ab}$ in the linearized EE, with Teukolsky's operator, $\O$, acting on the Weyl scalars $\psi_0$ or $\psi_4$---the linearizations of $\Psi_0$ or $\Psi_4$, see \eqref{eq:Weylcomp}---around Kerr. This ``SEOT' identity is 
\begin{equation}
\label{eq:SEOT}
\S \E = \O \T.
\end{equation}
$\E$, the linearized Einstein operator, is given concretely by
\begin{align}\label{eq:linearE}
\delta G_{ab} = (\mathcal{E} h)_{ab} \equiv \frac{1}{2} \big[ & -\nabla^c\nabla_c h_{ab} - \nabla_a\nabla_bh + 2 \nabla^c\nabla_{(a} h_{b)c} \nonumber\\
 & + g_{ab}(\nabla^c \nabla_c h - \nabla^c\nabla^d h_{cd}) \big].
\end{align}
The intertwining operator $\T$ prepares $\psi_0 \equiv \delta \Psi_0$ 
from $h_{ab}$. Concretely\footnote{This, and the following tensorial forms for $\O, \S$
were given by \cite{Araneda,Aksteiner:2014zyp,Aksteiner:2016pjt}.}, 
\begin{equation}
\delta \Psi_0 = \mathcal{T}h \equiv - \frac{1}{2}Z^{bcda} \nabla_a \nabla_b h_{cd},
\end{equation}
where $Z^{abcd} = Z^{ab}Z^{cd}, Z = l \wedge m$.
$\O$ is Teukolsky's wave operator, expressed in GHP covariant form as
\begin{equation}
\mathcal{O} \eta \equiv \left[ g^{ab}(\Theta_a + 4B_a)(\Theta_b + 4 B_b) - 16 \Psi_2 \right]\eta , 
\end{equation}
for any $\eta \circeq \GHPw{4}{0}$, 
where $B_a = -\rho n_a + \tau \bar{m}_a \circeq \GHPw{0}{0}$. 
Finally, the intertwining operator $\S$ acts on symmetric covariant rank two tensors $t_{ab}$ and is given by 
\begin{equation}
\mathcal{S}t \equiv \frac{1}{4} Z^{bcda} \zeta^{-4} \nabla_a (\zeta^4 \nabla_b t_{cd}),
\end{equation}
where $\zeta=\Psi_2^{-\frac{1}{3}}$ in Kerr.
$\S$ is in fact the operator that prepares the Teukolsky source, ${}_{+2}T=\S^{ab}T_{ab}$, from the stress tensor $T_{ab}$: If 
$h_{ab}$ is a metric perturbation satisfying the linearized EE $(\E h)_{ab} = T_{ab}$, such that $\psi_0 = \mathcal{T}^{ab}h_{ab}$ is the corresponding linearized 
bottom Weyl scalar, then the SEOT identity  gives precisely
\begin{equation}
\label{eq:Teukpsi0}
\O \psi_0 = {}_{+2} T.
\end{equation}
The coordinate form of Teukolsky's equation appearing in \cite{Teukolsky:1973ha,Teukolsky:1972my} is recovered by going to 
BL coordinates and the Kinnersley frame, and substituting the corresponding expressions for the optical scalars and GHP 
covariant derivatives in the GHP formalism. For reader's convenience, we summarize the well-known alternative  forms of 
the operators $\E$, $\S$, $\O$ and $\T$ in terms of the optical scalars and directional derivatives in the GHP formalism in
\ref{app:LinEinGHP}, \ref{app:ST}, \ref{sec:appO}. Applying the GHP priming operation to \eqref{eq:SEOT}, we get analogous equations involving the GHP prime of Teukolsky's operator, $\O'$, the top perturbed Weyl scalar $\psi_4 = \T^{\prime ab} h_{ab}$, and the spin $-2$ Teukolsky source $_{-2} T = \S^{\prime ab}T_{ab}$, i.e.,
\begin{equation}
\label{eq:SEOTprime}
\S' \E = \O' \T',
\end{equation}
using that $\E' = \E$.

Following the interpretation by \cite{Wald}, one may arrive at the metric reconstruction procedure by \cite{Chrzanowski:1975wv,Kegeles:1979an}
via an interesting change of perspective on the SEOT identity \eqref{eq:SEOT}. This change of perspective is brought about by simply applying the 
operation of (formal) adjoint to \eqref{eq:SEOT}. Namely, using that $\E^\dagger = \E$, we have
\begin{equation}
\label{eq:SEOT*}
\E \S^\dagger =  \T^\dagger \O^\dagger, 
\end{equation}
where for completeness, we state the formal definition of the adjoint operation $\dagger$:

\begin{definition}
\label{def:adj}
Let $\A$ be a partial differential operator from sections of a $\mathbb{C}$-vector bundle $\sV$ to a $\mathbb{C}$-vector bundle $\sW$. 
Let $\sW'$ and $\sV'$ be the 
dual bundles of $\mathbb{C}$-linear maps. Then $\A^\dagger$ is the unique partial differential operator from sections in $\sW'$ to sections in $\sV'$ defined by
\begin{equation}
w' (\A v) - (\A^\dagger w')v = \nabla^a f_a,
\end{equation}
where $w'$ is a section of $\sW'$, $v$ a section of $\sV$ and $f_a$ is a $\mathbb{C}$-valued 1-form locally constructed out of the derivatives 
of $v,w'$ and the coefficients of $\A$. Note that the adjoint operation is $\mathbb{C}$-linear, and that it depends on a choice of divergence operator, provided in our case by the metric $g_{ab}$.
\end{definition}

In our application involving $\T, \S, \E, \O$, the vector bundles in question are suitable tensor products of $\sL_{\pm 4,0}, T\sM, T^*\sM$. Consider now a 
GHP scalar $\Phi \circeq \GHPw{-4}{0}$ such that $\O^\dagger \Phi = 0$, i.e., it is a solution to the homogeneous adjoint Teukolsky equation. Applying 
\eqref{eq:SEOT*} to $\Phi$, we find that $(\S^\dagger \Phi)_{ab}$ is in the kernel of $\E$. Furthermore, $\E$ clearly is a real operator, so we have 
\begin{equation}
\E \Re (\S^\dagger \Phi)_{ab} = 0.
\end{equation}
Thus, one has obtained a vacuum solution, $h_{ab} = \Re(\S^\dagger \Phi)_{ab}$, from a ``Hertz potential'', $\Phi$ solving the adjoint Teukolsky equation \cite{Chrzanowski:1975wv,Kegeles:1979an}. The well-known 
explicit GHP-forms of $\O^\dagger, \S^\dagger, \T^\dagger$ are recalled for convenience in \ref{app:ST}, \ref{sec:appO}. Using such formulas, 
we can express the NP components of the reconstructed metric perturbation $h_{ab} = \Re(\S^\dagger \Phi)_{ab}$ in terms 
of the GHP directional derivative operators $\thorn,\eth,$ and the
optical scalars $\rho,\tau,\tau'$ as follows
\begin{subequations}
\begin{align}
\label{irg}
0&= h_{ll} = h_{ln} = h_{lm} = h_{m\bar m},\\ 
      \label{eq:hmbmb1}
  h_{\bar m \bar m } &=  -\half(\thorn - \rho)(\thorn + 3\rho) \Phi, \\
  \label{eq:hmbn1}
  h_{\bar m n} &=  -\frac{1}{4} \Big[ (\thorn - \rho + \bar{\rho})(\eth + 3\tau) + (\eth-\tau+\bar{\tau}')(\thorn +3\rho) \Big] \Phi, \\
  \label{eq:hnn1}
  h_{nn} &= -\half (\eth-\tau)(\eth+3\tau) \Phi + \textrm{c.c.}\quad ,
\end{align}
\end{subequations}
where ``c.c.'' means ``complex conjugate''. The conditions \eqref{irg} state that $h_{ab}$ automatically is 
in traceless ingoing radiation gauge, TIRG. Alternatively, in tensorial notation, we have \cite{Araneda}
\begin{equation}
\label{hreconstr}
h_{ab} = \Re(\S^\dagger \Phi)_{ab} = \frac{1}{2}  \Re \left\{\nabla^c \left[ \zeta^4 \nabla^d (\zeta^{-4}Z_{d(ab)c} \Phi) \right] \right\}.
\end{equation}
Intuitively speaking, 
a gravitational perturbation, i.e. solution to the linearized EE, has two degrees of freedom per point, which is equal to the number of 
real components of the complex GHP scalar $\Phi$. Thus, it would appear not totally implausible that, up to possibly a pure gauge perturbation, and possibly up to a 
finite number of modes representing various global obstructions, any sufficiently regular solution to the linearized {\it homogeneous} EE be 
be representable in reconstructed form \eqref{hreconstr} [or equivalently \eqref{irg}, \eqref{eq:hmbmb1}, \eqref{eq:hmbn1}, \eqref{eq:hnn1}], 
for some $\Phi \circeq \GHPw{-4}{0}$ solving the homogeneous adjoint Teukolsky equation $\O^\dagger \Phi = 0$.

It can be seen 
by a conceptual argument \cite{Prabhu:2018jvy} that a non-trivial ``zero mode'' $\dot g_{ab}$ (see \ref{app:D}) i.e., a linear perturbation in the 
Kerr family parameterized by $(M,a)$, definitely cannot be written in reconstructed form up to a pure gauge perturbation. 
Since a zero mode is in particular regular, including at $\sH^\pm, i^\pm, i^0$, and asymptotically flat 
at $\sI^\pm$, the best that we can hope for, even within this class of solutions, is  
a decomposition of the form $h_{ab} = (\S^\dagger \Phi)_{ab} + (\mathcal{L}_X g)_{ab} + \dot g_{ab}$. 

This statement is, in effect, widely accepted by practitioners using metric reconstruction. Evidence comes from the fact that 
it is true for mode solutions \cite{Ori:2002uv}. Furthermore \cite{Prabhu:2018jvy,Green:2019nam} have shown that solutions of the 
form $h_{ab} = (\S^\dagger \Phi)_{ab} + (\mathcal{L}_X g)_{ab} + \dot g_{ab}$ are in a sense dense in the set of all smooth, asymptotically flat 
solutions with respect to a---rather weak---norm on the space of initial data for the homogeneous linearized EE. However, 
these arguments fall considerably short of a general decomposition theorem for smooth, decaying, asymptotically flat (in the Bondi sense) solutions to the homogeneous EE. In this paper, we shall obtain such a result as corollary \ref{GHZcor} (see Sec. \ref{ch:GHZ}) of our much more general GHZ decomposition theorems for asymptotically flat solutions to the 
{\it inhomogeneous} linearized EE. 

\subsection{Held's version of GHP}

Held's version \cite{Held} of the GHP formalism is extremely convenient when considering asymptotic expansions near $\sI^+$, i.e., in powers of $1/r$ in retarded KN coordinates, in a GHP covariant setup. For this, it turns out to be most practical to work with the GHP scalar  $\rho$ instead of $1/r$. 
$\rho$ is complex and not a coordinate, but the relation
\begin{equation}
\label{eq:Heldrho}
\frac{1}{\rho} = - r + ia \cos \theta
\quad \text{(Kinnersley frame)}
\end{equation}
shows that $1/r = O(\rho)$ near $\sI^+$, so asymptotic expansions in $1/r$ can readily be converted to ones in $\rho$.
$\rho$ is superior, however, for GHP calculus because it has a useful relationship 
with the GHP covariant directional derivative $\thorn$ along $l^a$,
\begin{equation}
\thorn \frac{1}{\rho} = -1. 
\end{equation}
Thus, $\thorn$ is formally like a partial derivative ``$\thorn = \partial_{-1/\rho}$''. This formal relation is the essence of the Held formalism \cite{Held}. Following his conventions, we use the circle notation as in $x^\circ$ to mean that (i) $x^\circ$ is a GHP scalar of some definite weights and (ii) $\thorn x^\circ = 0$. 
The operators $\thorn', \eth, \eth'$ have the disadvantage that, when acting on a GHP scalar $x^\circ$ annihilated by $\thorn$, they do not produce another such quantity.
This is accomplished by Held's operators \cite{Held} (acting on weight $\GHPw{p}{q}$ GHP scalars),
\begin{subequations}\label{eq:Hops}
\begin{align}
\tilde{\thorn}' &=\thorn' - \bar{\tau} \tilde{\eth} - \tau \tilde{\eth}' + \tau \bar{\tau} \left( \frac{q}{\rho} + \frac{p}{\bar{\rho}} \right) + \half \left( \frac{q \bar{\Psi}_2}{\bar{\rho}} + \frac{p \Psi_2}{\rho} \right), \\
\tilde{\eth} &=\frac{1}{\bar{\rho}} \eth + \frac{q\tau}{\rho}, \\
\tilde{\eth}' &= \frac{1}{\rho} \eth' + \frac{p\bar{\tau}}{\bar{\rho}}.
\end{align}
\end{subequations}
To automate computations with these operators, one requires their commutators, expansions of the background GHP scalars
$\rho', \tau, \tau',\Psi_2$ in terms of $\rho, \bar \rho$, and certain expansion coefficients $\tau^\circ, \rho^{\prime \circ}, \Omega^\circ, \Psi^\circ$, 
as well as formulas for the action of the Held operators  \eqref{eq:Hops} on these. These identities are recalled in \ref{sec:Held} for convenience.

\subsection{Mode solutions of the Teukolsky equations} 
\label{sec:modesoln}
The Teukolsky equations $\O\psi_0 = 0$ ($s=2$) and $\O^{\prime} \psi_4 = 0$ ($s=-2$) are separable in BL coordinates and the Kinnersley frame \cite{Teukolsky:1972my}, so we have mode solutions of the form\footnote{In fact, the Teukolsky equation is separable for all spins
$s=0,1/2,1,3/2,2$. The equations corresponding to the adjoints of $\O$ and $\O'$ are also separable observing that
\begin{subequations}\label{OprimeOdag}
\begin{align}
\mathcal{O}^{\prime \dagger} =& \; \zeta^{-2s} \; \mathcal{O} \; 
\zeta^{2s} ,\\
\mathcal{O}^{\prime} =& \;  \zeta^{2s} \; \mathcal{O}^\dagger \; 
\zeta^{-2s}.
\end{align}
\end{subequations}
}
\begin{equation}
\label{modes}
{}_s \Psi_{\omega \ell m}(t,r,\theta,\varphi) = {}_s R_{\omega \ell m}(r) \, {}_s S_{\ell m}(\theta,\varphi;a\omega) e^{-i\omega t}. 
\end{equation}
Below, we will sometimes drop for simplicity the mode labels $\omega,\ell,m$ as e.g., in ${}_s R_{\omega \ell m} \equiv {}_s R$, which are implicitly understood in such cases. These functions obey the ``radial Teukolsky equation'' \cite{Teukolsky:1973ha}, which is not required in this paper.
The spin-weighted spheroidal harmonics
${}_s S_{\ell m}(\theta,\varphi;a\omega)$ have a harmonic dependence $e^{im\varphi}$ on $\varphi$, and
satisfy the ``angular Teukolsky equation'' \cite{Teukolsky:1973ha} 
\begin{equation}
\left[ \frac{1}{\sin \theta} \partial_\theta \big( \sin \theta \partial_\theta \big) + (a\omega \cos \theta)^2 -2sa\omega\cos\theta -\frac{(m+s\cos\theta)^2}{\sin^2 \theta} + {}_s E - s^2 \right] {}_s S = 0.
\end{equation}
For $\omega \in \mathbb{R}$ the possible angular eigenvalues of ${}_s E$ for given $m,\omega$ are determined by demanding that 
\eqref{modes} define a smooth GHP scalar, i.e., a smooth section in the bundle $\sL_{2s,0}$. This condition is 
equivalent to requiring that the limits of ${}_s S(\theta,\varphi;a\omega)$ as $\theta \to 0,\pi$ exist. The possible angular eigenvalues ${}_s E_{\ell m}(a\omega)$
for given $m,\omega$ are labelled by $\ell$. Traditionally, the labelling is chosen such that $m \in \{-\ell, \dots, \ell\}, \ell \in \{|s|, |s|+1, \dots\}$. 

The angular operator is a perturbation, uniformly bounded in the sense of $L^2([0,\pi])$, of the 
corresponding operator with $a\omega = 0$ (Schwarzschild), so standard results in perturbation 
theory (see e.g., \cite[Ch. 7]{Kato}) apply, showing that the angular eigenvalues ${}_s E_{\ell m}(a\omega)$
go to those for spin-weighted spherical harmonics  (see e.g., \cite{Chandrasekhar:1984siy}) when $a\omega \to 0$. 
The labelling of the angular eigenvalues is in fact such that
\begin{equation}
{}_s E_{\ell m}(a\omega) = \ell(\ell+1) + O(a\omega),
\end{equation}
uniformly in $m,\ell$. See e.g., \cite[App. B]{Kavanagh:2016idg} for a more systematic development of the perturbation theory for small $a\omega$.
On the other hand, for large $a\omega \to \pm \infty$ but fixed $m,\ell$, the angular eigenvalues have an asymptotic expansion of the form \cite{Casals:2004zq,Casals:2018cgx}
\begin{subequations}
\begin{align}
&{}_s E_{\ell m}(a\omega ) = -(a\omega)^2 +2a\omega\, {}_s q_{\ell m} +o_{\ell,m}(a\omega), \\
&{}_s q_{\ell m}=
\begin{cases}
\ell & \text{$\ell$ odd, $\ell \ge |m+s|+s$,}\\
\ell +1 & \text{$\ell+1$ even, $\ell \ge |m+s|+s$,}\\
2\ell+1 - (|m+s|+s) & \text{$\ell < |m+s|+s$.}
\end{cases}
\end{align}
\end{subequations}
We may and will chose the spin-weighted spheroidal harmonics to be real (for $\omega \in \mathbb{R}$) 
in the sense that $\overline{{}_s S_{\ell m}(\theta, \varphi; a\omega)} = {}_s S_{\ell m}(\theta, -\varphi; a\omega)$.
For the purposes of this paper, mostly $s=\pm 2$ will be relevant. $s=\pm 1, 0$ are required for a frequency domain analysis of the auxiliary Teukolsky equations appearing 
in Sec. \ref{sec:PTL}. 

The angular eigenvalues and eigenfunctions have analytic continuations to multi-valued functions on the complex $a\omega$-plane, see e.g., 
\cite[Prop. 2.1]{TdC20} for a summary of various classical results in this direction and \cite{Finster:2015xma} for a detailed study of the analytic continuation of 
the spin-weighted spheroidal functions to complex $a\omega$.  Below we require the following lemma.
\begin{lemma}
\label{lem:u}
For all $\omega \in \mathbb{R}$ we have 
\begin{equation}
\left| 
\frac{\dd{}_s E_{\ell m} (a\omega)}{\dd \omega} \right| = O(1+|a\omega|+ |m|),
\end{equation}
uniformly in $\ell$.
\end{lemma}
\begin{proof}
The analytic continuations of the angular eigenvalues have no branch cuts or poles in some open neighborhood of the real axis \cite[Prop. 2.1.3]{TdC20}, and so are in particular 
differentiable on the real $a\omega$-axis. The first derivative may be computed using standard perturbation theory. In fact, 
by writing the angular equation as a 1-dimensional Schr\" odinger equation on $L^2([0,\pi]_\theta)$ as e.g., in \cite[Eqs. 2.7,2.8]{Finster:2015xma}, the first order 
perturbation off of the real $a\omega$-axis is seen to be a bounded operator on $L^2([0,\pi]_\theta)$ with operator norm of order $O(1+|a\omega|+ |m|)$, from which the result follows.
\end{proof}
Both the radial and angular functions ${}_s R, {}_s S$ obey so-called Teukolsky-Starobinski (TS) identities \cite{Teukolsky1974,Starobinsky1974}, see e.g., \cite{Chandrasekhar:1984siy,Casals:2020fsb} for further detail\footnote{The angular TS identities for mode functions may be obtained from the covariant form 
\eqref{TSonshell1}.}. The (angular) TS identities that we shall require are ($\omega \in \mathbb{R}$):
\begin{subequations}
\label{S:TS}
\begin{align}
\left(
\prod_{j=0}^{3} {}_{2-j} \mathcal{L}^+(a\omega)
\right) \, {}_{+2} S_{\ell m}(\theta, \varphi; a\omega) =& \, {}_2 B_{\ell m}(a\omega) \, {}_{-2} S_{\ell m}(\theta, \varphi; a\omega),\\
\left(
\prod_{j=0}^{3} {}_{j-2}  \mathcal{L}^-(a\omega)
\right) \, {}_{-2} S_{\ell m}(\theta, \varphi; a\omega) =& \, {}_2 B_{\ell m}(a\omega) \, {}_{+2} S_{\ell m}(\theta, \varphi; a\omega),
\end{align}
\end{subequations}
where 
\begin{equation}
{}_s \mathcal{L}^\pm(a\omega) = -\frac{1}{\sqrt{2}}\bigg[ \partial_\theta \mp (i \csc \theta \partial_\varphi + a\omega \sin \theta) \pm s \cot \theta \bigg],
\end{equation}
and where for $\omega \in \mathbb{R}$, ${}_2B_{\ell m}(a\omega)$ is the non-negative real constant such that 
\begin{equation}
{}_2 B^2 = \frac{1}{16} {}_{-2} \lambda^2({}_{-2} \lambda +2)^2 - \frac{1}{2}(a\omega) {}_{-2} \lambda \Big[(a\omega-m)(5{}_{-2} \lambda +6) -12 \Big] + 9(a\omega)^2(a\omega -m)^2.
\end{equation}
Here 
${}_s \lambda = {}_s E - s(s+1) -2ma\omega+(a\omega)^2$ is the separation constant in the radial Teukolsky equation.
We shall in particular require the value ${}_2 B^2_{\ell m}(a\omega = 0) = (\ell-1)\ell(\ell+1)(\ell+2)/4$ for $a\omega = 0$.
As a consequence of the  discrete group ${\mathbb Z}_2 \times {\mathbb Z}_2$ of isometries of Kerr, the spin-weighted spheroidal harmonics have the symmetries
\begin{subequations}
\label{S:sym}
\begin{align}
{}_s S_{\ell m}(\theta,\varphi; a\omega) =& \ (-1)^{s+m} \, {}_{-s} S_{\ell m}(\theta,-\varphi;-a\omega),\\
{}_s S_{\ell m}(\theta,\varphi; a\omega) =& \ (-1)^{s+\ell} \, {}_{s} S_{\ell m}(\pi-\theta,-\varphi;-a\omega),
\end{align}
\end{subequations}
where the choice of signs ensures consistency with the TS identities.  

\section{The GHZ method}
\label{ch:GHZ}

\subsection{Proof of the GHZ decomposition: boundary conditions at $\sI^+$}
\label{GHZproof}

\subsubsection{Backward integration: Global decay assumptions at $\sI^+$ and zero mode $\dot g_{ab}$}
\label{sec:gdot}

In this section, we will give a proof of the GHZ decomposition in the exterior region of Kerr by imposing ``final'' conditions for the various transport equations along the orbits of $l^a$ at $\sI^+$. In order to be able to do so, we must require that the fields under consideration, namely the metric perturbation $h_{ab}$
and the stress tensor $T_{ab}$ in the linearized EE $(\E h)_{ab}=T_{ab}$, are sufficiently regular at $\sI^+$. 
In order to not further obstruct the already complicated construction by details of functional analytic nature, we will consider only smooth $h_{ab}$ and $T_{ab}$. Furthermore, we since we want to integrate our transport equations backwards from $\sI^+$, must clearly impose some decay conditions at $\sI^+$. 

For the metric perturbation, we will impose the linearized version of the usual Bondi-type asymptotic flatness conditions at $\sI^+$ (see e.g., \cite{Waldbook}). 
As we will explain more fully in  \ref{BondiIRG}, this implies that the NP components of $h_{ab}$ are required to have, near $\sI^+$, asymptotic 
expansions of the general form 
\begin{equation}
\label{asympt_exp}
    X = \sum_{j = j_0}^N X^{j \circ} \rho^j + O(\rho^{N+1}),
\end{equation}
where $N$ is as large as we like and $X$ is some NP component of $h_{ab}$. $X^{j \circ}$ are GHP scalars annihilated by $\thorn$ (so they may be 
considered as functions of the retarded KN coordinates $(u,\varphi_*,\theta) \in \mathbb{R} \times \mathbb{S}^2$ in the Kinnersley frame, see \ref{sec:Held}).
$j_0$ depends on the NP component considered and is the decay order, because in the Kinnersley frame, $\rho = O\left(r^{-1}\right)$ as $r \to \infty$ 
at fixed $(u,\varphi_*,\theta)$.  


\medskip
\noindent
{\bf Standing Decay Assumptions on $h_{ab}$:} The NP components of $h_{ab}$ have asymptotic expansions \eqref{asympt_exp}
with decay orders 
\begin{equation}
\label{eq:hdec}
\begin{split}
    h_{ll}, h_{lm}, h_{ln}, h_{m\bar{m}} &= O\left(\frac{1}{r^2}\right), \\
h_{nn}, h_{nm}, h_{mm} &= O\left(\frac{1}{r}\right),
\end{split}
\end{equation}
as $r \to \infty$ at fixed $(u,\varphi_*,\theta)$ (in the Kinnersley frame).
There exist $\delta M, \delta a$ such that, subtracting a corresponding zero mode (see \ref{app:D})
\begin{equation}
\label{zerosubtract}
    h_{ab} \to h_{ab} - \dot g_{ab},
\end{equation}
the leading expansion coefficients in the asymptotic expansion \eqref{asympt_exp}  all go to zero as $u \to -\infty$.

\medskip
\noindent
{\bf Remarks 1.}
\begin{enumerate}
\item
One can prove that $\delta M, \delta a$ correspond to the perturbed mass and specific angular momentum of $h_{ab}$, 
as calculated e.g., via its Cauchy data on some constant $t$ slice using the standard ADM formulas \cite{Waldbook}, or via Abbott-Deser integrals \cite{Abbott:1981ff} from $T_{ab}$.
\item
Note that we do not, at this stage, assume any particular behavior of the asymptotic expansion coefficients, let alone decay, for $u \to +\infty$. 
However, to complete the proof of the GHZ decomposition by the backward integration scheme, see theorem \ref{thm:2}, we will  require such a condition. 
\end{enumerate}

\medskip
From now, we assume that the zero mode $\dot g_{ab}$ has been subtracted as in \eqref{zerosubtract}. In the end of our construction, we have to remember to 
add $\dot g_{ab}$ back in. We also make corresponding assumptions about the NP components of $T_{ab}$  near $\sI^+$:

\medskip
\noindent
{\bf Standing Decay Assumptions on $T_{ab}$:} The NP components of $T_{ab}$ have asymptotic expansions \eqref{asympt_exp} with decay orders 
\begin{equation}\label{Tdec}
\begin{split}
T_{nn} =&  \ O\left(\frac{1}{r^2}\right) , \\
T_{ln}, T_{m\bar m}, T_{nm} =& \ O\left(\frac{1}{r^3}\right), \\
T_{lm}, T_{mm} =& \ O\left(\frac{1}{r^4}\right), \\
T_{ll} =& \ O\left(\frac{1}{r^5}\right),
\end{split}
\end{equation}
as $r \to \infty$, at fixed $(u,\varphi_*,\theta)$ (in the Kinnersley frame). The coefficients in the asymptotic expansion \eqref{asympt_exp} of the NP components of $T_{ab}$ go to zero as $u \to -\infty$.

\medskip
\noindent
{\bf Remark 2.}
These assumptions on $T_{ab}$ are consistent with, but not implied by, the Standing Decay Assumptions on $h_{ab}$, 
see \ref{BondiIRG} for a detailed analysis of the relations between the asymptotic expansions for the metric, the stress tensor, and the 
observables for gravitational radiation. Our Standing Decay 
Assumptions on $T_{ab}$ are sufficiently weak so as to allow fluxes of matter stress energy, 
i.e., a non-trivial $T_{nn}$ at order $O(1/r^2)$, or matter angular momentum, i.e., a nontrivial $T_{nm}$ at order $O(1/r^3)$.

\subsubsection{Backward integration for gauge vector field $\xi^a$}
\label{sec:gaugexi}

It is fairly easy to see that, without the use of any EE, a metric perturbation $h_{ab}$ may be put in IRG, i.e., $h_{ab} l^b=0$ \cite{Green:2019nam}. Indeed, if $h_{ab}$ is not in this gauge to begin with, we can solve for a gauge vector field $\xi^a$ satisfying 
\begin{equation}
[h_{ab} - (\lie{\xi} g)_{ab}] l^a = 0.
\end{equation}
Using formulas for the NP components of the Lie derivative recalled in \ref{sec:residualg}, this equation is seen to be equivalent to the GHP equations,
\begin{subequations}\label{eq:xidef}
\begin{align}
\label{a}
2 \thorn \xi_l =& \ h_{ll},\\
\label{b}
(\thorn + \bar{\rho}) \xi_m + (\eth + \bar{\tau}') \xi_l =& \ h_{lm},\\
\label{c}
\thorn \xi_n + \thorn' \xi_l + (\tau + \bar{\tau}') \xi_{\bar{m}} + (\bar{\tau} + \tau') \xi_m =& \ h_{ln} . 
\end{align}
\end{subequations}
Since $\thorn = \partial_r$ in retarded KN coordinates and the Kinnersley frame, this is a triangular system of transport equations for the NP components of $\xi_a$:
We first integrate \eqref{a} to obtain $\xi_l$,  which is substituted into \eqref{b}, and thereby gives a transport equation for $\xi_m$. Integrating that equation and substituting both solutions $\xi_l, \xi_m$ into \eqref{c} gives yet another transport equation for $\xi_n$. Having determined the NP tetrad components $(\xi_l, \xi_n, \xi_m, \xi_{\bar{m}})$ of $\xi_a$ we redefine $h_{ab} \to h_{ab} - ({\lie{\xi} g})_{ab}$, thereby putting the perturbation into IRG, $h_{ll}, h_{ln}, h_{lm} = 0$. 

The solution to the transport equations for $(\xi_l, \xi_n, \xi_m)$ is not unique since the system 
\eqref{eq:xidef} has a non-trivial kernel. Our Standing Decay Assumptions on $h_{ab}$ can be used to show that 
we can, and will, pick a solution such that, near $\sI^+$, 
\begin{equation}
\xi_l, \xi_n, \xi_m = O(\rho),
\end{equation}
and it is easy to see that this requirement renders the solution unique. The formulas for the Lie-derivative in GHP form in \ref{sec:residualg}, and the formulas 
for the GHP operators and optical scalars given e.g., in \ref{sec:Held},
show that, after the redefinition $h_{ab} \to h_{ab} - ({\lie{\xi} g})_{ab}$, we still have $h_{nn}, h_{nm}, h_{mm} = O(\rho), h_{m\bar{m}} = O(\rho^2)$, 
and additionally we have now the IRG conditions. In the end, we have to remember to add $({\lie{\xi} g})_{ab}$ back to $h_{ab}$.

\subsubsection{Backward integration for corrector $x_{ab}$}
\label{sec:xabdef}

As described already in the outline, sec. \ref{sec:GHZ decomposition}, the corrector tensor field $x_{ab}$ is next determined in such a way that
\begin{equation}\label{eq:x}
[T_{ab} - (\E x)_{ab}] l^b = 0, 
\end{equation}
with the idea to eliminate any $l$ NP component from $T_{ab}$. Our ansatz for $x_{ab}$ was \eqref{eq:xdef},
%
%
where $x_{n\bar{m}} = \bar{x}_{nm}$, so there are 4 real independent components encoded in $x_{m\bar{m}}, x_{mn}, x_{nn}$ by which we attempt to satisfy the 4 real independent equations \eqref{eq:x}. We first transvect \eqref{eq:x} with $l^a$ and use the $ll$ NP component of the Einstein operator $\E$, \eqref{eqn:Ell_app}, to obtain
\eqref{eq:xmmb}.
Next, we transvect \eqref{eq:x} with $m^a$ and use the $ml$ NP component of the Einstein operator $\E$, \eqref{eqn:Elm_app} and obtain:
\begin{equation}\label{eq:xnm}
\begin{split}
&\frac{\rho}{2(\rho+\bar\rho)} \thorn
\left\{ (\rho+\bar{\rho})^2
\thorn \left[ \frac{1}{\rho(\rho+\bar\rho)} x_{nm} \right] \right\}\\
=&\ T_{lm} - \half\bigg[(\thorn+\rho-\bar{\rho})(\eth+\bar{\tau}'-\tau) + 2\bar{\tau}'(\thorn-2\rho) -
(\eth-\tau-\bar{\tau}')\bar{\rho} +2\rho\tau\bigg]x_{m\bar{m}}.
\end{split}
\end{equation}
Finally, we transvect \eqref{eq:x} with $n^a$ and use the $nl$ NP component of the Einstein operator $\E$, \eqref{eqn:Eln_app} and obtain:
\begin{equation}
\label{eq:xnn}
\begin{split}
 &\half (\rho+\bar{\rho})^2\thorn \left( \frac{1}{\rho+\bar\rho}x_{nn} \right)\\
 =
 & \ T_{ln}- \half\bigg[(\eth'+\tau'-\bar{\tau})(\eth-\tau+\bar{\tau}') +
(\eth'\eth-\tau\tau'-\bar{\tau}\bar{\tau}'+\tau\bar{\tau}) - (\Psi_2+\bar{\Psi}_2)\\
&\pheq +(\thorn'-2\rho')\bar{\rho} + (\thorn-2\bar{\rho})\rho' +\rho(3\thorn'-2\bar{\rho}')
+\bar{\rho}'(3\thorn-2\rho)\\
&\pheq -2\thorn'\thorn + 2\rho\bar{\rho}' +2\eth'(\tau)-\tau\bar{\tau}\bigg]x_{m\bar{m}}
 \\
&\pheq - \half\bigg[(\thorn-2\rho)(\eth'-\bar{\tau}) + (\tau'+\bar{\tau})(\thorn+\bar{\rho})
-2(\eth'-\tau')\rho-2\bar{\tau}\thorn \bigg]x_{nm}\\
&\pheq - \half\bigg[(\thorn-2\bar{\rho})(\eth-\tau) + (\bar{\tau}'+\tau)(\thorn+\rho)
-2(\eth-\bar{\tau}')\bar{\rho}-2\tau\thorn \bigg]x_{n\bar{m}}.\\
\end{split}
\end{equation}
The system of equations \eqref{eq:xmmb}, \eqref{eq:xnm}, \eqref{eq:xnn}, given first in \cite{Green:2019nam}, is now solved integrating backwards from $\sI^+$ in a similar way as 
we did to obtain the NP components of $\xi^a$, noting again that e.g., in the Kinnersley frame and retarded KN coordinates, $\thorn = \partial_r$. We chose trivial final 
conditions at $\sI^+$, in the sense that 
\begin{equation}\label{xdec}
x_{nn}, x_{nm} = O(\rho), \quad  x_{m\bar{m}} = O\left( \rho^2 \right). 
\end{equation}
 The Standing Decay Assumptions that we are imposing on the NP 
components of $T_{ab}$ are sufficient to show that such a solution exists.
Having determined $x_{ab}$, we now adjust 
\begin{equation}
\label{eq:hadjust}
h_{ab} \to h_{ab} - x_{ab}. 
\end{equation}
The adjusted $h_{ab}$ still satisfies the Standing Decay Assumptions, still has $h_{ab} l^b = 0$, and additionally satisfies the linearized EE
\begin{equation}\label{eq:ES}
(\E h)_{ab} = S_{ab} \equiv T_{ab} - (\E x)_{ab} \ . 
\end{equation}
Our readjustment \eqref{eq:hadjust} has achieved that the new source, $S_{ab}$, in the linearized EE has vanishing NP components containing at least one $l^a$ tetrad leg. The transport equations for the corrector $x_{ab}$ have a kernel (i.e., non-trivial solutions even for $T_{ab}=0$) which was determined in \cite{HT2}. If we insist on remaining consistent with the decay \eqref{xdec}, then the remaining freedom corresponding to the kernel is to redefine
\begin{equation}
\label{eq:x12dec}
\begin{split}
      x_{nn} \to & \, x_{nn} + (\rho + \bar \rho)\eta_n^\circ,\\
      x_{nm} \to & \, x_{nm} + \rho(\rho + \bar \rho)\eta_m^\circ,
\end{split}
\end{equation}
for unspecified GHP scalars 
$\eta_n^\circ \circeq \GHPw{-3}{-3}, \eta_m^\circ \circeq \GHPw{-2}{-4}$ annihilated by $\thorn$. With the aim to improve the decay of certain NP components of $S_{ab}$ at $\sI^+$, 
we can, and will, adjust these so that 
\begin{equation}
  \begin{split}
  \label{eq:xreadj}
      \tilde{\thorn}' x_{nn}^{1\circ} = & \, T^{2 \circ}_{nn},\\
      - \frac{3}{2} \tilde{\thorn}' x_{nm}^{2\circ} = & \, T^{3\circ}_{nm} + \half \tilde{\eth} x_{nn}^{1\circ} .
\end{split}  
\end{equation}
Since $\tilde{\thorn}'=\partial_u$ in retarded KN coordinates $(r, u,\theta,\varphi_*)$ and the Kinnersley frame, and since GHP scalars annihilated by $\thorn = \partial_r$ do not depend on $r$ in those coordinates, this is always possible. The smooth solutions $x_{nn}^{1\circ}, x_{nm}^{2\circ}$ are 
unique if we insist that they vanish for $u \to -\infty$, as we do. With this choice, the adjusted $h_{ab}$ \eqref{eq:hadjust} still satisfies the 
Standing Decay Assumptions, $h_{ab} l^b = 0$, and \eqref{eq:ES}. We furthermore have the following lemma characterizing the improved 
properties of $S_{ab}$ relative to $T_{ab}$ for our purposes. 

\begin{lemma}
\label{lem:S}
    The tensor $S_{ab}$ defined by  
    \eqref{eq:ES} satisfies the Standing Decay Assumptions, as well as:
    \begin{itemize}
        \item $\nabla^a S_{ab} = 0$. 
        \item $S_{ab} l^b=0=S_{ab}m^a \bar m^b$.
        \item 
        $S_{nn} = O(\rho^3), S_{nm}, S_{mm} = O(\rho^4)$.
    \end{itemize}
\end{lemma}
\begin{proof}
    The first item is a consequence of $\nabla^a T_{ab} = 0 = \nabla^a (\E x)_{ab}$. $S_{ab}l^b=0$ follows by construction of $x_{ab}$, see  \eqref{eq:x}. From these two properties and by computing the divergence of a symmetric rank-2 tensor in GHP form, it can be seen that
    \begin{equation}
        0=l^a \nabla^b S_{ab} = -(\rho+\bar\rho) S_{m\bar m},
    \end{equation}
    from which $0=S_{ab}m^a \bar m^b$ follows immediately since $\Re \rho$ is non-vanishing everywhere. 
    
    For the third item, consider the explicit form of the non-trivial components of $S_{ab}$, given by:
\begin{equation}
\label{eq:S1s}
\begin{split}
S_{mm}
=& \pheq T_{mm}-\bigg[(\thorn-\bar{\rho})(\eth-\tau) - (\eth-\tau-\bar{\tau}')\bar{\rho} -\tau(\thorn+\rho)
+ \bar{\tau}'(\thorn-\rho+\bar{\rho})\bigg]x_{nm} \\
&\pheq+ \bigg[ (\tau+\bar{\tau}')\eth +
(\tau-\bar{\tau}')^2\bigg]x_{m\bar{m}},
\end{split}
\end{equation}
\begin{equation}
\begin{split}
S_{nn}
=& \pheq T_{nn}-
\bigg[(\eth'-\bar{\tau})(\eth-\tau) +\bar{\rho}'(\thorn-\rho+\bar{\rho}) -
(\thorn'-\bar{\rho}')\bar{\rho} + \bar{\Psi}_2\bigg]x_{nn}\\
&\pheq - \bigg[-(\thorn'-3\rho')(\eth'+\tau'-\bar{\tau}) + \tau'\thorn' - \rho'\eth'\bigg]x_{nm}\\
&\pheq - \bigg[-(\thorn'-3\bar{\rho}')(\eth+\bar{\tau}'-\tau) + \bar{\tau}'\thorn' - \bar{\rho}'\eth\bigg]x_{n\bar{m}},
\end{split}
\end{equation}
\begin{equation}
\begin{split}
S_{n\bar{m}}
=& \pheq T_{n\bar{m}}- \half\bigg[(\thorn-\rho+\bar{\rho})(\eth'-\bar{\tau}) - (\eth'-2\tau'+\bar{\tau})\rho +
\tau'(\thorn-\bar{\rho})\bigg]x_{nn}\\
&\pheq -\half\bigg[-\eth'(\eth'-2\tau') -
2\bar{\tau}(\tau'-\bar{\tau})\bigg]x_{nm}\\
&\pheq -\half\bigg[-(\thorn+\bar{\rho})(\thorn'-2\bar{\rho}') + \rho'(\thorn+2\rho-2\bar{\rho}) -
4\rho\bar{\rho}' + 2\Psi_2\\
&\pheq\pheq + (\eth+\bar{\tau}')(\eth'-2\bar{\tau}) - \tau'(\eth+\tau-2\bar{\tau}') -
\tau(\tau'-4\bar{\tau})\bigg]x_{n\bar{m}}\\
&\pheq -\half\bigg[(\thorn'+\rho'-\bar{\rho}')(\eth'-\tau'+\bar{\tau}) + 2\tau'(\thorn'-2\rho')
-(\eth'-\tau'-\bar{\tau})\bar{\rho}'+2\rho'\tau'\bigg]x_{m\bar{m}},
\end{split}
\end{equation}
using formulas for the linearized Einstein operator in GHP form in \ref{app:LinEinGHP}.

To obtain the stated decay properties we
expand the GHP operators and symbols in
terms of Held's operators and symbols, see \ref{app:D}, to 
make explicit the dependence upon $\rho$, and we use the decay properties of $x_{ab}$
given in  \eqref{xdec} and those of $T_{ab}$ given in \eqref{Tdec}. Then it follows at first that $S_{nn} = O(\rho^2), S_{nm} = O(\rho^3), S_{mm} = O(\rho^4)$. 

To get the improved decay rates $S_{nn} = O(\rho^3), S_{nm} = O(\rho^4)$ we consider the corresponding orders of the linearized EEs \eqref{eq:Tnn}, \eqref{eq:Tnm} for the unimproved metric, i.e. before subtracting the corrector. Subtracting the corrector with the properties \eqref{eq:xreadj} and 
\eqref{xdec} implies that we have $S^{2\circ}_{nn} = S^{3\circ}_{nm} = 0$, as we needed to show.
\end{proof}

Now, we look at the $ll$ NP component of the EE $(\E h)_{ll}=0$ for the adjusted metric \eqref{eq:hadjust}. Eq. \eqref{eqn:Ell_app} gives
\begin{equation}
0=\Big[\thorn(\thorn-\rho-\bar{\rho})+2\rho\bar{\rho} \Big]h_{m\bar{m}} \equiv \rho^2 \thorn \left[ \frac{\bar{\rho}}{\rho^3} \thorn \left(\frac{\rho}{\bar{\rho}} h_{m\bar{m}} \right) \right],
\end{equation}
which integrates to \cite{Price,Price2}
\begin{equation}\label{eq:hmmb}
h_{m\bar{m}} = \bar{h}^\circ \frac{\rho}{\bar{\rho}} + h^\circ \frac{\bar{\rho}}{\rho} + (j^\circ + \bar{j}^\circ)(\rho+\bar{\rho}),
\end{equation}
for undetermined GHP scalars $h^\circ \circeq \GHPw{0}{0} ,j^\circ \circeq \GHPw{-1}{-1}$ annihilated by $\thorn$. However, by 
\eqref{xdec}, \eqref{eq:hdec} we know that $h_{m\bar{m}}=O(\rho^2)$ near $\sI^+$. Thus, we learn that $h^\circ = j^\circ = 0$
and thus $h_{m\bar m}= 0$ identically.

To summarize, we have obtained at this stage from the original perturbation satisfying \eqref{eq:Eh=T} by a shift $h_{ab} \to h_{ab} - \dot g_{ab} - (\lie{\xi} g)_{ab}  - x_{ab}$ a new perturbation in TIRG satisfying \eqref{eq:ES}, the Standing Decay Assumptions improved to
\begin{subequations}
\begin{align}
&h_{ll} = h_{ln} = h_{lm} = h_{m\bar{m}} = 0, \\
&h_{nn} = \sum_{j=1}^N \rho^j h_{nn}^{j\circ} + O(\rho^{N+1}),\\ 
&h_{nm} = \sum_{j=1}^N \rho^j h_{nm}^{j\circ} + O(\rho^{N+1}),\\ 
&h_{mm} = \sum_{j=1}^N \rho^j h_{mm}^{j\circ} + O(\rho^{N+1}),
\end{align}
\end{subequations}
and satisfying the EE \eqref{eq:ES} 
whose source $S_{ab}$ has the properties stated in lemma \ref{lem:S}. 
Since we have described a constructive procedure for obtaining $\xi^a, \dot g_{ab}, x_{ab}, S_{ab}$ from the original metric perturbation and source $h_{ab},T_{ab}$, we can from now work with our shifted $h_{ab}$, remembering to add $\dot g_{ab}, x_{ab}$ and $(\lie{\xi} g)_{ab}$ back in at the end. 

\subsubsection{Backward integration for Hertz potential $\Phi$}
\label{Hertzbackward}

We now take the final step and show---making use of every last drop of gauge invariance left over at this stage---that the shifted $h_{ab}$ satisfying the EE \eqref{eq:ES} is of the form $h_{ab} = \Re(\S^\dagger \Phi)_{ab}$, for some potential $\Phi \circeq \GHPw{-4}{0}$ satisfying the $\O^\dagger$-Teukolsky equation with some source to be determined. From the definition of $\S^\dagger$ \eqref{eq:Sdag}, this means that $\Phi$ should simultaneously satisfy the three equations \eqref{eq:hmbmb1}, \eqref{eq:hmbn1}, 
\eqref{eq:hnn1}.
%
%
It is natural to use the first equation \eqref{eq:hmbmb1} in order to define $\Phi$. Eq. \eqref{eq:hmbmb1} can be viewed as a second order ODE along the geodesic tangent to $l^a$. At first, we fix any particular solution to this equation and ask to what extent the second and third equation hold. For this, we define
\begin{equation}
\label{eq:yzdef1}
\begin{split}
y =& \; \text{left side $-$ right side of \eqref{eq:hmbn1},}\\
z =& \; \text{left side $-$ right side of \eqref{eq:hnn1},}
\end{split}
\end{equation}
so obviously we would like that $y=z=0.$ In \cite{Green:2019nam}, this was shown under the hypothesis that 
the metric perturbation and stress tensors vanish in a neighborhood of $\sH^-$. In the present context, this is typically not the case, so we need a more involved argument. We begin with a lemma, which generalizes considerations in \cite[Sec. 4]{Green:2019nam}.

\begin{lemma}
\label{lem:10}
We have
\begin{equation}
\label{eq:yzdef}
\begin{split}
y =& \ \bar{\rho} (\rho + \bar{\rho}) a^\circ + \bar{\rho} (\bar{\rho} - 2\rho) \frac{b^\circ}{\rho^3},\\
z =& \ 2 \Re \left( \bar{\rho}^2 \tilde{\eth} a^\circ + 2 \rho \bar{\rho}^2 \tau^\circ a^\circ + \frac{\bar{\rho}^2}{\rho^3} \tilde{\eth} b^\circ - 4 \frac{\bar{\rho}^2}{\rho^2} \tau^\circ b^\circ \right) + (\rho + \bar{\rho}) c^\circ,
\end{split}
\end{equation}
for certain GHP scalars $a^\circ \circeq \GHPw{-4}{-2}$,
$b^\circ \circeq \GHPw{-1}{1}$, $c^\circ \circeq \GHPw{-3}{-3}$ annihilated by $\thorn$. Note that $c^\circ$ is real.
\end{lemma}

\begin{proof}
We consider the $l\bar{m}$ NP component of the 
EE \eqref{eq:ES}, $(\E h)_{l \bar{m}}=0$; see  \eqref{eqn:Elm_app}. Using 
the SEOT operator identity \eqref{eq:SEOT}
in the form $\E {\rm Re} ( \mathcal{S}^\dagger \Phi ) = {\rm Re}(
\T^\dagger \O^\dagger \Phi)$ and the fact that $\T^\dagger$ has a trivial $lm$ or $l\bar{m}$ NP component, this equation gives $\left(\mathcal{E} (h - \Re \mathcal{S}^\dagger \Phi) \right)_{l\bar{m}}=0$, into which we substitute \eqref{eq:hmbmb1}.
We get
\begin{equation}
[\thorn(\thorn -2\bar{\rho}) +2\rho(\bar{\rho}-\rho)] y = 0.
\end{equation}
This is a second order ODE along the geodesics tangent to $l^a$, which integrates using $\thorn \rho = \rho^2$ to the claimed expression for $y$. 

We next consider the $ln$ NP component of the EE \eqref{eq:ES}, $(\E h)_{ln}=0$; see  \eqref{eqn:Eln_app}. 
Using again the SEOT operator identity \eqref{eq:SEOT}
in the form $\E {\rm Re} ( \mathcal{S}^\dagger \Phi ) = {\rm Re}(
\T^\dagger \O^\dagger \Phi)$ and the fact that $\T^\dagger$ has a trivial $ln$ NP component, this equation implies $\left(\mathcal{E} (h - \Re \mathcal{S}^\dagger \Phi) \right)_{ln}=0$. Into this we substitute \eqref{eq:hmbmb1}, \eqref{eq:hmbn1}, and we use our formula for $y$.
Then, recalling the definition \eqref{eq:yzdef1} of $z$, we get after a lengthy calculation:
\begin{equation}
\begin{split}
\left(\mathcal{E} (h - \Re \mathcal{S}^\dagger \Phi) \right)_{ln} =& \Re \bigg[ \bar{\rho}^2 (\rho^2 - \rho \bar{\rho} - \bar{\rho}^2) \tilde{\eth} a^\circ + \frac{\bar{\rho}^2}{\rho^3} (- \bar{\rho}^2 + \rho \bar{\rho} + 4 \rho^2) \tilde{\eth} b^\circ\\
&-  2 \rho \bar{\rho}^3 (3 \rho + \bar{\rho}) \tau^\circ a^\circ + 4 \frac{\bar{\rho}^2}{\rho^2} (\bar{\rho}^2 - 3 \rho^2) \tau^\circ b^\circ \bigg]\\
&+ \half (\rho + \bar{\rho})^2 \thorn \left(\frac{z}{\rho + \bar{\rho}} \right)= 0.
\end{split}
\end{equation}
for some $c^\circ$.
This is a first order ODE along the geodesics tangent to $l^a$, which integrates using $\thorn \rho = \rho^2$ to
the claimed expression for $z$.
\end{proof}

We have a certain freedom in choosing $\Phi$ because we may add 
a homogeneous solution to equation \eqref{eq:hmbmb1}. Furthermore, we may add a pure gauge piece to $h_{ab}$, provided it preserves the TIRG. The gauge vector fields $\zeta^a$ preserving the TIRG were classified in \cite{Price,Price2} and are recalled in theorem \ref{thm:resgauge} of \ref{sec:residualg}. They have the form \eqref{eq:residual1} in terms of certain  GHP scalars $\zeta^\circ_l \circeq \GHPw{1}{1}$, $\zeta^\circ_n \circeq \GHPw{-1}{-1}$ and $\zeta^\circ_m \circeq \GHPw{2}{0}$ annihilated by $\thorn$, satisfying the conditions given in \eqref{eq:residcond}, i.e.,
\begin{subequations}\label{eq:residcond2}
\begin{equation}\label{eq:residcond:Xil}
\tilde \thorn' \zeta_l^\circ = -\Re \left( \tilde{\eth}' \zeta_m^\circ \right),
\end{equation}
\begin{equation}\label{eq:residcond:Xin}
\zeta_n^\circ = \Re \Big[ \left(\tilde{\eth}' \tilde{\eth} - \rho^{\prime\circ} \right)\zeta^\circ_l - \Omega^\circ \tilde{\eth}' \zeta_m^\circ + 2 \bar{\tau}^\circ \zeta^\circ_m \Big]. 
\end{equation}    
\end{subequations}
Altogether, if we insist on preserving the TIRG for $h_{ab}$ and on preserving \eqref{eq:hmbmb1}, we are free to choose $d^\circ \circeq \{-4,0\}$, $e^\circ \circeq \{-1,3\}$ and $\zeta^\circ_m, \zeta^\circ_n, \zeta^\circ_l$ restricted by \eqref{eq:residcond2}, and to make a change of the form
\begin{equation}
\label{Phireddef}
\Phi \to \Phi + d^\circ + \frac{2}{\rho} \left(\tilde{\eth}'\tilde{\eth}' \zeta_l^\circ - 2 \bar{\tau}^\circ \zeta^\circ_{\bar{m}} \right) - \frac{2}{\rho^2} \tilde{\eth}' \zeta_{\bar{m}}^\circ + \frac{1}{\rho^3} e^\circ, 
\end{equation}
and simultaneously
\begin{equation}
\label{hreddef}
h_{ab} \to h_{ab} - (\lie{\zeta} g)_{ab}, 
\end{equation}
where $\zeta^a$ is determined by 
\eqref{eq:residcond}.

We now seek to satisfy equations \eqref{eq:hmbn1}, \eqref{eq:hnn1}, i.e. to set $y=z=0$ by making the changes proposed through \eqref{Phireddef}, \eqref{hreddef}.
Firstly, for \eqref{eq:hmbn1} to hold, one must set $a^\circ, b^\circ$ of lemma \ref{lem:10} to zero. One shows by a lengthy calculation using the relations of Held's formalism (see \ref{sec:Held}) that this will be achieved if 
$d^\circ, e^\circ$ satisfy
\begin{equation}
\label{eq:dcircdet}
\begin{split}
\tilde{\eth} \, d^\circ =& \; -2 a^\circ + 2 \Omega^\circ \tilde{\eth}' \tilde{\eth}' \tilde{\eth} \; \zeta^\circ_l - 2 \Omega^{\circ 2} \tilde{\eth}' \tilde{\eth}' \zeta^\circ_m - 2 \left(2 \Omega^\circ \bar{\rho}^{\prime \circ} - \bar{\Psi}^\circ \right) \tilde{\eth}' \zeta^\circ_l \\
& - 2 \Omega^\circ \bar{\tau}^\circ \tilde{\eth} \; \zeta^\circ_{\bar{m}} + 2 \Omega^\circ \bar{\tau}^\circ \tilde{\eth}' \zeta^\circ_m + 2 \Omega^\circ \Psi^\circ \zeta^\circ_{\bar{m}} ,
\end{split}
\end{equation}
and
\begin{equation}
\label{eq:ecircdet}
\tilde{\eth} \, e^\circ = -2 b^\circ - 2 \tilde{\thorn}' \zeta^\circ_{\bar{m}},
\end{equation}
for given $\zeta^\circ_m, \zeta_n^\circ, \zeta_l^\circ$. Secondly, to satisfy equation \eqref{eq:hnn1}, we must further set $c^\circ$ of lemma \ref{lem:10} to zero. One sees after a lengthy calculation using the relations of Held's formalism (see \ref{sec:Held}) that this will be the case if 
%
\begin{equation}\label{eq:c2}
c^\circ = - \Re \bigg\{ \tilde{\eth}' \tilde{\eth}' \tilde{\eth} \, \tilde{\eth} \, \zeta^\circ_l 
- \bigg[  \Omega^\circ \left(\tilde{\eth}' \tilde{\eth} + \tilde{\eth} \, \tilde{\eth}' \right) + 2 \tau^\circ \tilde{\eth}'  - 2 \bar{\tau}^\circ \tilde{\eth} + \half \Psi^\circ - \frac{7}{2} \bar{\Psi}^\circ \bigg] \tilde{\eth}'
\zeta^\circ_m \bigg\}.
\end{equation}
Finally, $\zeta^\circ_l, \zeta^\circ_n, \zeta^\circ_m$ must be related through \eqref{eq:residcond2}.

We postpone the fairly non-trivial proof that there exists a smooth solution $(d^\circ, e^\circ, \zeta^\circ_l, \zeta^\circ_n, \zeta^\circ_m)$ to equations \eqref{eq:residcond2}, \eqref{eq:c2}, \eqref{eq:ecircdet}, \eqref{eq:dcircdet} (under a decay assumption on the Bondi news tensor, see theorem \ref{thm:2})
to \ref{sec:zetaproof}, and perform the changes proposed through \eqref{hreddef} and \eqref{Phireddef}. Then \eqref{eq:hmbmb1}, \eqref{eq:hnn1}
and \eqref{eq:hmbn1} hold. 

If we combine the gauge vector field $\zeta^a$ defined by 
\eqref{eq:residual1} and the gauge vector field $\xi^a$ defined by
\eqref{eq:xidef}
into 
\begin{equation}
    X^a = \zeta^a + \xi^a,
\end{equation}
if we take into account the zero mode $\dot g_{ab}$, and the corrector $x_{ab}$ defined in Sec. \ref{sec:xabdef}, then  we can summarize our analysis\footnote{Including the steps relegated to \ref{sec:zetaproof}.} as follows. 

\begin{theorem}
\label{thm:2}
Suppose the smooth perturbation $h_{ab}$ satisfies the inhomogeneous linearized EE $(\E h)_{ab} = T_{ab}$ on the exterior region of Kerr, and that the NP components of $h_{ab}, T_{ab}$ satisfy the Standing Decay Assumptions (see Sec. \ref{sec:gdot}) near $\sI^+$.
Assume, additionally, that the Bondi news tensor $N_{AB}$
(see \ref{BondiIRG}) satisfies, for some $\epsilon>0$, 
\begin{equation}
N_{AB} = O(|u|^{-\frac{3}{2} - \epsilon}) \quad \text{as $|u| \to \infty$,}
\end{equation}
and that the same holds for all of its derivatives along $\sI^+$.
Then there exists a gauge vector field $X^a$, a zero mode $\dot g_{ab}$, a GHP scalar $\Phi \circeq \GHPw{-4}{0}$, and a corrector $x_{ab}$ satisfying the transport equations \eqref{eq:xmmb}, \eqref{eq:xnm}, \eqref{eq:xnn}, such that
\begin{equation}
\label{eq:decomp2}
h_{ab} = x_{ab} + \dot g_{ab} + (\lie{X} g)_{ab} + \Re(\S^\dagger \Phi)_{ab},    
\end{equation}
and such that the nonzero NP components 
(in the Kinnersley frame) have the following decay at $\sI^+$ ($r \to \infty$ at fixed $(u,\theta,\varphi_*)$ in retarded KN coordinates)
\begin{subequations}
    \begin{align}
        \Phi =& \ O(r^3),\\
        x_{nn}, x_{nm} =& \ O(r^{-1}),\\
        x_{m\bar m} =& \ O(r^{-2}),\\
X_l,X_n =& \ O(1),\\
X_m =& \ O(r),\\
\dot g_{nn} =& \ O(r^{-1}),\\
\dot g_{nm} =& \ O(r^{-2}).
    \end{align}
\end{subequations}
$\Phi, x_{ab}, X^a$ are uniquely determined by these conditions if we insist that the leading terms in their asymptotic expansions as in \eqref{asympt_exp} all go to zero as $u \to -\infty$.
\end{theorem}

\noindent
{\bf Remark 2.} 
From the perspective of the Cauchy problem (see e.g., \cite{Ringstrom}) the decay of the Bondi news tensor as $u \to -\infty$, can be achieved comfortably if we require the Cauchy data to coincide with those of a zero mode outside a compact set; see also footnote \ref{footnote:1}. By contrast, the decay assumption as $u \to +\infty$ on the Bondi news tensor is a much more non-trivial statement. It is consistent with, though significantly weaker than, ``Price's law'' \cite{RPrice}, as proven for the Teukolsky equations of various 
spins in sub-extremal Kerr by \cite{Angelopoulos,Hintz,Ma,Shlapentokh-Rothman:2023bwo}, and for the metric perturbation itself by 
\cite{Hafner} (for small $a/M$ and somewhat non-optimal decay rates). For prior works in Kerr, see also e.g., \cite{DHR}. [N.B.: Price's law would correspond to 
\begin{equation}
N_{AB} = O(u^{-5 + \epsilon}) \quad \text{as $u \to \infty$,}
\end{equation}
for an arbitrarily small $\epsilon>0$.]
The proof of the GHZ decomposition by backward integration, including the steps relegated to \ref{sec:zetaproof}, shows that all quantities in the decomposition have full asymptotic expansions of the form \eqref{asympt_exp}. However, some of these quantities may fail to decay as $|u| \to \infty$, as would be the case e.g., for certain terms in the asymptotic expansion of $\Phi$ or certain NP components of $X^a$. 

\subsection{Proof of the GHZ decomposition: boundary conditions at $\sH^-$}
\label{GHZproofH-}

\subsubsection{Forward integration: Conditions at $\sH^-$ and zero mode}
\label{sec:forward}

We will now construct a GHZ decomposition \eqref{eq:decompi} corresponding to a boundary condition for the transport equations at $\sH^-$. 
For this, we will need to require certain conditions on $h_{ab}$ (and by the EE, on $T_{ab}$) at $\sH^-$ analogous to \eqref{eq:hdec}. These 
conditions are physically more restrictive than those \eqref{eq:hdec} imposed at $\sI^+$, in the sense that they forbid ingoing fluxes of 
gravitational radiation and matter stress-energy-angular momentum at $\sH^-$. 
Our conditions on $h_{ab}$ are precisely as follows:

\medskip
\noindent
{\bf Standing Assumptions on $h_{ab}$
at $\sH^-$:}
We assume that $h_{ab}$ is twice continuously extendible across $\sH^-$ and that $h_{ab}$ satisfies the following conditions in a regular NP frame at $\sH^-$, possibly up to the addition of a zero mode or a pure gauge perturbation:
\begin{equation} 
\label{eq:hdec1}
\begin{split}
&h_{nn}, h_{nm},h_{m\bar m} = O((r-r_+)^2), \\
&h_{mm}, h_{lm}, h_{ll}, h_{ln} = O(r-r_+) \quad \text{at $\sH^-$.}
\end{split}
\end{equation}

\medskip
\noindent
{\bf Remark 3.}
The conditions $h_{nn} = O((r-r_+)^2), h_{nl}, h_{nm} = O(r-r_+)$, alone can be seen as mere gauge conditions.
Additionally, the conditions $h_{m\bar m} = {\rm const.} + O(r-r_+)$ and $\th' h_{m\bar m} = O(r-r_+)$ could be imposed at any given single cross section of $\sH^-$ such as e.g., the bifurcation surface $\mathscr B$, by a further change of gauge\footnote{As given, the argument by \cite{HWcanonicalenergy} only applies to sub-extremal black holes.} \cite{HWcanonicalenergy}. 
All other conditions are physical restrictions, related to the matter stress energy tensor, gravitational radiation, and geometric properties of $\sH^-$ under perturbations, as the following proposition
\ref{prop:standass} clarifies.

\begin{proposition}
\label{prop:standass}
Assume that the Standing Assumptions
at $\sH^-$ and the linearized EE $(\E h)_{ab}=T_{ab}$
hold. Then
\begin{equation}
\label{eq:Standassequiv}
T_{nn} = T_{nm} = T_{nl} = \psi_4 = 0 \quad \text{at $\sH^-$,}
\end{equation}
i.e., one is imposing the absence of outgoing fluxes of 
gravitational radiation and matter stress-energy-angular momentum at $\sH^-$. Furthermore,
\begin{equation}
\delta A = \delta \kappa = \delta \theta = 0 \quad \text{at $\sH^-$,}
\end{equation}
where $\delta A,\delta \theta$ and $\delta \kappa$ are the perturbed horizon area, perturbed expansion of $n^a$ along $\sH^-$, and 
perturbed surface gravity of $\chi^a = t^a + \Omega_+ \phi^a$, respectively.   

Conversely, assume that $h_{ab}$ is twice continuously differentiable across $\sH^-$ and that the 
NP components of $h_{ab}, \nabla_c h_{ab}$ in a regular frame Lie-derived by $\chi^a$ go to zero along 
affine parameters of $\sH^-$ reaching $i^-$. Then \eqref{eq:Standassequiv} imply the 
Standing Assumptions on $h_{ab}$ at $\sH^-$ in some gauge.
\end{proposition}

\begin{proof}
$\Longrightarrow:$
The congruence of geodesics tangent to the principal null direction $n^a$ 
is expansion, shear and twist free, so we have 
$
\rho' = 0 \quad \text{on $\sH^-$.}
$
Using this fact in the $nm$- and $nn$ components of the linearized Einstein equations in GHP form in a type D background 
(see \ref{app:LinEinGHP}) gives, on $\sH^-$:
\begin{equation}
\begin{split}
T_{nn} 
&= \pheq  \Big[(\eth'-\bar{\tau})(\eth-\tau)  -\thorn'\bar{\rho} + {\bar\Psi}_2\Big]h_{nn} + \Big[-\thorn'(\eth'+\tau'-\bar{\tau}) + \tau'\thorn'\Big]h_{nm}\\
&\pheq + \Big[-\thorn'(\eth+\bar{\tau}'-\tau) + \bar{\tau}'\thorn' \Big]h_{n\bar{m}} + \thorn'\thorn'h_{m\bar{m}},
\end{split}
\end{equation}
\begin{equation}
\begin{split}
T_{ln} &= \half\Big[\rho(\thorn-\rho) + \bar{\rho}(\thorn-\bar{\rho})\Big]h_{nn}\\
&\pheq + \half\Big[-(\eth'+\tau'+\bar{\tau})(\eth-\tau-\bar{\tau}') -
(\eth'\eth+3\tau\tau'+3\bar{\tau}\bar{\tau}') + 2(\bar{\tau}+\tau')\eth\\
&\pheq\pheq+ \thorn(\rho') +\thorn'\bar{\rho}  -
\rho\thorn' -\Psi_2 -{\bar\Psi}_2\Big]h_{ln}\\
&\pheq + \half\Big[\thorn'(\eth'-\tau') + \bar{\tau}\thorn'
-\tau'\thorn'\Big]h_{lm}\\
&\pheq + \half\Big[\thorn'(\eth-\bar{\tau}') + \tau\thorn'
-\bar{\tau}'\thorn'\Big]h_{l\bar{m}}\\
&\pheq + \half\Big[(\thorn-2\rho)(\eth'-\bar{\tau}) + (\tau'+\bar{\tau})(\thorn+\bar{\rho})
-2(\eth'-\tau')\rho-2\bar{\tau}\thorn\Big]h_{nm}\\
&\pheq + \half\Big[(\thorn-2\bar{\rho})(\eth-\tau) + (\bar{\tau}'+\tau)(\thorn+\rho)
-2(\eth-\bar{\tau}')\bar{\rho}-2\tau\thorn\Big]h_{n\bar{m}}\\
&\pheq + \half\Big[-(\eth'-\bar{\tau})(\eth'-\tau') +
\bar{\tau}(\bar{\tau}-\tau')\Big]h_{mm}\\
&\pheq + \half\Big[-(\eth-\tau)(\eth-\bar{\tau}') +
\tau(\tau-\bar{\tau}')\Big]h_{\bar{m}\bar{m}}\\
&\pheq + \half\Big[(\eth'+\tau'-\bar{\tau})(\eth-\tau+\bar{\tau}') +
(\eth'\eth-\tau\tau'-\bar{\tau}\bar{\tau}'+\tau\bar{\tau}) - (\Psi_2+{\bar\Psi}_2)\\
&\pheq \pheq+\thorn'\bar{\rho} +3\rho\thorn'-2\thorn'\thorn +2\eth'(\tau)-\tau\bar{\tau}\Big]h_{m\bar{m}},
\end{split}
\end{equation}
and
\begin{equation}
\begin{split}
T_{n\bar{m}} &= 
\pheq  \half\Big[(\thorn-\rho+\bar{\rho})(\eth'-\bar{\tau}) - (\eth'-2\tau'+\bar{\tau})\rho +
\tau'(\thorn-\bar{\rho})\Big]h_{nn}\\
&\pheq - \half\thorn'(\eth'+\tau'-\bar{\tau})h_{ln}+\half\thorn'\thorn'h_{l\bar{m}}\\
&\pheq +\half\Big[-\eth'(\eth'-2\tau') -
2\bar{\tau}(\tau'-\bar{\tau})\Big]h_{nm}\\
&\pheq +\half\Big[-(\thorn+\bar{\rho})\thorn' + 2\Psi_2 + (\eth+\bar{\tau}')(\eth'-2\bar{\tau}) - \tau'(\eth+\tau-2\bar{\tau}') -
\tau(\tau'-4\bar{\tau})\Big]h_{n\bar{m}}\\
&\pheq -\half\thorn'(\eth-\tau+\bar{\tau}')h_{\bar{m}\bar{m}}+\half\Big[\thorn'(\eth'-\tau'+\bar{\tau}) + 2\bar{\tau}\thorn'\Big]h_{m\bar{m}} .
\end{split}
\end{equation}
A zero mode or a pure gauge perturbation is in particular a solution to the linearized EE, thus making no contribution to either $T_{nm}, T_{nn}.$ So we may ignore any such contribution to $h_{ab}$ in the Standing Assumptions for this discussion. 
The GHP directional derivative operators $\eth,\eth',\thorn'$ are along $m^a, \bar m^a, n^a$, which are tangent to $\sH^-$. Therefore, except for 
terms involving $\thorn h_{nm}, \thorn h_{nn}, \thorn h_{m\bar m}$ in $T_{nm}$ and in $T_{nl}$, all derivatives in the above expressions are tangent to $\sH^-$. However, those terms involve at most one derivative ($\thorn$) transversal to $\sH^-$, while $h_{nm},h_{nn},h_{m\bar m} = O((r-r_+)^2)$ vanish to second order. So $T_{nn}=T_{nm}=T_{nl}=0$ on $\sH^-$.

The perturbed Weyl component $\psi_4 = \T' h$ is obtained by acting with the operator $\T'$ on $h_{ab}$, see \ref{app:ST}. This operator annihilates any pure gauge perturbation or zero mode. So we may ignore again any such contributions to $h_{ab}$ in the Standing Assumptions. Taking the GHP prime of equation \eqref{eq:S} and using $\rho'=0$
gives
\begin{equation}
    \begin{split}
        \psi_4 ={}& \half(\eth'-\bar \tau)(\eth'-\bar \tau) h_{nn} + \half \thorn'\thorn' h_{\bar m\bar m} \nonumber \\
            &-\half\Big[\thorn'(\eth'-2\bar \tau) + (\eth' -\bar \tau)\thorn'\Big] h_{(n\bar m)}
    \end{split}
\end{equation}
on $\sH^-$. Since $\eth,\eth',\thorn'$ are derivatives along $\sH^-$, it follows again that 
$\psi_4=0$ there.

The perturbed area element on a cut of 
$\sH^-$ is proportional to $h_{m\bar m}$
and so it vanishes.
The perturbed expansion $\delta \theta$ of $n^a$ along $\sH^-$
is proportional to $n^a \nabla_a h_{m\bar m}$, which vanishes, since $h_{m\bar m}$ itself is constant on $\sH$. 
See \cite{Hollands:2024vbe} for a proof that $\delta \kappa$ vanishes if $h_{ab}\chi^b = O(r-r_+)$, $h_{ab}\chi^a \chi^b = O((r-r_+)^2)$, which in turn 
follows from the Standing Assumptions since on $\sH^-$, $n^a$ is proportional to $\chi^a$. 

\medskip
$\Longleftarrow:$ The conditions $h_{nn} = O((r-r_+)^2), h_{nl}, h_{nm} = O(r-r_+)$, can be imposed as gauge conditions by passing to a Gaussian null gauge based on an affine parameter on $\sH^-$ \cite{HWcanonicalenergy}. Actually, by choosing such a gauge, we can furthermore impose $h_{ab}r^b=0$, where $r^a$ is a null vector such that $r^a n_a = 1$ and such that the span of $n^b, r^b$ is orthogonal to the affine foliation, i.e. 
\begin{equation}
    n_{[a} r_b \nabla_c r_{d]} = 0 = 
    n_{[a} r_b \nabla_c n_{d]},
\end{equation}
by Frobenius' theorem. Now, from the expression for $T_{nn}$ we then get $\thorn' \thorn' h_{m\bar m}=T_{nn}=0$ on $\sH^-$. In view of the assumptions, as $\thorn'$ is a 
covariant derivative along $n^a$, we get $h_{m\bar m}=0$ on $\sH^-$. Next, from the expression for $\psi_4$
we get $\frac{1}{2} \thorn' \thorn' h_{\bar m\bar m}=\psi_4=0$ on $\sH^-$, so by the same 
argument, we have $h_{mm}=0$ on $\sH^-$. 
Using this and a GHP commutator, 
we then also get $-\frac{1}{2}\thorn'\thorn h_{nm} = T_{nm}=0$ on $\sH^-$, so again by the same 
argument, we have $\thorn h_{mn}=0$ on $\sH^-$, and since $\thorn$ is transverse to $\sH^-$, we have $h_{mn}=O((r-r_+)^2)$.
$l^a$ is a linear combination of $r^a, n^a, m^a, \bar m^a$, so from $h_{ab}r^b=0$
we then get also $h_{ll}, h_{ln}, h_{lm}=O(r-r_+)$. Using all of this information in our formula for $T_{nl}$, 
we finally get $\thorn' \thorn h_{m\bar m}=0$ on $\sH^-$, so again $\thorn h_{m\bar m}=0$ since $\thorn'$ is a covariant derivative along the null generators $\sH^-$ and $\thorn h_{m\bar m}=0$ goes to zero along those generators to $i^-$. Therefore $h_{m\bar m}=O((r-r_+)^2)$,
demonstrating all of the claims in the Standing Assumptions at $\sH^-$. 
\end{proof}

From now on, we assume the Standing Assumptions at $\sH^-$---or their equivalent characterization given by proposition \ref{prop:standass})---hold after 
a suitable zero mode and gauge perturbation has been subtracted, i.e. we will consider from now 
$h_{ab}-\dot g_{ab}-(\mathcal{L}_\eta g)_{ab}$ in place of $h_{ab}$.
We have to remember adding $\dot g_{ab}$
and $(\mathcal{L}_\eta g)_{ab}$ back in at the end. 

\subsubsection{Forward integration for gauge vector field $\xi^a$}
\label{sec:gaugexi-}

We now demonstrate that the perturbed metric can be put in IRG ($h_{ab} l^b=0$) without affecting the Standing Assumptions at $\sH^-$. 
 If $h_{ab}$ is not in IRG to begin with, we can solve for a gauge vector field $\xi^a$ satisfying 
$[h_{ab} - (\lie{\xi} g)_{ab}] l^a = 0$ by integrating the equations \eqref{eq:xidef} for the NP components 
$\xi_l, \xi_n, \xi_m$. As before, 
we first integrate the equation for 
$\xi_l$, viewing the first equation in \eqref{eq:xidef} as on ODE along the orbits of the principal null direction $l^a$, starting at 
the horizon, $\sH^-$, i.e., at $r=r_+$, with 
zero initial condition.
Since the source $h_{ll}$ in that ODE
is of order $O(r-r_+)$ by the Standing Assumptions, we have $\xi_l = O((r-r_+)^2)$. Next, we view the second equation in \eqref{eq:xidef} as on ODE
for $\xi_m$. Again, we impose zero initial condition, and we see that the source 
in the equation, now 
$-(\eth + \bar{\tau}') \xi_l + h_{lm}$, 
is again of order $O(r-r_+)$. Thus, 
we have $\xi_m = O((r-r_+)^2)$. Finally, 
we view that last equation in \eqref{eq:xidef} as on ODE
for $\xi_m$. Again, we impose zero initial condition, and we see that the source 
in the equation, now 
$-\thorn' \xi_l - (\tau + \bar{\tau}') \xi_{\bar{m}} - (\bar{\tau} + \tau') \xi_m + h_{ln}$, is again of order $O(r-r_+)$.
Thus, we have, altogether
\begin{equation}
    \xi_n, \xi_l, \xi_m = O((r-r_+)^2), 
\end{equation}
and we have $(h_{ab} - \lie{\xi} g_{ab}) l^a = 0$. This shows that the IRG perturbation $h_{ab} - (\lie{\xi} g)_{ab}$
fulfills the Standing Assumptions for all NP components involving $l$. 

However, we need to show that we have not destroyed any of the remaining Standing Assumptions. We therefore consider, using formulas of \ref{sec:residualg}:
\begin{subequations}
\begin{align}
\label{Liexignn}
\left( \lie{\xi} g \right)_{nn} =& \  2 \thorn' \xi_n = O((r-r_+)^2), \\
\label{Liexignm}
\left( \lie{\xi} g \right)_{nm} =& \  \left( \eth + \tau \right) \xi_n + \left( \thorn' + \rho' \right) \xi_m = O((r-r_+)^2), \\
\label{Liexigmm}
\left( \lie{\xi} g \right)_{mm} =& \  2 \eth \xi_m = \ O((r-r_+)^2), \\
\label{Liexigmmb}
\left( \lie{\xi} g \right)_{m\bar{m}} =& \left( \rho'+ \bar{\rho}' \right) \xi_l + \left( \rho + \bar{\rho} \right) \xi_n + \eth' \xi_m - \eth \xi_{\bar{m}} = O((r-r_+)^2),
\end{align}
\end{subequations}
since $\thorn,\eth,\eth'$ are tangent to $\sH^-$. This shows that the IRG perturbation $h_{ab} - (\lie{\xi} g)_{ab}$
fulfills the Standing Assumptions for all NP components. We now pass from $h_{ab}$
to $h_{ab} - (\lie{\xi} g)_{ab}$, remembering to add $(\lie{\xi} g)_{ab}$
in the end. Then the new metric perturbation
satisfies all Standing Assumptions, 
and additionally the IRG.

\subsubsection{Forward integration for the corrector $x_{ab}$ and Hertz potential $\Phi$}
\label{sec:x_abPhi-}

Next we integrate the equations for the corrector, \eqref{eq:xmmb}, \eqref{eq:xnm}, \eqref{eq:xnn}. As with all transport equations in the GHZ method, these are again understood as ODEs along the outgoing null geodesics tangent to $l^a$. The equations for \eqref{eq:xmmb}, \eqref{eq:xnm} are of second order in the affine parameter $r$ along those geodesics, so we may impose trivial initial conditions at $r=r_+$ for $x_{nm},x_{m\bar m}$ as well as their first $r$-derivative. Thus, we trivially get that $x_{nm},x_{m\bar m} =O((r-r_+)^2)$. By contrast, the equation \eqref{eq:xnn} for $x_{nn}$ is only of first order in $r$, so we may only impose the vanishing of $x_{nn}$ at $r=r_+$. 
However, based on what we know already, 
and by $T_{nl} = O(r-r_+)$ (see proposition \ref{prop:standass}), the source in \eqref{eq:xnn} is of order $O(r-r_+)$, and so we get $x_{nn} =O((r-r_+)^2)$, in fact.

We now pass from $h_{ab}$ to $h_{ab}-x_{ab}$ just as we did in the backward integration method, and the new $h_{ab}$
will then still satisfy all of our Standing Assumptions, and it has a new stress energy tensor, $S_{ab}$, such that 
$S_{ab}l^b=0$. 

Finally, we need to find the Hertz potential, $\Phi$. As in the backward integration scheme, it must simultaneously satisfy the three equations \eqref{eq:hmbmb1}, \eqref{eq:hmbn1}, 
\eqref{eq:hnn1}. As before, we view the first equation \eqref{eq:hmbmb1}, a second order transport equation along the 
orbits of $l^a$, as the defining equation for $\Phi$. In the forward integration scheme, we impose initial conditions at $r=r_+$. We may ask that $\Phi$ as well as its first $r$-derivative vanishes at $r=r_+$. Furthermore, since the source, 
$h_{\bar m\bar m}$, in the equation \eqref{eq:hmbmb1} is of order $O(r-r_+)$
by the Standing Assumptions, it follows that, in fact 
\begin{equation}
    \Phi=O((r-r_+)^3).
\end{equation}
Consider the failure of the remaining two equations, \eqref{eq:hmbn1}, 
\eqref{eq:hnn1} for $\Phi$, to hold, i.e. consider $y,z$ as in \eqref{eq:yzdef}. These quantities must 
have the form stated in lemma \ref{lem:10}. On the other hand, $\Phi=O((r-r_+)^3)$ and the Standing Assumptions give us $y=O((r-r_+)^2)=z$.
This gives us 2 independent equations per outgoing null geodesic for both $y,z$, that must hold at $r=r_+$. On the other hand, $y,z$ are completely fixed  the three complex GHP scalars $a^\circ, b^\circ, c^\circ$ annihilated by $\thorn$, i.e. constant along each outgoing null generator. Hence, we get 4 equations for 
$a^\circ, b^\circ, c^\circ$ per outgoing null generators, which are easily seen to imply that $a^\circ = b^\circ = c^\circ = 0$, so $y=z=0$. Therefore, 
all three equations \eqref{eq:hmbmb1}, \eqref{eq:hmbn1}, 
\eqref{eq:hnn1} hold, and we learn that $h_{ab} = \Re(\S^\dagger \Phi)_{ab}$. 

Now set $X^a = \xi^a + \eta^a$, we add $x_{ab}, (\mathcal{L}_X g)_{ab}, \dot g_{ab}$ back into $h_{ab}$, and summarize our analysis.

\begin{theorem}
\label{thm:3}
Suppose the perturbation $h_{ab}$ satisfies the inhomogeneous linearized EE $(\E h)_{ab} = T_{ab}$ on the exterior region of Kerr, where $h_{ab}$ satisfies the Standing Assumptions at $\sH^-$ (or the equivalent characterizations stated in proposition \ref{prop:standass}).
Then there exists a gauge vector field $X^a$, a GHP scalar $\Phi\circeq \GHPw{-4}{0}$, a zero mode $\dot g_{ab}$, and a corrector $x_{ab}$ satisfying the transport equations \eqref{eq:xmmb}, \eqref{eq:xnm}, \eqref{eq:xnn}, such that
\begin{equation}\label{eq:decomp3}
h_{ab} = \Re(\S^\dagger \Phi)_{ab} + x_{ab} + (\lie{X} g)_{ab} + \dot g_{ab},
\end{equation}
and such that the nonzero NP components of each quantity 
have the following behavior at $\sH^-$ ($r =r_+$ at fixed $(u,\theta,\varphi_*)$ in retarded KN coordinates and the Kinnersley frame)
\begin{subequations}
    \begin{align}
        \Phi =& \ O((r-r_+)^3),\\
        x_{nn}, x_{nm}, x_{m\bar m} =& \ O((r-r_+)^2),\\
X_l,X_n,X_{m} =& \ O((r-r_+)^2),\\
\dot g_{nn}, \dot g_{nm} =& \ O(1).
    \end{align}
\end{subequations}
\end{theorem}

\subsubsection{Comparison between the backward and forward GHZ integration schemes}

The main difference between the GHZ decompositions in theorems \ref{thm:2} and \ref{thm:3} is that, in the latter case, we 
are imposing physically much more restrictive conditions  
on the metric at $\sH^-$. These effectively forbid any outgoing gravitational radiation or matter energy-angular-momentum fluxes through $\sH^-$ by proposition \ref{prop:standass}. By contrast, such fluxes at $\sI^+$ are allowed in theorem \ref{thm:2}. Not surprisingly, the proof of theorem \ref{thm:2} has been much harder. 

It is unclear to us whether a proof of theorem \ref{thm:3} could be given under weaker conditions allowing fluxes at $\sH^-$. At any rate, in physical setups related to the computation of the self-force on a body in a bound orbit around Kerr, we are typically interested in forward (retarded) solutions. For the such solutions, there should be no outgoing gravitational radiation or matter energy-angular-momentum fluxes through $\sH^-$, in which case the forward integration method would apply. 


\subsection{Properties of Hertz-potential $\Phi$ in GHZ decomposition}

\subsubsection{Equation satisfied by $\Phi$}
\label{sec:Phieq}

We now turn to the equation satisfied by $\Phi$ in the GHZ decomposition, as provided by theorems \ref{thm:2} respectively \ref{thm:3}. To this end, we act with the linearized Einstein operator, $\E$, on the GHZ decomposition \eqref{eq:decompi} respectively \eqref{eq:decomp2}, and we use the SEOT identity \eqref{eq:SEOT}. This gives
\begin{equation}\label{eq:86}
\Re(\T^\dagger \O^\dagger \Phi)_{ab} = T_{ab}-(\E x)_{ab} \equiv S_{ab}.
\end{equation}
Thus, $\Phi$ solves the sourced, adjoint Teukolsky equation
\begin{equation}
\O^\dagger \Phi = \eta,
\end{equation}
with an $\eta \circeq \GHPw{-4}{0}$ such that
\begin{equation}
\label{eq:eta}
\Re (\T^\dagger \eta)_{ab} = S_{ab}. 
\end{equation}
In NP components, \eqref{eq:eta} is by  \eqref{eq:Tdag} equivalent to
\begin{subequations}
\begin{align}
\label{eq:S1}
S_{\bar{m} \bar{m}} =& \  \frac{1}{4} (\thorn-\rho)(\thorn-\rho)\eta ,\\
\label{eq:S2}
S_{nn} =& \ \half (\eth -\tau)(\eth-\tau)\eta + \textrm{c.c.} ,\\
\label{eq:S3}
S_{\bar{m} n} = & \ \frac{1}{4}  \bigg[(\eth + \bar{\tau}' -\tau)(\thorn - \rho) + (\thorn - \rho + \bar{\rho})(\eth-\tau) \bigg] \eta.
\end{align}
\end{subequations}
We would now like to determine $\eta$ from these equations. To do so, 
we need information about the global behavior of $S_{ab}$, hence of the metric $h_{ab}$ and the stress tensor $T_{ab}$, and we also need to 
distinguish what integration scheme (forward respectively backward) was used to determine the corrector $x_{ab}$ which enters $S_{ab}$ (see 
Sec. \ref{GHZproofH-} respectively Sec. \ref{GHZproof}). 

We will discuss separately the following model cases:
\begin{enumerate}
    
\item The stress tensor $T_{ab}$ is zero for $r>r_\textrm{max}$ for a certain $r$-value $r_\textrm{max}>r_+$. In this case, we shall use the 
backward integration scheme and we shall make additionally the assumptions on $h_{ab}$ described in theorem \ref{thm:2}. 

\item The stress tensor $T_{ab}$ is zero for $r<r_\textrm{min}$ for a certain $r$-value $r_\textrm{min}>r_+$. In this case, we shall use the 
forward integration scheme and we shall make additionally the assumptions on $h_{ab}$ described in theorem \ref{thm:3}. 

\end{enumerate}
In order to characterize $\eta$ in either of these cases, we will need a lemma. 

\begin{lemma}
\label{lem:pk}
In each of the above cases 1) respectively 2) there exist -- possibly different -- GHP scalars $k^\circ \circeq \GHPw{-4}{0}$ and $p^\circ \circeq \GHPw{-5}{-1}$ in region $r>r_{\rm max}$ respectively $r<r_{\rm min}$, such that
\begin{equation}\label{s1'}
\eta = p^\circ \rho + k^\circ ,
\end{equation}
and such that
\begin{equation}\label{eq:etaconstraints}
\begin{split}
0 =&   \, \tilde{\eth} \, p^\circ - \Omega^\circ \tilde{\eth}\, k^\circ,\\
0 =&   \, \Re \left(\tilde{\eth}\, \tilde{\eth}\, k^\circ\right),
\end{split}
\end{equation}
in the respective regions.
\end{lemma}
\begin{proof}
 In either case 1) or 2), we consider the region where the source $T_{ab}$ in the linearized Einstein equation is zero, i.e. $r>r_\textrm{max}$ or $r>r_\textrm{min}$, respectively. In case 1), since we are using the backward integration scheme for the GHZ transport equations, the corrector $x_{ab}$ also vanishes for $r>r_\textrm{max}$, 
and therefore so does $S_{ab}$. Likewise, in case 2), since we are using the forward integration scheme for the GHZ transport equations, the corrector $x_{ab}$ also vanishes for $r<r_\textrm{min}$, 
and therefore so does $S_{ab}$. 

Consider \eqref{eq:S1} which can be seen as a second order ODE in $r$ in retarded KN coordinates and the Kinnersley frame, because $\thorn = \partial_r$ -- we have already used this many times. 
Integrating this ODE in the regions $r>r_\textrm{max}$ respectively $r>r_\textrm{min}$, we have $\eta = p_\textrm{out}^\circ \rho + k^\circ_\textrm{out}$ respectively $\eta = p_\textrm{in}^\circ \rho + k^\circ_\textrm{in}$. 

We now drop the subscripts ``in'' and ``out'', and we substitute equation \eqref{s1'} into equations \eqref{eq:S2} and 
\eqref{eq:S3} with $S_{nn}=S_{\bar m n} = 0$. Then the following equations \eqref{eq:etaconstraints1}
follow when we convert the GHP operators to Held's operators 
and subsequently use the relations of Held's calculus, see \ref{sec:Held}: 
    \begin{equation}
    \label{eq:etaconstraints1}
\begin{split}
0 =&  \frac{1}{4} \rho \bar{\rho}^2 (\tilde{\eth} \, p^\circ - \Omega^\circ \tilde{\eth} \, k^\circ),\\
0 =&  \Re \left( \half \bar{\rho}^2 \tilde{\eth} \, \tilde{\eth} \, k^\circ + \half \rho \bar{\rho}^2 \tilde{\eth} \, \tilde{\eth} \, p^\circ - \rho \bar{\rho}^2 \tau^\circ \tilde{\eth} \, k^\circ \right).
\end{split}
\end{equation}
From the first of these equations we get $\tilde{\eth} \, p^\circ = \Omega^\circ \tilde{\eth} \, k^\circ$. Substituting this into the second equation gives 
\begin{equation}
0 =  \Re \left( \half \bar{\rho}^2 \tilde{\eth}\,  \tilde{\eth} \,  k^\circ + \half \Omega^\circ \rho \bar{\rho}^2 \tilde{\eth}\,  \tilde{\eth} \,  k^\circ  \right) = \half \Re \left( \rho \bar{\rho} \, \tilde{\eth} \,  \tilde{\eth} \,  k^\circ \right) = \half \rho \bar{\rho} \Re \left( \tilde{\eth} \, \tilde{\eth} \,  k^\circ \right).
\end{equation}
This gives $0=\Re \left( \tilde{\eth} \tilde{\eth} k^\circ\right)$.
\end{proof}

We now need to analyze in more detail 
what forms $p^\circ, k^\circ$ in 
\eqref{s1'} may take. This analysis is very easy in case 2) but  more complicated in case 1), where we need to appeal to the full arsenal of results from \ref{sec:zetaproof}. The outcome of the analysis is summarized in the following theorem:

\begin{theorem}
\label{thm:4} 
1) (Backward integration scheme) Suppose $T_{ab}=0$ for $r>r_\textrm{max}$ and suppose the assumptions of theorem \ref{thm:2} hold. Then $\eta \circeq \GHPw{-4}{0}$ 
in the adjoint Teukolsky equation $\O^\dagger \Phi = \eta$ 
is the unique solution to \eqref{eq:S1} which vanishes for $r>r_\textrm{max}$. 
The corrector $x_{ab}$ vanishes for $r>r_\textrm{max}$, too.
    
2) (Forward integration scheme) Suppose $T_{ab}=0$ for $r<r_\textrm{min}$ and suppose the assumptions of theorem \ref{thm:3} hold. Then  $\eta \circeq \GHPw{-4}{0}$ 
in the adjoint Teukolsky equation $\O^\dagger \Phi = \eta$
is the unique solution to \eqref{eq:S1} which vanishes for $r<r_\textrm{min}$. 
The corrector $x_{ab}$ vanishes for $r<r_\textrm{min}$, too.
\end{theorem}

\begin{proof}
\medskip
\noindent
{\bf Case 1)} ($r>r_{\rm max}$). In this case, the Hertz potential $\Phi$ was constructed in Sec. \ref{Hertzbackward}
by backward integration from $\sI^+$, where it was shown to have an expansion of the form 
given by  \eqref{eq:phiexp}. The coefficients in this expansion were given by 
 \eqref{eq:phivalues}. Inserting  \eqref{eq:phiexp} into $\O^\dagger \Phi = \eta$, and 
 using Held's operators and their relations (see \ref{sec:Held}), the expansion of the 
$\eta$ is found to be
\begin{equation}
\label{eq:530}
\begin{split}
\eta =& -2 \bar{\rho} \tilde{\eth}' \tilde{\eth} \Phi^{1\circ} -2 \frac{\bar{\rho}}{\rho} \tilde{\eth}' \tilde{\eth} \Phi^{2\circ} -2 \frac{\bar{\rho}}{\rho^2} \tilde{\eth}' \tilde{\eth} \Phi^{3\circ} + 6 \rho \tilde{\thorn}' \Phi^{0\circ} + 4 \tilde{\thorn}' \Phi^{1\circ} + \frac{2}{\rho} \tilde{\thorn}' \Phi^{2\circ}\\
&- 2 \frac{\bar{\rho}}{\rho} (5 \rho + \bar{\rho}) \bar{\tau}^\circ \tilde{\eth} \Phi^{2\circ} - 2 \frac{\bar{\rho}}{\rho^2} (7 \rho + \bar{\rho}) \bar{\tau}^\circ \tilde{\eth} \Phi^{3\circ} - 4 \rho \bar{\rho}^{\prime \circ} \Phi^{1\circ} - 2 (\rho^{\prime \circ} + \bar{\rho}^{\prime \circ}) \Phi^{2\circ}\\
&- 2 (\rho \Psi^\circ - \bar{\rho} \bar{\Psi}^\circ) \Phi^{2\circ} - 6 \Psi^\circ \Phi^{3\circ} + O(\rho^2).
\end{split}
\end{equation}
We now compare this expansion with that given in lemma \ref{lem:pk} (for $r>r_{\rm max}$). One may verify that the leading order ($\rho^{-1}$) gives $\tilde{\thorn}' \Phi^{2\circ} = \tilde{\eth}' \tilde{\eth} \Phi^{3\circ} \Rightarrow \tilde{\eth}' (2 \tilde{\thorn}' \xi^\circ_{\bar{m}} + \tilde{\eth} e^\circ) = 0$ which we already know is identically satisfied in view of \eqref{eq:ecircdet}. The first and second subleading orders ($\rho^0$ and $\rho^1$) are seen to be precisely equal to the GHP scalars $p^\circ$ and $k^\circ$ considered in the proof of lemma \ref{lem:9}, see \eqref{Kcircdef}, \eqref{Pcircdef}, after a very lengthy calculation using identities of lemma \ref{lem:11}. 
However, we show in the proof of lemma \ref{lem:9} that $k^\circ = p^\circ = 0$, by \eqref{eq:zm}, \eqref{eq:f}, \eqref{eq:d}.  
Therefore, lemma \ref{lem:pk}
shows that $\eta = 0$ in the region $r>r_{\rm max}$, as we wanted to prove.

\medskip
\noindent
{\bf Case 2)} ($r<r_{\rm min}$). In this case, the Hertz potential $\Phi$, corrector $x_{ab}$, and modified source 
$S_{ab}$ were constructed in Sec. \ref{sec:x_abPhi-} by forward integration from $\sH^-$.
We can conclude from the nature of this integration scheme that $S_{ab}, x_{ab}$ vanish in an open neighborhood of $\sH^-$. 
Furthermore, since $h_{ab}, \nabla_a h_{bc}$ are required to vanish at $\sH^-$ by our Standing Assumptions in Sec. \ref{GHZproofH-}, 
integrating \eqref{eq:hmbmb1} forwards from $\sH^-$ shows that $\Phi$ can be extended across $\sH^-$ by $0$ to the past of $\sH^-$ as a three times 
continuously differentiable function. Consequently, by $\O^\dagger \Phi = \eta$, $\eta$ can be extended by $0$ to the past of $\sH^-$ as a  
continuously differentiable function. However, for $r<r_{\rm min}$, $\eta$ is also given by the formulas of 
lemma \ref{lem:pk}, which together implies $p^\circ = k^\circ = 0$. Thus, $\eta=0$ for $r<r_{\rm min}$, as we wanted to prove.

\end{proof}

\medskip

The following corollary \ref{GHZcor} follows immediately from theorem \ref{thm:4}. The statement is usually taken for granted in the literature. Reference \cite{Ori:2002uv} has given an argument for the special case of mode solutions (i.e. having time-dependence $e^{-i\omega t}$ in BL coordinates). 
 Another precursor of corollary \ref{GHZcor} is due to \cite{Prabhu:2018jvy}, who have argued that solutions of the reconstructed form in the corollary must be 
dense in the set of all solutions satisfying suitable asymptotic conditions in some $L^2$-norm defined in terms of initial data. 

\begin{corollary}
\label{GHZcor} 
    Let $(\E h)_{ab} = 0$ in the exterior region of Kerr, and assume that $h_{ab}$ satisfies either the Standing  Assumptions at $\sI^+$ respectively at  $\sH^-$, as assumed in theorems 
    \ref{thm:2} respectively \ref{thm:3}. Then there exists a 
    Hertz potential $\Phi \circeq \GHPw{-4}{0}$ solving $\O^\dagger \Phi = 0$, 
    a zero mode\footnote{This zero mode is necessarily trivial in the situation covered by theorem \ref{thm:3}.} $\dot g_{ab}$, and a gauge vector field $X^a$ such that 
    \begin{equation}
        h_{ab} = \Re(\S^\dagger \Phi)_{ab} + \dot g_{ab} + (\mathcal{L}_X g)_{ab}.
    \end{equation}
\end{corollary}

Let us now clarify the relationship between the equation
$\O^\dagger \Phi = \eta$ in theorem~\ref{thm:4} and the source appearing in the inhomogeneous 
Teukolsky equation for the perturbed Weyl scalar $\psi_0$. For this purpose, we apply the SEOT identity \eqref{eq:SEOT} to the 
linearlized inhomogeneous EE $\E h_{ab} = T_{ab}$. This gives $\O \psi_0 = {}_{+2} T$, in which 
$\psi_0 = \T(h)$ is the perturbed 0-Weyl-scalar
\eqref{eq:Weylcomp} and in which ${}_{+2}T=\S^{ab}T_{ab}$ is the source in Teukolsky's
equation $\O \psi_0 = {}_{+2}T$.

Now, we substitute for $h_{ab}$ the GHZ decomposition
$h_{ab} = x_{ab} + \Re(\S^\dagger \Phi)_{ab} + \lie{X} g_{ab} + \dot g_{ab}$ of
theorem~\ref{thm:4}. Since the perturbed Weyl scalar $\psi_0$ is
gauge invariant on any background in which the corresponding
background scalar $\Psi_0=0$, it follows that $\T(\lie{X} g) =
0$. Furthermore, $\T(x)=0$, since the corrector field $x_{ab}$ has
vanishing $lm,ll,mm$ components [see \eqref{eq:T},\eqref{eq:xdef}]. Finally, 
$\T(\dot g)=0$, since $\Psi_0$ remains zero for perturbations to other Kerr black holes.
Therefore, $\psi_0=\T \Re \S^\dagger \Phi$.  On the other hand, in the
proof of theorem~\ref{thm:4}, we have seen that
$\Re (\T^\dagger \eta)_{ab} = S_{ab} = T_{ab} - (\E x)_{ab}$.
Applying $\S$, as given in \eqref{eq:S}, to this equation and using
$\S\E x = \O \T x = 0$ in view of \eqref{eq:SEOT}, we find
${}_{+2}T=\S \Re \T^\dagger \eta$. The expressions can be further
simplified using the covariant form of the TS identities (see \ref{sec:TScov}), leading to the following theorem.

\begin{theorem}
\label{thm:17}
Let $\psi_0$ be the perturbed Weyl scalar associated with a metric perturbation  $h_{ab} = x_{ab} + \dot g_{ab} + \Re(\S^\dagger \Phi)_{ab} + \lie{X} g_{ab}$  
as in theorems  \ref{thm:2} respectively \ref{thm:3}, where is $\Phi$ satisfying $\O^\dagger \Phi = \eta$ with a source $\eta$ as described in theorem~\ref{thm:4}. Then 
\begin{equation}
\label{eq:Opsi0}
\O \psi_0 = {}_{+2} T, 
\end{equation}
where
\begin{subequations}
\begin{align}
\label{eq:eta1}
{}_{+2} T &= 
- \frac{1}{4} \rho^4 \bar{\rho} \thorn \left[ \frac{\bar{\rho}^2}{\rho^4} \thorn \left( \frac{1}{\bar{\rho}^3} \thorn^2 \bar{\eta} \right) \right]
, \\
\label{eq:Phi}
\psi_0 &=  -\frac{1}{4} \thorn^4 \bar{\Phi}.
\end{align}
\end{subequations}
\end{theorem} 
\noindent
{\bf Remark 4.} Theorem \ref{thm:17} may be read as a prescription for finding the Hertz potential $\Phi$. Either we first solve $\O\psi_0={}_{+2} T$ for $\psi_0$, and then get $\Phi$ from the inversion relation \eqref{eq:Phi} and the ``boundary conditions'' on $\Phi$ given in theorems \ref{thm:2} respectively \ref{thm:3}, or we solve either one of the transport equations \eqref{eq:eta1} or \eqref{eq:S1} for $\eta$, with the ``boundary conditions'' specified by theorem \ref{thm:4}, and then solve $\O^\dagger \Phi = \eta$ for $\Phi$. We will explain the first strategy in more detail in the context of the Cauchy problem in Sec. \ref{sec:Cauchy}.

\subsubsection{Energy flux  and gravitational memory formulas in terms of $\Phi$}
\label{sec:energyflux}

Consider again a metric perturbation $h_{ab}$ satisfying the sourced linearized EE $({\mathcal E}h)_{ab} = T_{ab}$ 
and the decay assumptions at $\sI^+$ of theorem \ref{thm:2}.
Recall that the linearized Bondi news tensor $N_{AB}$ at $\sI^+$ is defined in terms of Bondi coordinates $(\tilde u, \tilde r, \tilde x^A)$, see \ref{BondiIRG}. The 
formulas given there lead to 
\begin{equation}
\label{NAB}
    N_{AB} = 2 \, \Re \left( \rho^2 \bar m_{A} \bar m_{B} \, N^\circ \right)_{\sI^+} \circeq \GHPw{0}{0},
\end{equation}
in terms of the GHP scalar $N^\circ = -\tilde{\thorn}' h_{mm}^{1\circ} \circeq \GHPw{0}{-4}$. We now express the news tensor in terms of the Hertz potential $\Phi$ appearing in the GHZ decomposition 
$
h_{ab} = x_{ab} + \dot g_{ab} + (\lie{X} g)_{ab} + \Re(\S^\dagger \Phi)_{ab}   
$,
obtained by the ``backward integration scheme'', see theorem \ref{thm:2}.  According to this theorem, $\Phi$ behaves near $\sI^+$ as 
\begin{equation}
\Phi 
= \frac{\Phi^\circ}{\rho^3} + O\left(\frac{1}{\rho^2}\right),
\end{equation}
for some $\Phi^\circ \circeq \GHPw{-1}{3}$.
We know, by \eqref{eq:e}, \eqref{eq:phivaluesd}, that $-2\tilde{\thorn}' \bar f^\circ = \Phi^\circ$, and we know, by lemmas \ref{lem:9}
and \ref{lem:11}, that $N^\circ$ and $f^\circ$ are related by \eqref{eq:fN1}, \eqref{eq:fN2}. Furthermore, when acting on GHP scalars annihilated by $\th$, Held's operator 
$\tilde \th'$ in retarded KN coordinates and the Kinnersley frame is $\tilde \th' = \partial_u$ (see table \ref{tab:2} in \ref{sec:Held}), which is equivalent to multiplication by $(-i\omega)$ 
when acting on a mode (see Sec. \ref{sec:modesoln}). 

With the help of these relations, we may express quantities related to gravitational radiation carried by $h_{ab}$ at $\sI^+$ in terms of the leading order contribution $\Phi^\circ$ of the Hertz potential. We now give some examples. 

First, we consider the total flux of energy\footnote{That this flux is finite follows per assumption on the news tensor in theorem \ref{thm:2}.}, $F_{\sI^+}$, through $\sI^+$, assuming for simplicity that there is no 
matter contribution to this flux which is the case if $T_{nn} = O(r^{-3})$ (otherwise, this contribution is simply added on). 

Using the well-known formula \eqref{eq:Flux} for $F_{\sI^+}$ in terms of $N_{AB}$ recalled
in \ref{BondiIRG}, and using \eqref{eq:fN1}, we obtain
\begin{equation}
   F_{\sI^+} = -\frac{\pi}{4}
    \sum_{\ell,m} \int\limits_{-\infty}^\infty
    \frac{\dd \omega}{\omega^2}
    \bigg| {}_2B_{\ell m}(a\omega) \Phi^\circ_{\ell m}(\omega) + 3iM\omega(-1)^m \, \overline{ \Phi^\circ_{\ell -m}(-\omega)} \bigg|^2,
\end{equation}
where ${}_2B_{\ell m}(a\omega)$ is the angular TS constant, and where $\Phi^\circ_{\ell m}(\omega)$ are the mode amplitudes of $\Phi^\circ$ (see Sec. \ref{sec:modesoln}). In the presence of a matter stress tensor one gets an additional non-trivial energy flux at first perturbation order proportional given by the integral of $T^{2\circ}_{nn}$ along $\sI^+$.

Vice versa, from equation \eqref{eq:fN2}, we may also obtain a formula for the leading term $\Phi^\circ/\rho^3$ in the Hertz potential in terms of the mode coefficients of the Bondi news tensor: 
\begin{equation}
\begin{split}
   \Phi^\circ(u,\theta,\varphi_*) = \sum_{\ell,m} \int\limits_{-\infty}^\infty & \frac{\dd \omega}{\omega}
   \frac{{}_2B_{\ell m}(a\omega) \overline {N_{\ell -m}^\circ(-\omega)}
-3iM\omega(-1)^m  N^\circ_{\ell m}(\omega)}{-2i[{}_2B_{\ell m}(a\omega)^2+9M^2\omega^2]}\\
& \times
 {}_{-2}S_{\ell m}(\theta,\varphi_*;a\omega) e^{-i\omega u}, 
 \end{split}
\end{equation}
where ${}_{-2}S_{\ell m}(\theta,\varphi_*;a\omega)$ are the spin-weighted spheroidal harmonics, see Sec. \ref{sec:modesoln}.

Next, we turn to the memory tensor, $\Delta_{AB}$, at linear order \eqref{eq:Memory}. It is parameterized by a GHP scalar 
 $\Delta^\circ \circeq \GHPw{0}{-4}$ via a formula analogous to \eqref{NAB} i.e., 
 \begin{equation}
\label{DAB}
    \Delta_{AB} = 2 \, \Re \left( \rho^2 \bar m_{A} \bar m_{B} \, \Delta^\circ \right)_{\sI^+} ,
\end{equation}
 so that $\Delta^\circ = \int_{-\infty}^\infty N^\circ \dd u$. 
 Consequently, in view of $-2\tilde{\thorn}' \bar f^\circ = \Phi^\circ$,
where $N^\circ$ and $f^\circ$ are related by equations \eqref{eq:fN1}, it follows using the value ${}_2B_{\ell m}(a\omega=0)=\frac{1}{4}(\ell-1)\ell(\ell+1)(\ell+2)$ and 
$ {}_{2}S_{\ell m}(\theta,\varphi_*;a\omega=0)=: {}_{2}Y_{\ell m}(\theta,\varphi_*)$ (indicating the spin-weighted spherical harmonics), that
\begin{equation}
   \Delta^\circ(\theta,\varphi_*) = \frac{\pi}{4i} \sum_{\ell,m} (\ell-1)\ell(\ell+1)(\ell+2)
   \bigg(
   \frac{
    \overline {\Phi_{\ell -m}^\circ(\omega)}}{\omega}
\bigg)_{\omega=0}
 {}_{2}Y_{\ell m}(\theta,\varphi_*).
\end{equation}
At first perturbative order, the memory tensor is non-zero only for a non-zero total matter energy flux (see e.g., \cite{Hollands:2016oma}), which 
requires that $T_{nn} = O(r^{-2})$ and that $T^{2\circ}_{nn}$ is non-zero somewhere.

As is well-known (see e.g., \cite{Hollands:2016oma}), the memory tensor is related to a BMS supertranslation $u \to u+u(\theta,\varphi_*)$ at $\sI^+$ by 
the formula 
\begin{equation}
  \Delta_{AB} = D_A D_B u- \half s_{AB} D^C D_C u.  
\end{equation} 
Using properties of spin-weighted spherical harmonics, 
$u$ is therefore  
\begin{equation}
   u(\theta,\varphi_*) = \frac{\pi}{4i} \sum_{\ell,m} (\ell-1)\ell
   \bigg(
   \frac{
    \overline {\Phi_{\ell -m}^\circ(\omega)}}{\omega}
\bigg)_{\omega=0}
Y_{\ell m}(\theta,\varphi_*).
\end{equation}
One may also obtain formulas for the Bondi mass aspect and Bondi angular momentum aspect using formulas \eqref{eq:bondimu} and \eqref{eq:Bondij}. 
Such formulas will in general not only involve $\Phi$ but also the corrector, $x_{ab}$.

\subsection{Solving the Cauchy problem with the GHZ decomposition}
\label{sec:Cauchy}

\subsubsection{Detailed description}

We now show how to extract a convenient algorithm for solving the Cauchy problem for the linearized EE from the the GHZ decomposition developed in Sec. \ref{GHZproof}. 

\usetikzlibrary{decorations.pathmorphing}
\tikzset{zigzag/.style={decorate, decoration=zigzag}}
\begin{figure} 
\begin{center}
\begin{tikzpicture}[scale=0.6, transform shape]
\draw[thick](3,-3)--(0,0)--(3,3);
\draw[blue, ultra thick](0,0)--(6,0);
\draw[red, ultra thick](1.5,1.5)--(4.5,1.5);
\node[anchor=north] at(3,-.3) {{\Large \color{blue} $\mathscr{S}$}};
\node[anchor=north] at(3,1.2) {{\Large \color{red} $\mathscr{S}$}};

\draw[double](3,3)--(6,0)--(3,-3);

\node[anchor=west] at(4.3,2.1) {{\Large ${\mathscr I}^+$}};
\draw (0,0) node[draw,shape=circle,scale=0.4,fill=black]{};
\node[anchor=east] at(1.9,2.1) {{\Large ${\mathscr H}^+$}};

\node[anchor=north] at(3.1,3.8){{\Large $i^+$}};
\node[anchor=west] at(6.2,0) {{\Large $i^0$}};

\draw (3,3) node[draw,shape=circle,scale=0.5,fill=white]{};
\draw (6,0) node[draw,shape=circle,scale=0.5,fill=white]{};
\draw (3,-3) node[draw,shape=circle,scale=0.5,fill=white]{};

\end{tikzpicture}
\end{center}
\caption{\label{fig:5} Examples of Cauchy surfaces $\sS$.
The red surface represents a ``hyperboloidal'' slice and is 
a Cauchy surface only for its domain of dependence.}
\end{figure}
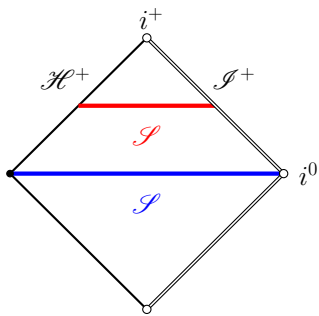


Since the linearized EE equation implies certain constraints on the Cauchy data, we first recall the notion of an initial data set, following the ADM formalism, see e.g., \cite{Waldbook}. The ADM decomposition of a metric 
is 
\begin{equation}
g^{ab} = N^{-2}(t^a - N^a)(t^b - N^b) - q^{ab},
\end{equation}
where $-q^{ab}$ is a Riemannian metric on a Cauchy surface $\sS$ (see Fig. \ref{fig:5}), $t^a=N\nu^a+N^a$ is the decomposition of $t^a = (\partial_t)^a$ into lapse and shift, and $\nu^a$ is 
the forward timelike normal to $\sS$. The ADM momentum is
$p^{ab} = \sqrt{-q}(K^{ab} - \frac{1}{2} q^{ab}K),$
where $K_{ab}$ is the extrinsic curvature of $\sS$. The Cauchy data in the full non-linear theory are $(q_{ab}, p^{ab})$. In the linear theory, 
we linearize the Cauchy data around their background values in Kerr. The linearized Cauchy data are $(\delta q_{ab}, \delta p^{ab})$, where we also set
$\delta N = 0 = \delta N^a$, referred to as ``ADM gauge''. $(\delta q_{ab}, \delta p^{ab})$ must satisfy the linearized constraints given e.g., 
in \cite{HWcanonicalenergy}. 

For simplicity, we will assume that $(\delta q_{ab}, \delta p^{ab})$ is smooth and vanishes\footnote{
\label{footnote:1} This assumption is made for simplicity. The existence of 
a wide class of such data is guaranteed by the gluing theorem \cite{CorvinoSchoen}.} for points in $\sS$ such that $r>r_{\rm max}$, for some $r_{\rm max}$. We also assume that $T_{ab}$ satisfies the Standing Decay Assumptions at $\sI^+$ (see Sec. \ref{GHZproof}), and additionally that the support of $T_{ab}$   is contained in $J(\sS \cap \{ r \le r_{\rm max} \})$, where $J=J^+ \cup J^-$ is the union of the causal past and future of a set, see e.g., \cite{Waldbook}.

We want to study the 

\medskip
\noindent
{\bf Cauchy problem:} To find $h_{ab}$ satisfying $(\E h)_{ab}=T_{ab}$ for given $T_{ab}$ and for given Cauchy data $(\delta q_{ab}, \delta p^{ab})$ on $\sS$. 

\medskip
The existence, smoothness, and uniqueness up to gauge for solutions of the Cauchy problem is not in question, see e.g., \cite{Ringstrom}.
The problem we address here is how to concretely find a solution to the Cauchy problem using the backward GHZ integration scheme
in Sec. \ref{GHZproof}, i.e., we seek the solution in the form\footnote{Note that the zero mode can be omitted since we are assuming that the initial data vanish near spatial infinity, so the perturbation will automatically have $\delta M = \delta a = 0$.}
\begin{equation}
\label{GHZ:short}
h_{ab} = \Re \S^\dagger_{ab} \Phi + x_{ab} +(\lie{X}g)_{ab} ,
\end{equation}
where $x_{ab}, \Phi$ are to be determined in such a way that the Cauchy problem is solved up to gauge (the last term parameterized by $X^a$). 

An awkward point about this strategy is that theorem \ref{thm:2}, to which we would like to appeal here, assumed certain decay conditions on $h_{ab}$ towards $\sI^+$, as well as along $\sI^+$. Clearly, we are given only the Cauchy data and the stress tensor $T_{ab}$, so it would seem that
we cannot decide whether or not these hold without actually solving the Cauchy problem in the first place! Fortunately, the situation is not as bad as it may seem, 
for the following reasons. Firstly, our conditions for the initial data imply that the solution is asymptotically flat at $\sI^+$ e.g., by \cite{GX},  
so the fall-off conditions \eqref{eq:hdec} towards $\sI^+$ hold in a suitable gauge such as Bondi gauge, see \ref{BondiIRG}. Secondly, by the usual domain of dependence property for the Cauchy problem (see e.g., \cite{Ringstrom}), we know that a gauge exists in which $h_{ab}$ vanishes in an open neighborhood
of $i^0$, because the initial data vanish there. Thus, the only question is if the news tensor $N_{AB}$ has the behavior 
as $u \to \infty$ assumed in theorem \ref{thm:2}, which appears consistent with recent results about the linear stability of Kerr, see remark 2 below theorem \ref{thm:2}. 

In this section, we will take the assumption about $N_{AB}$ in theorem \ref{thm:2} for granted, but note that we could enforce it by simply replacing $h_{ab} \to \chi h_{ab}$ where $\chi$ is a window function that vanishes in a small neighborhood\footnote{More precisely, let $\tau$ be some time function that is equal to retarded KN time $u$ for $r>r_{\rm max}$ and equal to advanced KN time $v$ for $r_+ \le r < r_{\rm min}$. Then $\chi(\tau)=1$ for $\tau>\tau_0$ and, e.g., $\chi(\tau)=0$ for $\tau > 2\tau_0$} of $i^+$. The error term incurred by the derivatives of $\chi$ would effectively be subsumed in a change of $T_{ab}$, which, furthermore, would not affect our actual construction of $h_{ab}$
except for an arbitrarily small neighborhood of $\sH^+$ and of $i^+$. We will comment more on this in Sec. \ref{sec:concl} but will leave a full discussion of the limiting behavior of the solution as $\chi \to 1$, i.e., as the cutoff is removed, to a future investigation.

Thus, we now simply assert that the solution $h_{ab}$ of the Cauchy problem has a GHZ decomposition as described in theorem \ref{thm:2}, which we assume can be found using the backward integration procedure described in Sec. \ref{GHZproof}. Following this procedure, we first determine the corrector $x_{ab}$ by solving the GHZ transport equation with ``trivial'' final conditions at $\sI^+$. In practice, this may be done e.g., by applying the advanced 
Green's function for the corrector given in \cite{Casals:2024ynr} to $T_{ab}$. If necessary, we need to adjust the leading coefficient of $x_{ab}$ in the expansion in $r^{-1}$ so that \eqref{eq:xreadj} holds.

To find $\Phi$, we first consider the Teukolsky equation \eqref{eq:Opsi0} (see theorem \ref{thm:17}) for $\psi_0$ with source ${}_{+2} T = \S^{ab}T_{ab}$. We need to know 
what are the Cauchy data for $\psi_0$ on $\sS$, i.e., what is $(\psi_0, {\mathcal L}_t \psi_0)$ on $\sS$. 
$\psi_0$ is gauge invariant, so it can be calculated from $h_{ab}$ in any gauge, including the ADM gauge 
that we use to specify the Cauchy data for $h_{ab}$. Since $\psi_0$ contains up to two derivatives of $h_{ab}$ off of $\sS$, we need, in ADM gauge, the 
induced metric and up to three time-derivatives, i.e., $\delta q_{ab}, {\mathcal L}_t \delta q_{ab}, \dots, {\mathcal L}_t^3 \delta q_{ab}$ on $\sS$. As is well known, for the full metric, 
we have
\begin{equation}
{\mathcal L}_t q_{ab} = N\left(p_{ab} - \half p q_{ab}\right)/\sqrt{-q} + 2D_{(a} N_{b)}, 
\end{equation}
the linearization of which yields ${\mathcal L}_t \delta q_{ab}$ from the Cauchy data. The ADM evolution equation (see e.g., \cite{Waldbook}) 
likewise gives an equation for ${\mathcal L}_t p^{ab}$, and thereby, upon linearization, ${\mathcal L}_t^2 \delta q_{ab}$ from the Cauchy data. 
Finally, iterating the ADM evolution operator once more, we obtain  ${\mathcal L}_t^3 \delta q_{ab}$ from the Cauchy data. Thus, we may 
compute $(\psi_0, {\mathcal L}_t \psi_0)$ restricted to $\sS$, which comprise the Cauchy data for $\psi_0$. 
Many formulas that are useful for such laborious, but in principle straightforward, computations
are provided e.g., in \cite{Campanelli}. With $(\psi_0, {\mathcal L}_t \psi_0)$ at hand, 
we may solve the Cauchy problem for 
$\O \psi = {}_{+2} T$. In practice, this would be done most conveniently in the frequency domain, exploiting the separability of the Teukolsky operator \cite{Teukolsky:1972my}. Either way, it is now clear in principle how to obtain $\psi_0$ from $T_{ab}$ and the Cauchy data for $h_{ab}$ in ADM form. 

However, for \eqref{GHZ:short}, we require the Hertz potential $\Phi$, not $\psi_0$. The two are related 
by $\psi_0 =  -\frac{1}{4} \thorn^4 \bar{\Phi}$ in view of theorem \ref{thm:17}. The operator $\thorn^4$ has as 
its kernel the degree three polynomials in the GHP scalar $1/\rho$, with arbitrary coefficients of the appropriate 
GHP weights that are annihilated by $\thorn$. This kernel hence corresponds precisely to the leading 
terms of $\Phi$ in its asymptotic expansion \eqref{eq:phiexp} near $\sI^+$, see lemma \ref{lem:11}. 
By this lemma, the leading coefficients in \eqref{eq:phiexp} have to take the values \eqref{eq:phivalues}.
Thus, the kernel of $\thorn^4$ has to be used precisely to fix these values \eqref{eq:phivalues}, and 
therefore what remains for us is to find them. 

To find the values \eqref{eq:phivalues}, we require $(d^\circ, e^\circ, \zeta^\circ_l, \zeta^\circ_m, h^{1\circ}_{mm})$. 
However, by the formulas stated in lemma \ref{lem:9}, once we know the master scalar $f^\circ$, we can find $(d^\circ, e^\circ, \zeta^\circ_l, \zeta^\circ_m)$ quite easily. Thus, we only require $(f^\circ, h^{1\circ}_{mm})$, and these must now be extracted from $\psi_0$. 

By the peeling theorem, which follows from asymptotic flatness and the EE\footnote{One needs to assume 
suitable conditions on $T_{ab}$ as in our Standing Decay Assumptions at $\sI^+$. The peeling theorem also follows from the 
formulas given in \ref{BondiIRG}.} \cite{GerochAS}, we have 
\begin{equation}
\label{Teukasy}
 \psi_0 = \rho^5 \psi_0^{5\circ} + O(\rho^6), \quad \psi_4 = \rho \psi^{1\circ}_4 + O(\rho^2),
 \end{equation}
 near $\sI^+$. The GHP scalars $\psi^{5\circ}_0 \circeq \GHPw{-1}{-5} , \psi^{1\circ}_4 \circeq \GHPw{-5}{-1}$ annihilated by $\thorn$
 are related by the TS relations, expressible in Held form by
 \begin{equation}
 \tilde{\thorn}^{\prime} \tilde{\thorn}^{\prime} \tilde{\thorn}^{\prime} \tilde{\thorn}^{\prime}  \psi_0^{5\circ} = \tilde{\eth} \: \tilde{\eth}  \: \tilde{\eth}\: \tilde{\eth} \, \psi^{1\circ}_4 + 3\Psi^\circ \tilde{\thorn}' \bar \psi^{1\circ}_4. 
\end{equation}
(This follows e.g., by applying the rules
 of Held's calculus in \ref{sec:Held} to the relation in proposition \ref{prop:3'}, using the assumed asymptotic behavior of $T_{ab}$.)
Furthermore, from the definition of $\psi_4$, we also have
\begin{equation}
\psi^{1\circ}_4 = \frac{1}{2} \tilde{\thorn}^{\prime} \tilde{\thorn}^{\prime} h_{\bar m \bar m}^{1\circ}.
\end{equation}
Using the same reasoning as in the proof of proposition \ref{prop:13}, we can determine  $(f^\circ, h^{1\circ}_{mm})$
from these equations in terms of the modes of $\psi_0^\circ$ as:
\begin{equation}
\label{Hmode}
\begin{split}
   h^{1\circ}_{\bar m\bar m}(u,\theta,\varphi_*) =  \sum_{\ell,m} \int\limits_{-\infty}^\infty \dd \omega \, \omega^2  & \, 
   \frac{ {}_2B_{\ell m}(a\omega) \overline {\psi_{0\ell -m}^{5\circ}(-\omega)}
+3iM\omega(-1)^m  \psi^{5\circ}_{0\ell m}(\omega) }{{}_2B_{\ell m}(a\omega)^2+9M^2\omega^2} \\
& \times
 {}_{-2}S_{\ell m}(\theta,\varphi_*;a\omega) e^{-i\omega u}, 
 \end{split}
\end{equation}
where ${}_2B_{\ell m}(a\omega)$ is the angular TS constant, 
${}_{-2}S_{\ell m}(\theta,\varphi_*;a\omega)$ are the spin-weighted spheroidal harmonics, see Sec. \ref{sec:modesoln}, and 
\begin{equation}
\label{Fmode}
   f^\circ(u,\theta,\varphi_*) =  2i \sum_{\ell,m} \int\limits_{-\infty}^\infty \dd \omega
   \, \omega^3 \, \psi^{5\circ}_{0\ell m}(\omega) \, 
 {}_{-2}S_{\ell m}(\theta,\varphi_*;a\omega) e^{-i\omega u}. 
\end{equation}

\subsubsection{Summary of procedure}
To summarize, the GHZ scheme for the Cauchy problem is as follows.

\medskip
\noindent
{\bf Step 1.}
Determine the corrector $x_{ab}$ solving the GHZ transport equations with ``trivial'' final conditions at $\sI^+$, e.g., by applying the advanced Green's function for the corrector \cite{Casals:2024ynr} to $T_{ab}$. Then adjust the 
leading coefficient of $x_{ab}$ in the expansion in $r^{-1}$
so that \eqref{eq:xreadj} holds.

\medskip
\noindent
{\bf Step 2.}
Obtain $\delta q_{ab}, {\mathcal L}_t \delta q_{ab}, \dots, {\mathcal L}_t^3 \delta q_{ab}$ on $\sS$ from the Cauchy data using the linearization of the ADM evolution equations. From this, compute the Cauchy data $(\psi_0, {\mathcal L}_t \psi_0)$ for $\psi_0$.

\medskip
\noindent
{\bf Step 3.}
Solve the Cauchy problem for the sourced Teukolsky equation $\O \psi_0 = \S^{ab} T_{ab}$, where $\S$ is Teukolsky's source operator, see \ref{sec:TScov}.

\medskip
\noindent
{\bf Step 4.}
Obtain $f^\circ$ and $h^{1\circ}_{mm}$ via \eqref{Fmode}, \eqref{Hmode} (which rely on the angular TS identities) from the leading asymptotic coefficient $\psi^\circ_0$ of $\psi_0 = r^{-5} \psi^\circ_0 + O(r^{-6})$. 

\medskip
\noindent
{\bf Step 5.}
Determine $d^\circ$ and $\zeta^\circ_l$ 
from $f^\circ$ solving \eqref{eq:d} and \eqref{eq:zl} [substituting \eqref{eq:zm}], 
which is straightforward in the frequency domain.

\medskip
\noindent
{\bf Step 6.}
Find the unique solution $\Phi$ of the transport equation $\psi_0 =  -\frac{1}{4} \thorn^4 \bar{\Phi}$ such that $\Phi = O(r^{-1})$, in retarded KN coordinates and the Kinnersley frame, near $\sI^+$. This may be accomplished e.g., by viewing 
$\psi_0 =  -\frac{1}{4} \thorn^4 \bar{\Phi}$ as a fourth order transport equation using $\thorn = \partial_r$, or by using a radial TS identity (see e.g., \cite{Ori:2002uv})  
in the frequency domain. Then add:
\begin{equation}
\Phi \to \Phi + d^\circ + \left[ 2 \left(\tilde{\eth}'\tilde{\eth}'  \zeta_l^\circ - 2 \bar{\tau}^\circ \tilde{\eth} \, \bar f^\circ  \right)
+ \bar{h}^{1 \circ}_{mm} 
\right] \frac{1}{\rho}
-2 \left( \tilde{\eth}' \tilde{\eth} \, \bar f^\circ
\right) \frac{1}{\rho^2} -2
\left(\tilde{\thorn}' \bar f^\circ \right) 
\frac{1}{\rho^3}.
\end{equation}

\medskip
\noindent
{\bf Step 7.} The metric perturbation solving the Cauchy problem up to a gauge transformation is  $h_{ab} = \Re \S^\dagger_{ab} \Phi + x_{ab}$.


\section{Lorenz gauge}
\label{sec:PTL}

Various inconvenient features of the IRG were already pointed out in Sec. \ref{sec:gaugeissues}. Here we show how to pass to the Lorenz gauge, which can be superior in some ways, e.g. in as far as the propagation of singularities or the fall-off near $\sI^+$ is concerned. 

\subsection{Reduction to a Maxwell equation}
\label{sec:RedMax}

Consider a solution to the EE $(\E h)_{ab} = T_{ab}$ on some globally hyperbolic background spacetime $(\sM, g_{ab})$ satisfying $R_{ab}=0$. It is well-known (see e.g., \cite{Waldbook}) how to find an infinitesimal gauge transformation 
$h_{ab} \to h_{ab} + (\mathcal L_\zeta g)_{ab}$ putting $h_{ab}$ into the  Lorenz gauge, $\nabla^a(h_{ab}-(1/2)hg_{ab})=0$ (here $h=h_a{}^a$ is the trace with respect to $g_{ab}$):
The equation that the gauge vector field $\zeta^a$ must satisfy is
\begin{equation}
\label{eq:wave}
\nabla^c \nabla_{c} \zeta_{b} = -\nabla^a  \left( h_{ab}- \frac{1}{2} h g_{ab} \right), 
\end{equation}
which is a wave equation. Such equations have a unique solution for any sufficiently regular source and any initial conditions on some Cauchy surface.

The problem is to concretely find such a solution in the exterior of the Kerr black hole. Ideally, one would like to decouple the components of \eqref{eq:wave} and be able to solve the individual decoupled equations by a separation of variables ansatz. Unfortunately, there is no known---to the authors---way of achieving this directly at the level 
of \eqref{eq:wave}. But one can make progress by passing from \eqref{eq:wave} to the sourced Maxwell equation by imposing the Lorenz gauge condition on the gauge vector field, $\nabla^a \zeta_a=0$. More precisely, we consider a decomposition
\begin{equation}
\zeta_a = \eta_a + \nabla_a \chi,
\end{equation}
with $\nabla^a \eta_a = 0$. We may then write 
\begin{equation}
\label{Jdef}
2 \nabla^a (\nabla_{[a} \eta_{b]}) = -\nabla^a \left( 
h_{ab} - \frac{1}{2} h g_{ab} 
\right) - \nabla_b \nabla^c \nabla_c \chi \equiv J_b.
\end{equation}
We must choose $\chi$ in such a way that the right side, $J_a$, becomes divergence free for consistency, which is seen to require that
\begin{equation}
\label{doublechi}
(\nabla^c \nabla_c)^2 \chi = T_c{}^c - \frac{1}{2} \nabla^c \nabla_c h,
\end{equation}
using $R_{ab}=0$ and
\begin{equation}
-\nabla^a \nabla^b \left( 
h_{ab} - \frac{1}{2} h g_{ab} 
\right) = T_c{}^c  - \frac{1}{2} \nabla^c \nabla_c h.
\end{equation}
We now need to solve these equations for the unknown quantities $\chi, \eta_a$. $\chi$ satisfies a double massless wave equation, the solution of which depends on the initial conditions that we wish to choose.  The problem of finding $\chi$ in practice is made feasible 
by the fact that $\nabla^a\nabla_a=\O$ is the separable spin-0 Teukolsky operator 
in Kerr, so the usual frequency domain techniques based on the separability of the Teukolsky equation \cite{Teukolsky:1973ha,Teukolsky:1972my} apply. 
The remaining task is to solve the Maxwell equation \eqref{Jdef} for $\eta_a$ under the divergence free condition $\nabla^a \eta_a = 0$. 
In the next subsection, we will show that also this problem can be reduced to a pair of spin-$1$ Teukolsky equations, and in Sec. 
\ref{sec:MaxGHZ}, we outline how it can alternatively be solved solved by a spin-1 version of the GHZ approach \cite{Hollands:2020vjg}.

\subsection{Solving the sourced Maxwell equation by a pair of potentials}
\label{SourcedMax}

We would like to solve the sourced Maxwell equation 
\begin{equation}
\label{eq:maxwell}
\nabla^a F_{ab} = J_b, \quad F_{ab} = 2 \nabla_{[a} A_{b]}
\end{equation}
in Kerr. For the problem at hand, we would like a solution in Lorenz gauge $\nabla^a A_a = 0$. This gauge may as usual be imposed, if necessary, 
by changing $A_a \to A_a - \nabla_a \eta$, where $\nabla^a \nabla_a \eta = \nabla^a A_a$. The last equation is 
a spin-0 Teukolsky equation in Kerr, which can be solved in the frequency domain by the separation of variables method  \cite{Teukolsky:1973ha,Teukolsky:1972my}.

With this possibility understood, we will from now ignore the Lorenz gauge condition on $A_a$.
We begin by defining a basis of self-dual 
2-forms, $X,Z,Z'$, given by 
\begin{equation}
\label{eq:XZZp}
Z = l \wedge m, \quad X = n \wedge l - \bar m \wedge m, \quad Z' = n \wedge \bar m
\end{equation}
in differential forms notation (for further details on the ``bivector calculus'' with these forms see e.g., \cite{Araneda,Aksteiner:2014zyp,Fayos}, and \ref{sec:bivec}). The Maxwell GHP scalars are
$\phi_0 = F_{ab} Z^{ab}, \phi_1 = \frac{1}{4} F_{ab}X^{ab}, \phi_2 = F_{ab}Z^{\prime ab}$. We also set $\xi^a = M^{-1/3}t^a$, i.e., it is proportional to the asymptotically timelike
Killing vector field $t^a$ ($=(\partial_t)^a$ in BL coordinates). Letting $\S J$ respectively $\S'J$ be the spin-1 Teukolsky sources, see \ref{app:ST}, the Maxwell scalars $\phi_0$, $\phi_2$
satisfy the Teukolsky equations \cite{Teukolsky:1973ha,Teukolsky:1972my}
\begin{equation}
\O \phi_0 = \S J, \quad \O' \phi_2 =  \S' J,
\end{equation}
where the GHP forms of $\O, \O'$ are recalled in \ref{sec:appO}.
Consider now two GHP scalars $\Phi_0 \circeq \{2,0\}, \Phi_2 \circeq \{-2, 0\}$ such that
\begin{equation}
\label{eq:Hertzspin-1}
{\mathcal L}_\xi \Phi_0 = \phi_0, \quad 
{\mathcal L}_\xi \Phi_2 = -\phi_2, 
\end{equation}
where in these expressions ${\mathcal L}_\xi$ is the GHP Lie derivative \cite{Edgar}. 
From this, we define the ``Hertz 2-form'' $H \circeq \{ 0, 0 \}$, as 
\begin{equation}
\label{Hab}
H_{ab} \equiv -\Phi_0 Z_{ab}^{\prime} + \Phi_2 Z^{}_{ab}. 
\end{equation}
Furthermore, from the current $J_a$ in Maxwell's equations, we define the 1-form $G$
as 
 \begin{equation}
 G_a \equiv J^b X_{ba} .
 \end{equation} 
The following theorem is due in this form originally to 
\cite{Green2,Green3}. In related work by \cite{Dolan}, 
an ansatz similar to \eqref{eq:Aansatz} was used. Ref. \cite{Dolan} (extended later by \cite{Dolan2}) has furthermore discovered that, in the case $h_{ab} = \Re \S^\dagger_{ab} \Phi$ is of reconstructed form, their ansatz for finding the gauge vector putting such a $h_{ab}$ into the Lorenz gauge leads non-trivially to a rather economical and elegant algorithm, which only requires taking anti-time-derivatives of certain derivatives of $\Phi$. 

\begin{theorem} 
\label{thm:MaxwellHertz}
The vector potential 
\begin{equation}
\label{eq:Aansatz}
A_a \equiv 2 \ {\rm Re} \left[ \nabla^b (\zeta H_{ab}) + G_a \right],
\end{equation}
where $\zeta = \Psi_2^{-\tfrac{1}{3}}$, 
solves the sourced Maxwell equation $\nabla^a F_{ab} = J_b$ up to a term Lie-derived by $\xi^a$. Thus, if it solves this equation at any point of 
some set $\sS$ in the hole's exterior, intersecting each orbit of $t^a$ ($=(\partial_t)^a$ in BL-coordinates) at least once, then $A_a$ solves the sourced Maxwell equations
everywhere in the hole's exterior.
\end{theorem}

\medskip
\noindent
{\bf Remark 5.}
The ansatz \eqref{eq:Aansatz} gives a practical way of reducing the problem of solving the Maxwell equations to Teukolsky equations in many situations. For example,
if we knew that the current, $J_a$, was supported in the causal future or past of some compact subset $\mathscr{K}$ of the exterior of the Kerr black hole, then we could find the retarded respectively advanced solutions as follows. 

We solve equations \eqref{eq:Hertzspin-1}
imposing that $\Psi_0$ and $\Psi_2$ vanish in the past or future of that set $\mathscr{K}$. It then follows 
that the Maxwell equations are satisfied identically outside the causal future or past of 
some compact subset $\mathscr{K}$ because both $A_a$ and $J_a$, vanish there. Hence, by the 
above theorem, the Maxwell equations hold everywhere. 

We could also solve the 
initial value problem for the Maxwell equation with prescribed Cauchy data 
on some Cauchy surface $\sS$, by ensuring that $\Phi_0, \Phi_2$ are chosen in such a way on $\sS$ that the Maxwell equations hold there. Then, again, this gives rise to a pair of Teukolsky equations for $\Phi_0, \Phi_2$,
with suitable initial data on $\sS$, and by the above theorem, the ansatz 
\eqref{eq:Aansatz} for $A_a$ is a solution to the Cauchy problem with the given data on $\sS$. 

\begin{proof}
Our proof of theorem \ref{thm:MaxwellHertz} is a variant of that by Green \cite{Green2,Green3}. We use differential forms notation throughout: $\intd$ is the exterior differential on $p$-forms, and $\intd^\dagger =\star \intd \star=$ adjoint of $\intd$ (divergence) in the sense of definition \ref{def:adj}, on ($4-p$)-forms. We use a dot $\cdot$ to mean that a vector is contracted into the first index of a differential form. In 
differential forms notation, the ansatz for the gauge potential is
\begin{equation}
A \equiv  \intd^\dagger  (\zeta H) + G + \text{c.c.},
\end{equation}
and
\begin{equation}
\label{LorentzGdef}
G = J \cdot X.
\end{equation}
The Maxwell equation is 
\begin{equation}
\intd^\dagger F = -J.
\end{equation}
We will also use the notation 
\begin{equation}
\begin{split}
B =& -\rho n + \tau \bar{m},\\
C =& -\rho m + \tau l.
\end{split}
\end{equation}
See \ref{sec:bivec} for many identities connecting the  2-forms $Z,Z',X$ and the 1-forms $B,B',C,C'$ used implicitly in this proof.

We know that the source $J$ is divergence free, i.e., $\intd^\dagger J = 0$. Hence,
\begin{equation}\label{LieXiJIdentity}
{\mathcal L}_\xi J = \star {\mathcal L}_\xi \star J  = \star (\xi \cdot \intd \star J) + \star \intd (\xi \cdot \star J) = 
- \intd^\dagger \star (\xi \cdot \star J) = - \intd^\dagger (\xi \wedge J).
\end{equation}
On the other hand, one can observe that the two form $H$ satisfies the identity
\begin{equation}\label{HIdentity}
\intd \intd^\dagger (\zeta H) = \zeta^{-1} \intd (\zeta^2 \intd^\dagger H) - U H,
\end{equation}
where $U H = -{\mathcal L}_\xi \Phi_0 Z' - {\mathcal L}_\xi \Phi_2 Z - \zeta \left[ C \cdot \Theta \Phi_2 + C' \cdot \Theta \Phi_0 \right] X$ and
\begin{equation}\label{XiwedgeJIdentity}
\zeta (B + B') \wedge J \cdot X = - \xi \wedge J - 2 V J 
\end{equation}
where $V J = \zeta \left( C \cdot J Y' + C' \cdot J Z \right)$. We have the following  TS-like identity for the Maxwell field,
\begin{equation}\label{LorentzTSI}
- \zeta^{-2} \intd (\zeta^2 \intd^\dagger {\mathcal L}_\xi H) 
= \half \zeta^{-2} \intd (\zeta^2 J \cdot X)
\end{equation}
using ${\mathcal L}_\xi H = -\phi_0 Z' - \phi_2 Z$. 
Note that $U H, {\mathcal L}_\xi H$ and $V J$ are purely anti-self dual 2 forms. Therefore $(1+i\star) U H = (1+i\star) {\mathcal L}_\xi H  = (1+i\star) V J = 0$. Consequently, using the definitions \eqref{LorentzTSI} and \eqref{LorentzGdef} and the identities \eqref{LieXiJIdentity}, \eqref{HIdentity} and \eqref{XiwedgeJIdentity}, we have
\begin{equation}
\begin{split}
{\mathcal L}_\xi \intd^\dagger F =& \ {\mathcal L}_\xi \intd^\dagger \left( \intd \intd^\dagger  (\zeta H) + \intd G \right) + \text{c.c.}\\
=& - i \star {\mathcal L}_\xi\intd (1 + i \star) \left[ \intd \intd^\dagger (\zeta H) + \intd G \right] + \text{c.c.}\\
=& - i \star {\mathcal L}_\xi \intd (1 + i \star) \left[ \zeta^{-1} \intd (\zeta^2 \intd^\dagger H) - UH + \intd G \right] + \text{c.c.}\\
=& - i \star \intd (1 + i \star) \left[ \zeta^{-1} \intd (\zeta^2 \intd^\dagger {\mathcal L}_\xi H) - U {\mathcal L}_\xi H + \half \intd (\zeta J \cdot Y) \right] + \text{c.c.}\\
=& - i \star \intd (1 + i \star) \left[ \zeta^{-1} \intd ( \zeta^2 \intd^\dagger {\mathcal L}_\xi H) + \half \intd( \zeta J \cdot Y) \right] + \text{c.c.}\\
=&  - i \star \intd (1 + i \star) \left[ - \half \zeta^{-1} \intd (\zeta^2 J \cdot Y) + \half \intd( \zeta J \cdot Y) \right] + \text{c.c.}\\
=&  \ \half i \star \intd (1 + i \star) \left[ \zeta (B + B') \wedge J \cdot Y \right] + \text{c.c.}\\
=& - \half i \star \intd (1 + i \star) \left( \xi \wedge J + 2 V J \right) + \text{c.c.}\\
=& - \half i \star \intd (1 + i \star) \left( \xi \wedge J \right) + \text{c.c.}\\
=& \ \intd^\dagger \left( \xi \wedge J \right)\\
=& - {\mathcal L}_\xi J.
\end{split}
\end{equation}
Hence, the Maxwell equation holds up to some term annihilated by ${\mathcal L}_\xi$. 
\end{proof}

\subsection{Solving the sourced Maxwell-equation using GHZ}
\label{sec:MaxGHZ}

As an alternative to the method described in the previous section, we may also use a version of GHZ for spin-1 \cite{Hollands:2020vjg} to solve 
the Maxwell equations \eqref{eq:maxwell}.
That method considers a decomposition of a vector potential $A_a$ solving \eqref{eq:maxwell} of the form
\begin{equation}
\label{GHZ:spin1}
A_a = \Re \S^\dagger_a \Psi + x_a + \dot A_a + \nabla_a X ,
\end{equation}
where $\Psi \circeq \GHPw{-2}{0}$ is a suitable Hertz potential, $x_a$ is a corrector, $\dot A_a$ is an infinitesimal perturbation
towards the vector potential of the Kerr-Newman solution (``zero mode'') in IRG, $\dot A_a l^a = 0$, corresponding to a linear change $\delta Q$ of the electric charge, and $X$ is some gauge function. 

Proceeding in a similar way as we did for spin-2 in Sec. \ref{ch:GHZ}, one can obtain a proof of the 
spin-1 GHZ decomposition \eqref{GHZ:spin1}, under much weaker assumptions than in \cite{Hollands:2020vjg}, by a backwards or forwards integration 
scheme. For instance, in the backward scheme, the asymptotic conditions at $\sI^+$ analogous to our Standing Decay Assumptions on $h_{ab}$ and $T_{ab}$, see Sec. \ref{GHZproof}, on $A_a$ and $J^a$ are now that
\begin{subequations}
\label{eq:Adec}
\begin{align}
\label{eq:Adeca}
    &A_n = O(r^{-1}), \quad A_m = O(r^{-1}), \quad A_l = O(r^{-2}), \\
\label{eq:Adecb}
    &J_n = O(r^{-2}), \quad J_m = O(r^{-3}), \quad J_l = O(r^{-4}),
\end{align}
\end{subequations}
hold, and that, possibly after subtracting from $A_a$ a zero mode, 
that all expansion coefficients in \eqref{asympt_exp} for $A_a, J_a$
tend to zero as $u \to -\infty$ with their derivatives. 
Many of these arguments would involve lengthy repetitions, though of a much more straightforward nature, of those given in Secs. \ref{GHZproof} and \ref{GHZproofH-} for spin-2. We will, therefore, not give a detailed discussion here but only outline the main steps. 

For definiteness, we consider a backward integration scheme
similar to that in Sec. \ref{GHZproof}. To prove a theorem analogous to theorem \ref{thm:2}, we require \eqref{eq:Adec}, as well as a decay condition on $E_A = \lim_{\tilde r \to \infty} \tilde r F_{A\tilde u}$ (the asymptotic electric field at $\sI^+$ in Bondi coordinates, see \ref{BondiIRG}) as $|\tilde u| \to \infty$, similar to that on the Bondi news tensor in theorem \ref{thm:2}.

We begin by subtracting a suitable zero mode, $\dot A_a$, from $A_a$ to
ensure that the expansion coefficients in \eqref{asympt_exp} for $A_a$
tend to zero as $u \to -\infty$. Next, we subtract a suitable gauge piece $\nabla_a X$ from $A_a$ to put it into IRG, $A_a l^a = 0$. The prerequisite transport equation is integrated backwards from $\sI^+$ with ``trivial'' final conditions (meaning $X = O(r^{-1})$).\footnote{The residual gauge transformation would consist of gauge pieces $\nabla_a \zeta$, where $\zeta \circeq \GHPw{0}{0}$ is annihilated by $\thorn$. However, unlike for spin-2, no such residual gauge transformations need to be considered here.} 

After this, we solve the transport ODE 
\begin{equation}
\label{x1def}
\rho^2 \thorn \left[ \frac{\bar{\rho}}{\rho^3} \thorn \left( \frac{\rho}{\bar{\rho}} x_n \right) \right] = -J_l.
\end{equation}
for $x_n$ demanding that
\begin{equation}
\label{eq:xndef}
x_n = x^{1\circ}_n \rho + O(\rho^2), \quad    
\tilde{\thorn}' x^{1\circ}_n = J^{2\circ}_n,
\end{equation}
which picks a unique solution if we demand that 
$x^{1\circ}_n$ goes to zero as $u \to -\infty$.
$x_n$ determines the corrector as
\begin{equation}
x^a = x_n l^a.
\end{equation}
We subtract this $x_a$ from $A_a$, thereby effectively eliminating 
the $l$ NP component from $J_a$, similar to the spin-2 case.

In fact, one now argues that, after these steps, $A_a$ can be written in reconstructed form as $A_a = \Re \S^\dagger_a \Psi$: We use the $\bar m$ NP
component of this equation in order to define $\Psi$. Writing the
$\bar m$ NP component of $\S^\dagger_a$ out (see \ref{app:ST}), we see that we need to integrate the transport equation
\begin{equation}
A_{\bar m} = -\frac{1}{4\rho}\thorn \left( \rho \Psi \right).
\end{equation}
We pick a unique solution demanding that
\begin{equation}\label{eq:Psidef}
\Psi = - 4 A^{1\circ}_{\bar m} + O(\rho)
\end{equation}
near $\sI^+$.
An analysis analogous to that in the proof of lemma \ref{lem:10} shows that, at this stage
\begin{equation}
\label{eq:Adiff}
    A_a - \Re \S^\dagger_a \Psi = l_a \left[a^\circ (\rho + \bar{\rho}) + b^\circ \left(\frac{\rho}{\bar{\rho}} + \frac{\bar{\rho}}{\rho} \right) \right],
\end{equation}
for certain real GHP scalars $a^\circ, b^\circ$ annihilated by $\thorn$. $b^\circ$ is identical to zero according to \eqref{eq:Psidef} and \eqref{eq:Adeca}. In addition, from \eqref{eq:xndef} we learn in view of \eqref{eq:Adecb} that $S^{2 \circ}_n$ vanishes. Knowing that $S_n = O(\rho^3)$ and substituting \eqref{eq:Adiff} into the Maxwell equation, we learn that $\tilde{\thorn}' a^\circ = 0$ (from its $n$ NP component). Then using that $a^\circ$ vanishes in the $u \to - \infty$ limit by the Standing Decay Assumptions, we see that $a^\circ = 0$.
Therefore, we get \eqref{GHZ:spin1} after adding $x_a, \nabla_a X, \dot A_a$ back to $A_a$. 

Having established the spin-1 GHZ decomposition \eqref{GHZ:spin1}, we 
now explain how it can be used in practice to solve the Cauchy problem. We first note that, similar to the spin-2 case, we 
have a relation between the Hertz potential $\Psi$ and the bottom Maxwell scalar, i.e. 
\begin{equation}
\label{eq:invspin-1}
    \phi_0 = -\dfrac{1}{4} \thorn^2 \bar \Psi.
\end{equation}
This suggests again that we should first solve the Cauchy problem for the sourced spin-1 Teukolsky equation $\O \phi_0 = \S^a J_a$. Since we are given $J_a$, we furthermore obtain the non-trivial component $x_n$ of the corrector by solving \eqref{x1def} subject to the 
final condition \eqref{eq:xndef}. To get the correct $\Psi$ from $\phi_0$ via the inversion relation \eqref{eq:invspin-1}, we must, as in the spin-2 case, take into account the kernel of that equation. 
The correct solution is picked by demanding that \eqref{eq:Psidef}
must hold. This requires $A^{1\circ}_{\bar m}$, for which we have the relation
\begin{equation}
\label{Abm}
    -\tilde{\thorn}' A^{1\circ}_{\bar m} = \phi_2^{1 \circ}, 
\end{equation}
where we are referring to the leading asymptotic expansion coefficient(s) of the Maxwell scalar(s), i.e., 
\begin{equation}
\label{peeling}
    \phi_2 = \phi_2^{1 \circ} \rho + O(\rho^2), \quad 
    \phi_0 = \phi_0^{3 \circ} \rho^3 + O(\rho^4).
\end{equation}
These in turn are related by an angular TS identity, 
\begin{equation}
    \tilde{\eth}' \tilde{\eth}' \phi_2^{1 \circ} = 
    \tilde{\thorn}'\tilde{\thorn}' \phi_0^{3 \circ}. 
\end{equation}
This enables us, via an expansion into spin-weighted spheroidal harmonics of the general form \eqref{eq:Fourier}, and the relations in Sec. \ref{sec:modesoln}, to obtain 
$\phi_2^{1 \circ}$ from $\phi_0^{3 \circ}$, which is known. We thereby get $A^{1\circ}_{\bar m}$ from \eqref{Abm}. Explicitly, 
\begin{equation}
\label{Acirc}
    A^{1\circ}_{\bar m}(u, \theta, \varphi_*) = - \sum_{m,\ell} \int\limits_0^\infty \dd \omega \
\frac{\omega}{{}_{1} B_{\ell m}(a\omega)} \phi_{0 \ell m}^{3 \circ}(\omega) 
{}_{-1} S_{\ell m}(\theta, \varphi_*; a\omega) e^{-i\omega u},
\end{equation}
where $(u,\theta,\varphi_*)$ are the retarded KN coordinates,
where ${}_{1} B_{\ell m}(a\omega)$ is the spin-1 angular TS
constant, and where ${}_{-1} S_{\ell m}(\theta, \varphi_*; a\omega)$ are the spin-weighted spheroidal harmonics, see Sec. \ref{sec:modesoln}. Eq. \eqref{Acirc} should be regarded as analogous to \eqref{Hmode} in the spin-2 case. The decay assumption on the electric field is needed in order to show that this integral exists, though we do not give the prerequisite argument here which is analogous to proof of proposition \ref{prop:13}.

To summarize, our scheme for solving the Cauchy problem for the sourced Maxwell equation, $\nabla^a F_{ab} = -J_b$, in Kerr is as follows, assuming for simplicity compactly supported initial data\footnote{Hence there is no zero-mode in the GHZ decomposition \eqref{GHZ:spin1}, and we automatically have 
the peeling property \eqref{peeling}.} for $A_a$, and assuming 
a source $J_a$ with the decay \eqref{eq:Adec} at $\sI^+$. 

\medskip
\noindent
{\bf Step 1.} Solve the Cauchy problem for $\phi_0$ 
in the sourced spin-1 Teukolsky equation, $\O \phi_0 = \S^a J_a$.

\medskip
\noindent
{\bf Step 2.} Obtain $A^{1\circ}_{\bar m}$ from \eqref{Acirc}.

\medskip
\noindent
{\bf Step 3.} Obtain the non-trivial NP component $x_n$ of the corrector by solving \eqref{x1def} subject to the 
final condition \eqref{eq:xndef}.

\medskip
\noindent
{\bf Step 4.} Solve the inversion problem \eqref{eq:invspin-1} to get $\Psi$ from $\phi_0$, by picking the unique solution to this transport equation for which \eqref{eq:Psidef} holds. 

\medskip
\noindent
{\bf Step 5.} The solution of the Cauchy problem (up to gauge) is 
$A_a = x_a + \Re \S^\dagger_a \Psi$.

\section{Conclusions}
\label{sec:concl}

In this paper we have re-analyzed the GHZ method \cite{Green:2019nam}. 
Our main aim was to obtain a version of that method applicable to the Cauchy problem 
for the sourced, linearized EE on Kerr. Our secondary, yet closely related, aim was to connect the quantities in the GHZ approach to observables for gravitational radiation such as the news, memory tensors or the corresponding BMS supertranslations. 

While the analysis was quite involved algebraically, the end recipe for solving the Cauchy problem is rather simple: We must solve, as usual, the standard Cauchy problem for the sourced Teukolsky equation for the Weyl scalar $\psi_0$, one set of transport equations for the NP components of the corrector tensor (for which a completely explicit Green's function exists \cite{Casals:2024ynr}), and an inversion problem to go from $\psi_0$ to the Hertz potential, $\Phi$. The precise understanding of this inversion problem in our setup was by far the most intricate part of our analysis, but its solution requires, at the end of the day, essentially only the spin-2 angular Teukolsky-Starobinsky identity.

Since our scheme is formulated for fairly general stress-energy tensors in the linearized EE, we should be able to iterate the procedure, in order to obtain the metric solving the initial value problem to as high an order in perturbation theory around Kerr as we like, in principle. 

In practice, for the extreme mass ratio inspiral problem, the second order in perturbation theory is known to be required to fully compute the self-force to the necessary accuracy for systems that will be observable by LISA, see e.g., \cite{Poisson:2011nh,Pound:Wardell,Barack}. Carrying out this analysis concretely, and with the prerequisite efficiency, for general inclined bound orbits in Kerr is one of the current most important open problems in this area, see \cite{Wardell:2021fyy} for the state-of-the-art. The GHZ method is a possible tool for this difficult challenge. For a proof-of-principle analysis of the self-force at first order for circular orbits in Schwarzschild, see \cite{Bourg:2024vre}.

Besides the self-force problem, perturbative approaches are, of course, useful also for a number of other questions in black hole physics, such as when studying possible non-linear effects in quasinormal mode ringdown after a black hole merger \cite{Sberna:2021eui,Mitman:2022qdl,Cheung:2022rbm,Lagos:2022otp,Ma:2024qcv},  when exploring the possibilities of turbulent dynamics around black holes \cite{Yang:2014tla,Green:2022htq}, or for quantum field theory on Kerr.

An obvious mathematical question is whether the extension of the GHZ method obtained in this paper can be used to directly transfer decay results on the Teukolsky equation (see \cite{DHR,Ma,Hintz,Angelopoulos,Shlapentokh-Rothman:2023bwo,Andersson}) to the sourced\footnote{For the case of the source-free linearized EE, such an analysis was given by \cite{Andersson:2019dwi}.} linearzed EE. This is perhaps possible, though one would have to pass to a better behaved gauge such as the Lorenz or no-string gauge, see Sec. \ref{sec:gaugeissues}, and one would have to generalize the analysis of \ref{sec:zetaproof} to a setup where no a priori knowledge on the decay of the news tensor is needed. This would require a better understanding of the mapping properties of the operator appearing in the angular Teukolsky-Starobinsky identities in Kerr. 

A considerably more speculative question is whether the apparent efficiency of our perturbative approach, combined with the possibility to switch between IRG (for infinity), ORG (for the horizon) and other gauges, could be useful for shedding further light on the non-linear Kerr stability problem \cite{Klainerman:2021qzy,Hintz2,Dafermos:2021cbw}.

\medskip
{\bf Acknowledgements:} We thank Gabriele Benomio, Marc Casals, Dejan Gajic, Dietrich H\" afner, Peter Hintz, Anna-Laura Sattelberger, Bernd Sturmfels and Barry Wardell for useful discussions. 
S.H. thanks the Erwin Schr\" odinger Institut in Vienna, where part of this work has been completed, for its hospitality and support.
S.H. is grateful to the Max-Planck Society for supporting the collaboration between MPI-MiS and Leipzig U., grant Proj. Bez. M.FE.A.MATN0003. 
S.H. acknowledges support of the Institut Henri Poincaré (UAR 839 CNRS-Sorbonne Université), and LabEx CARMIN (ANR-10-LABX-59-01).
Many of the calculations in this work were enabled by the xAct \cite{xact1,xact2} tensor algebra package for Mathematica.


\appendix

\section{GHP relations in Kerr \cite{Kinnersley:1969zza,Price}}
\label{sec:GHPformulas}
The formulas in this appendix are essentially due to \cite{Kinnersley:1969zza}, who however did not employ the GHP formalism. For derivations in the GHP formalism, see e.g., \cite{Price}. We are assuming that the background spacetime is Kerr, with an NP tetrad aligned with the principal null directions, though many formulas remain valid in
broader classes of type D spacetimes, see e.g., \cite{Price} for a more detailed discussion.
The GHP commutators are
\begin{subequations}
\begin{align}
[\thorn, \thorn'] =& \ (\bar{\tau} - \tau') \eth + (\tau - \bar{\tau}') \eth' - p(\Psi_2 - \tau \tau') - q (\bar{\Psi}_2 - \bar{\tau} \bar{\tau}'),\\
[\eth, \eth'] =& \ (\bar{\rho} - \rho') \thorn + (\rho - \bar{\rho}) \thorn' - p(\Psi_2 + \rho \rho') - q (\bar{\Psi}_2 + \bar{\rho} \bar{\rho}'),\\
[\thorn, \eth] =& \ \bar{\rho} \eth - \bar{\tau}' \thorn + q\bar{\rho} \bar{\tau}'
\end{align}
\end{subequations}
when acting on GHP scalars of weights $\GHPw{p}{q}$. The remaining commutators may be obtained from these by applying the priming- and overbar operations in the GHP formalism, i.e.,
\begin{equation}
    \bar \eth = \eth', \quad
    \bar \thorn = \thorn, \quad
    \bar \thorn' = \thorn'.
\end{equation}
The action of the GHP operators on the non-vanishing optical scalars  is summarized in table \ref{tab:4}.

\begin{table}[t]\label{DerivativeofGHPQuantities}
\begin{indented}
\item[]
\begin{tabular}{ c | c  c c }
\br
& $\rho$ & $\tau$ & $\Psi_2$\\ \mr
$\thorn$ & $\rho^2$ & $\rho(\tau-\bar \tau')$ & $3\rho\Psi_2$\\
$\thorn'$ & $\rho \rho'+ \tau'(\tau-\bar \tau') - \tfrac{1}{2} \Psi_2 - \frac{\rho}{2\bar \rho} \bar \Psi_2$ & $2\rho'\tau$ & $3\rho'\Psi_2$ \\
$\eth$ & $\tau(\rho-\bar\rho)$ &  $\tau^2$& 
$3\tau \Psi_2$ \\
$\eth'$ & $2\rho\tau'$ & $\tau \tau' + \rho'(\rho-\bar{\rho}') + \tfrac{1}{2}\Psi_2 - \frac{\rho}{2\bar{\rho}} \bar{\Psi}_2$ & $3\tau' \Psi_2$ \\
\br
\end{tabular}
\end{indented}
\caption{Action of GHP operators on non-vanishing background GHP scalars in Kerr. All other combinations may be obtained by applying the GHP priming and overbar operations to various entries in the table.} \label{tab:4}
\end{table}

We have 
\begin{equation}
\label{eq:GHPrelations}
\left( \frac{\Psi_2}{ \bar{\Psi}_2} \right)^{\tfrac{1}{3}} = 
\frac{\rho}{\bar{\rho}} = 
\frac{\rho'}{\bar{\rho}'} = 
-\frac{\tau}{\bar{\tau}'} = 
-\frac{\tau'}{\bar{\tau}},
\end{equation}
and these are the only remaining relations on a general non-accelerating type D spacetime, i.e. the scalars 
$\tau, \tau', \rho, \rho', \Psi_2$ are independent up to these relations as functions of the spacetime coordinates and the moduli of non-accelerating type D spacetime.

\section{Formulas for Held's version of GHP \cite{Held,StewardWalker}}
\label{sec:Held}
Held \cite{Held} has shown that 
\begin{subequations}\label{eq:exp}
  \begin{align}
    \rho' 
    &= \half (\rho + \bar{\rho}) \rho^{\prime \circ} + \half (\rho - \bar{\rho}) \bar{\rho}^{\prime \circ} - \half \rho^2 \Psi^\circ - \half \rho \bar{\rho} \bar{\Psi}^\circ - \rho^2 \bar{\rho} \tau^\circ \bar{\tau}^\circ,\\
    \tau' &= -\bar{\tau}^\circ \rho^2, \frac{}{} \\
    \tau &= \tau^\circ \rho\bar\rho, \frac{}{}\\
    \Psi_2 &= \Psi^\circ \rho^3 \frac{}{},
  \end{align}
\end{subequations}
for any vacuum, type D spacetime of the Case II in  Kinnersley's classification~\cite{Kinnersley:1969zza} (which includes Kerr) as long as the NP frame is aligned with the principal null directions.
Here, $\tau^\circ \circeq \GHPw{-1}{-3}$, $\Psi^\circ \circeq \GHPw{-3}{-3}$, $\rho^{\prime \circ} \circeq \GHPw{-2}{-2}$, and
\begin{equation}
\label{eq:barrho}
\Omega^\circ = \frac{1}{\bar{\rho}} - \frac{1}{\rho} \circeq \GHPw{-1}{-1}.
\end{equation}
Using these formulas, one finds \cite{Held} the following relationships between Held's operators \eqref{eq:Hops} and the GHP operators 
\begin{subequations}\label{eq:Hops2}
\begin{align}
\thorn' &= \tilde{\thorn}' + \rho \bar{\rho} \bar{\tau}^\circ \tilde{\eth} + \rho \bar{\rho} \tau^\circ \tilde{\eth}' - \rho \bar{\rho} \tau^\circ \bar{\tau}^\circ \left( p \rho + q\bar{\rho} \right) - \half \left( p \Psi^\circ \rho^2 + q \bar{\Psi}^\circ \bar{\rho}^2 \right), \\
\eth &= \bar{\rho} \tilde{\eth} - q \bar{\rho}^2 \tau^\circ, \\
\eth' &= \rho \tilde{\eth}' - p \rho^2 \bar{\tau}^\circ \ . 
\end{align}
\end{subequations}
These hold when acting on GHP scalar with the weight of $\GHPw{p}{q}$. 
Table \ref{tab:1} summarizes the action of Held's operators on the GHP scalars $\rho, \Omega^\circ, \rho^{\prime \circ}, \tau^\circ$ and $\Psi^\circ$.
\begin{table}[t]\label{DerivativeofHeldQuantities}
\begin{indented}
\item[]
\begin{tabular}{ c | c  c  c }
\br
& $\tilde{\thorn}'$ & $\tilde{\eth}$ & $\tilde{\eth}'$ \\ \mr
$\rho$ & 
$\rho^2 \rho^{\prime \circ} - \half \rho^2 (\rho \Psi^\circ + \bar{\rho} \bar{\Psi}^\circ) - \rho^3 \bar{\rho} \tau^\circ \bar{\tau}^\circ$ &
$\tau^\circ \rho^2$ &
$-\bar{\tau}^\circ \rho^2$ \\
$\Omega^\circ$ & $0$ & $2\tau^\circ$ & $-2\bar{\tau}^\circ$\\ 
$\rho^{\prime \circ}$ & $0$ & $0$ & $0$ \\
$\tau^\circ$ & $0$ & $0$ & 
                       $\half(\rho^{\prime \circ} + \bar{\rho}^{\prime \circ})\Omega^\circ + \half(\Psi^\circ - \bar{\Psi}^\circ)$\\
$\Psi^\circ$ & $0$ & $0$ & $0$ \\
\br
\end{tabular}
\end{indented}
\caption{Action of Held's operators 
in type D spacetimes 
of Kinnersley's Case II \cite{Held}, taken from \cite{HT2}.} \label{tab:1}
\end{table}
To automate calculations using Held's formalism, one requires furthermore the commutators between Held's operators. These can easily be obtained 
from \cite{Held} by specializing to type D, see also \cite{StewardWalker}. The first set of commutators
\begin{equation}
\begin{split}
[\thorn, \tilde{\thorn}'] 
&= \left(\rho \bar{\rho} (\rho + \bar{\rho}) \tau^\circ \bar{\tau}^\circ + \half \rho^2 \Psi^\circ + \half \bar{\rho}^2 \bar{\Psi}^\circ \right) \thorn,  \\
[\thorn, \tilde{\eth}] &= 
0, 
\end{split}
\end{equation}
expresses that if $\tilde{\eth}, \tilde{\eth}', \tilde{\thorn}'$ are applied to a GHP scalar annihilated by $\thorn$, then we obtain another such GHP scalar. The second set of commutators is
\begin{equation}
\begin{split}
[\tilde{\eth}, \tilde{\eth}'] 
&= \left( - \frac{\rho^{\prime \circ}}{\bar{\rho}} + \frac{\bar{\rho}^{\prime \circ}}{\rho} + \half \rho \Omega^\circ \Psi^\circ + \half \bar{\rho} \Omega^\circ \bar{\Psi}^\circ + \rho \bar{\rho} \Omega^\circ \tau^\circ \bar{\tau}^\circ \right) \thorn + \Omega^\circ \tilde{\thorn}' + p \rho^{\prime \circ} - q \bar{\rho}^{\prime \circ},\\
[\tilde{\thorn}', \tilde{\eth}'] 
&= \left( \rho \rho^{\prime \circ} \bar{\tau}^\circ - \bar{\rho} \bar{\rho}^{\prime \circ} \bar{\tau}^\circ - \half \rho^2 \bar{\rho} \bar{\tau}^\circ \Omega^\circ \Psi^\circ - \half \rho \bar{\rho}^2 \bar{\tau}^\circ \Omega^\circ \bar{\Psi}^\circ - \rho^2 \bar{\rho}^2 \Omega^\circ \tau^\circ \bar{\tau}^{\circ 2} \right) \thorn 
\end{split}
\end{equation}
%
%
%
when acting on GHP scalars of weights $\GHPw{p}{q}$.
In the Kinnersley tetrad and retarded KN coordinates, the Held GHP scalars and operators take the forms given in table \ref{tab:2}.
\begin{table}[t]\label{eq:Held coeffs in Kinn}
\renewcommand{\arraystretch}{1.5}
\begin{indented}
\item[]
\begin{tabular}{  c  c }
\br
Held quantity & Value in Kinnersley frame \\ \mr 
$\tau^\circ$ & $-\frac{1}{\sqrt{2}} ia\sin\theta$ \\
$\rho^{\prime\circ}$ & $-\half$ \\
$\Psi^\circ$ & $M$ \\[-.5em]
$\Omega^\circ$ & $-2ia\cos\theta$ \\[-.5em]
$\thorn$ & $\partial_r$ \\
$\tilde{\thorn}'$ & $\partial_u - \half \frac{r^2 - 2 M r + a^2}{r^2 + a^2 \cos^2 \theta} \partial_r$\\
$\tilde{\eth}$ & $-\frac{1}{\sqrt{2}} \left( \partial_\theta 
+ i \csc \theta \partial_{\varphi_*} -i a \sin \theta \partial_u- s \cot \theta \right)$\\
$\tilde{\eth}'$ & $-\frac{1}{\sqrt{2}} \left( \partial_{\theta} - i \csc \theta \partial_{\varphi_*} + ia\sin \theta \partial_u + s \cot \theta \right)$\\
\br
\end{tabular}
\end{indented}
\caption{Summary of Held's GHP quantities in Kerr in retarded KN coordinates and the Kinnersley frame, taken from \cite{HT2}. 
Held's operators $\tilde{\eth}$ and $\tilde{\eth}'$ are thereby equal to  ${}_s \mathcal{L}^\pm(a\omega)$ (see Sec. \ref{sec:modesoln})
when acting on a mode with dependence $e^{-i\omega t+im\varphi}$. 
They are also equal up to normalization to the operators appearing in the theory of spin-weighted harmonics, see e.g., \cite{Chandrasekhar:1984siy}.} \label{tab:2}
\end{table}
\section{Zero modes in TIRG \cite{HT2}}\label{app:D}

We use the term zero mode for algebraically special perturbations within the Kerr class in TIRG. Such a perturbation is gauge equivalent to the perturbation
\begin{equation}\label{dgZ}
h_{ab} = \left( \delta M\, \frac{\partial}{\partial M} g^{M,a}_{\mu\nu} + \delta a\, \frac{\partial}{\partial a}  g^{M,a}_{\mu\nu} \right) (\dd x^\mu)_a (\dd x^\nu)_b,
\end{equation}
with $g^{M,a}_{\mu\nu}$ the Kerr-family in some coordinate system. After transforming this perturbation to TIRG, we obtain the 2-dimensional space of zero modes, $\dot g_{ab}$. 
Its NP components are \cite{HT2}
\begin{equation}
\begin{split}
\dot g_{nn} =& \left( \rho + \bar{\rho} \right) g^\circ_M + 2 \rho \bar{\rho} \left( \rho + \bar{\rho} \right) \Re \left( \half \Omega^\circ \tilde{\eth}' g^\circ_a + \bar{\tau}^\circ g^\circ_a \right),\\
\dot g_{nm} =& \rho \left( \rho + \bar{\rho} \right) g^\circ_a,
\label{zeromode1}
\end{split}
\end{equation}
and 
\begin{equation}
\dot g_{mm} = \dot g_{nl} = \dot g_{ml} = \dot g_{ll} = \dot g_{\mb m} = 0,
\label{zeromode2}
\end{equation}
where $g^\circ_M = - 2 \rho^{\prime \circ} \Psi^\circ \frac{\delta M}{M} \circeq \left\{ -3, -3 \right\}$ and $g^\circ_a = - \tau^\circ \Psi^\circ \frac{\delta a}{a} \circeq \left\{ -2, -4 \right\}$.

\section{Bondi expansion in IRG}
\label{BondiIRG}

\subsection{Bondi gauge}

Recall that in the framework of conformal infinity (for details see e.g., \cite{GerochAS,Waldbook}), a spacetime $(\sM, g_{ab})$ is called asymptotically flat 
if it has a suitable conformal envelope $(\tilde \sM, \tilde g_{ab}, \Omega)$. The boundary of this envelope $\tilde \sM$ contains 
the null infinities $\sI^\pm \cong {\mathbb R} \times {\mathbb S}^2$. On $\sI^\pm$, we have $\Omega = 0$ but $\nabla_a \Omega \neq 0$, 
and $\tilde g_{ab} = \Omega^2 g_{ab}$ as well as $\Omega$ is required to be smooth across $\sI^\pm$. 
There is much freedom in the specific choice of $\Omega$. But one can show (see e.g., \cite{Hollands:2016oma}) 
that  it is possible to pick a coordinate system $(\tilde u, \tilde x^A)$
on e.g., $\sI^+$, and to pick $\Omega$ in such a way that, if we set $\tilde r = 1/\Omega$,
the metric components behave as 
\begin{subequations}
\label{Bondicord}
\begin{align}
&g_{\tilde r \tilde r} = g_{\tilde r A} = 0, \quad g_{\tilde r \tilde u} = 1,\\
&g_{AB} = -\tilde r^2 s_{AB} + O(\tilde r), \\
&g_{\tilde u A} = O(1),\\
&g_{\tilde u \tilde u} = -1/2 + O(\tilde r^{-1}),\\
&g_{AB} s^{AB} = O(1),
\end{align}
\end{subequations}
when $\tilde r  \to \infty$ at fixed $(\tilde u, \tilde x^A)$ (i.e., near $\sI^+$).
Here, $s_{AB} \dd \tilde x^A \dd \tilde x^B=\dd\tilde \theta^2 + \sin^2 \tilde \theta \, \dd \tilde \varphi^2$ is the standard metric on $\mathbb{S}^2$, and the result holds
provided that $g_{ab}$ satisfies the EE with a stress tensor vanishing sufficiently rapidly towards $\sI^+$, and provided $g_{ab}$
approaches a stationary metric at a sufficiently fast rate as\footnote{$\tilde u$ is chosen to increase towards the future.} $\tilde u \to -\infty$. 

We will refer to \eqref{Bondicord} as the ``Bondi gauge''\footnote{Technically speaking, this gauge should be called more properly ``conformal Gaussian null gauge'', because, while very similar, it is not identical with that usually referred to as ``Bondi'' or ``Bondi-Sachs gauge'' in most of the literature.}. For perturbation theory, one first needs to find this gauge i.e., 
the ``Bondi coordinates'' $(\tilde u, \tilde r, \tilde x^A)$, for the Kerr metric itself. Starting from BL coordinates, this can easily be accomplished using coordinates found by \cite{Fletcher}. Their coordinates are stated in terms of the auxiliary quantities 
\begin{subequations}
\begin{align}
B =& \ \left[ (r^2 + a^2)^2 - a^2\Delta \right]^{1/2} = r^2 + \half a^2 + O(r^{-1})\\
\alpha =& \ -a \int_r^\infty \left[ r^4 + a^2 r^2 + 2Ma^2 r\right]^{1/2} \dd r = -\frac{a}{r} + \frac{a^3}{6r^3} + O(r^{-4}).
\end{align}
\end{subequations}
At first, one sets $\tilde r = r$ and 
\begin{subequations}
\begin{align}
t =& \ \tilde u + a \frac{\tanh \alpha + \sin \tilde \theta}{1+(\tanh \alpha) \sin \tilde \theta} + \int^r \frac{B}{\Delta} \dd r,\\
\theta =& \ \arcsin\left[ \frac{\tanh \alpha + \sin \tilde \theta}{1+(\tanh \alpha) \sin \tilde \theta}\right] , \\
\varphi =& \ \tilde \varphi +\int^r \frac{2Mar}{B\Delta} \dd r ,
\end{align}
\end{subequations}
wherein $(t,r,\theta,\varphi)$ are the BL coordinates. As shown by \cite{Fletcher}, the Kerr metric is almost
in Bondi-gauge in the coordinates $(\tilde u, \tilde r, \tilde \theta, \tilde \varphi)$ with the only exception that (see \cite{Fletcher}, Eq. (48))
$
g_{\tilde r \tilde u} =  \frac{\Sigma(r,\theta)}{B(r)}
$
and not 
$g_{\tilde r \tilde u} = 1$ as one would like. This is easily remedied by introducing a different $\tilde r$ related to $r$ by
\begin{equation}
\tilde r =  \int^{\tilde r} \frac{\Sigma}{B} \dd r. 
\end{equation}
Having put the Kerr metric into Bondi gauge, we can say that a linear perturbation $h_{ab}$ is asymptotically flat 
at $\sI^+$ if all its Bondi coordinate $(\tilde u, \tilde r, \tilde \theta, \tilde \varphi)$ components have an asymptotic 
expansion in $1/\tilde r$ with coefficient functions that are smooth in $(\tilde u, \tilde \theta, \tilde \varphi)$, and 
furthermore, such that the linearization of the conditions \eqref{Bondicord} hold, i.e., 
\begin{subequations}
\label{Bondicordh}
\begin{align}
&h_{\tilde r \tilde r} = h_{\tilde r A} = h_{\tilde r \tilde u} = 0,\\
&h_{AB} = O(\tilde r), \\
&h_{\tilde u A} = O(1),\\
&h_{\tilde u \tilde u} = O(\tilde r^{-1}),\\
&h_{AB} s^{AB} = O(1).
\end{align}
\end{subequations}
To find the NP component version of these conditions, we consider the Kinnersley tetrad \eqref{eq:Kintet} in BL coordinates.
By expressing this NP tetrad in the Bondi coordinates $(\tilde u, \tilde r, \tilde \theta, \tilde \varphi)$ using 
the above definitions, one can see that the linearized Bondi gauge conditions \eqref{Bondicordh} imply the conditions 
\eqref{eq:hdec} given in the main text. 

As is well-known, in Bondi gauge, quantities related to gravitational radiation at $\sI^+$ have particularly simple expressions, see e.g., 
\cite{Hollands:2016oma}. For example, the linearized Bondi news tensor is given by
\begin{equation}
N_{AB} = -\lim_{\tilde r \to \infty} \frac{\partial_{\tilde u} h_{AB}}{\tilde r}.
\end{equation}
Assuming that $N_{AB}$ has a sufficient decay as $|\tilde u| \to \infty$, we can define the memory tensor
\begin{equation}
\label{eq:Memory}
\Delta_{AB} = \int_{-\infty}^\infty \dd \tilde u 
\, N_{AB}. 
\end{equation}
Finally, assuming that $T_{ab}$ decays sufficiently fast at $\sI^+$ (e.g. if it is supported for $r<r_{\rm max}$), 
the only contribution to the flux of energy through $\sI^+$ comes from the metric and is given in our units ($8\pi G=1$) at first non-trivial 
(quadratic) perturbation order by 
\begin{equation}
\label{eq:Flux}
F_{\sI^+} = -\frac{1}{4} \int_{-\infty}^\infty \dd \tilde u 
\int_{\mathbb{S}^2} \dd^2 \tilde x \, \sqrt{s} 
N_{AB} N^{AB}, 
\end{equation}
where indices are contracted with the round metric $s_{AB}$. By expressing the Bondi coordinate components in terms of NP components in 
the Kinnersley frame \eqref{eq:Kintet} via the above formulas, we may easily express $N_{AB}, \Delta_{AB}, F_{\sI^+}$ in terms of NP 
components of the asymptotic expansion coefficients of the metric perturbation as done in the main text.

\subsection{Bondi expansion}

The EE gives relationships between the various expansion coefficients of the metric perturbation in powers of $1/\tilde r$. These are traditionally analyzed in coordinate components in Bondi gauge \eqref{Bondicordh}, see e.g., \cite{Hollands:2016oma}.
In this paper, we require, however, the consequences of the EE to infer relationships between the various NP components. 
In view of the complicated form of the NP tetrad aligned with the principal null directions in Bondi coordinates
$(\tilde u, \tilde r, \tilde \theta, \tilde \varphi)$ described above, it is easiest to derive these relationships from scratch. 
To do this, we will make use of the fact that a metric perturbation $h_{ab}$ satisfying the Standing Decay Assumptions 
at $\sI^+$ in Sec. \ref{sec:gdot} may be put into IRG, $h_{ab}l^b=0$, by a gauge transformation preserving these conditions (see Sec. \ref{sec:gaugexi}). 

We therefore assume for the rest of this section that $h_{ab}$ is in IRG, and we carry out the asymptotic expansions of the type \eqref{asympt_exp} for its 
non-zero NP components and those of $T_{ab}$, the stress tensor.\footnote{These expansion are equivalent 
to corresponding expansion in $1/r$, because $O(\rho) = O(1/r)$ in the Kinnersley frame and retarded KN coordinates.} 
Then we substitute these expansions into the linearized EE in GHP form, see \ref{app:LinEinGHP},
replacing all GHP operators and optical scalars by their counterparts in Held's variant of GHP, 
and then expanding the resulting formulas in powers of $\rho$.  See  \ref{sec:Held} for 
a compilation of the prerequisite identities. The precise 
decay assumptions that we will work with at first are:
\begin{equation}
\label{hdecay}
h_{nn}, h_{nm}, h_{mm}, h_{m\bar m} = O(\rho),
\end{equation}
as well as 
\begin{equation}\label{eq:BondiExpansionforh}
T_{ll} = O(\rho^4), \quad
T_{lm},T_{mm},T_{m\bar m} = O(\rho^3), \quad
T_{ln},T_{nm} = O(\rho^2), \quad
T_{nn} = O(\rho),
\end{equation}
which, as we note, are slightly weaker than what we postulated in our Standing Decay Assumptions, \eqref{eq:hdec}, \eqref{Tdec} in Sec. \ref{sec:gdot}. Consequently, 
\eqref{hdecay} are also slightly weaker than the conditions imposed by \eqref{Bondicordh} (the condition on the trace is instead $h_{AB} s^{AB} = O(\tilde r)$).
Before working out the consequences of decay assumptions \eqref{hdecay}, \eqref{eq:BondiExpansionforh}, we note that the 
relation $\sfrac{1}{\bar{\rho}}-\sfrac{1}{\rho}=\Omega^\circ$ trivially yields
\begin{equation}
\begin{split}
\bar{h}^{1 \circ}_{m\bar{m}} =& \ h^{1 \circ}_{m\bar{m}},\\
\bar{h}^{2 \circ}_{m\bar{m}} =& \ h^{2 \circ}_{m\bar{m}} + \Omega^\circ h^{1 \circ}_{m\bar{m}},\\
\bar{h}^{3 \circ}_{m\bar{m}} =& \ h^{3 \circ}_{m\bar{m}} + 2 \Omega^\circ h^{2 \circ}_{m\bar{m}} + \Omega^{\circ 2} h^{1 \circ}_{m\bar{m}}.
\end{split}
\end{equation}
as well as
\begin{equation}
\begin{split}
\bar{h}^{1 \circ}_{nn} =& \ h^{1 \circ}_{nn},\\
\bar{h}^{2 \circ}_{nn} =& \  h^{2 \circ}_{nn} + \Omega^\circ h^{1 \circ}_{nn},
\end{split}
\end{equation}
which are implicitly used in many equations below.
The $ll$ NP component of the EE gives, at leading orders,
\begin{equation}
\begin{split}
T^{4 \circ}_{ll} =& \ 2 h^{2 \circ}_{m\bar{m}} + \Omega^\circ h^{1 \circ}_{m\bar{m}},\\
T^{5 \circ}_{ll} =& \ 6 h^{3 \circ}_{m\bar{m}} + 2 \Omega^\circ h^{2 \circ}_{m\bar{m}} - 2 \Omega^{\circ 2} h^{1 \circ}_{m\bar{m}}.\\
\end{split}
\end{equation}
The $ln$ NP component of the EE gives, at leading orders, 
\begin{equation}
\begin{split}
T^{2 \circ}_{ln} =& \tilde{\thorn}' h^{1 \circ}_{m\bar{m}},\\
T^{3 \circ}_{ln} =& \Re \left( \tilde{\eth}' \tilde{\eth} h^{1 \circ}_{m\bar{m}} - 2 \tilde{\eth}' h^{1 \circ}_{nm} - \tilde{\eth}' \tilde{\eth}' h^{1 \circ}_{mm} \right) - \Omega^\circ \tilde{\thorn}' h^{1 \circ}_{m\bar{m}}.
\end{split}
\end{equation}
The $lm$ NP component of the EE gives, at leading orders, 
\begin{equation}
\label{eq:Tlm}
\begin{split}
T^{3 \circ}_{lm} =& \ \half \tilde{\eth} h^{1 \circ}_{m\bar{m}} - \half \tilde{\eth}' h^{1 \circ}_{mm} - h^{1 \circ}_{nm},\\
T^{4 \circ}_{lm} =& \ \tilde{\eth} h^{2 \circ}_{m\bar{m}} - \tilde{\eth}' h^{2 \circ}_{mm} - \half  \Omega^\circ \tilde{\eth}' h^{1 \circ}_{mm} + \Omega^\circ h^{1 \circ}_{nm} + 2 \bar{\tau}^\circ h^{1 \circ}_{mm} + \tau^\circ h^{1 \circ}_{m\bar{m}}.
\end{split}
\end{equation}
The $nn$ NP component of the EE gives, at leading orders, 
\begin{equation}
\label{eq:Tnn}
\begin{split}
T^{1 \circ}_{nn} =& \ \tilde{\thorn}' \tilde{\thorn}' h^{1 \circ}_{m\bar{m}},\\
T^{2 \circ}_{nn} =& \ - \Re \left( 2 \tilde{\eth}' \tilde{\thorn}' h^{1 \circ}_{nm} + \tilde{\thorn}' h^{1 \circ}_{nn} \right) + \tilde{\thorn}' \tilde{\thorn}' h^{2 \circ}_{m\bar{m}}.
\end{split}
\end{equation}
The $nm$ NP component of the EE gives, at leading orders, 
\begin{equation}\label{eq:Tnm}
\begin{split}
T^{2 \circ}_{nm} =& \ \half \tilde{\eth} \tilde{\thorn}' h^{1 \circ}_{m\bar{m}} - \half \tilde{\eth}' \tilde{\thorn}' h^{1 \circ}_{mm} - \tilde{\thorn}' h^{1 \circ}_{nm},\\
T^{3 \circ}_{nm} =& \ - \half \tilde{\eth} \tilde{\eth} \bar{h}^{1 \circ}_{nm} + \half \tilde{\eth}' \tilde{\eth} h^{1 \circ}_{nm} + \half \tilde{\eth} \tilde{\thorn}' h^{2 \circ}_{m\bar{m}} - \half \Omega^\circ \tilde{\eth} \tilde{\thorn}' h^{1 \circ}_{m\bar{m}} - \half \tilde{\eth}' \tilde{\thorn}' h^{2 \circ}_{mm} + \half \tilde{\eth} h^{1 \circ}_{nn}\\
&+ \half \rho^{\prime \circ} \tilde{\eth} h^{1 \circ}_{m\bar{m}} - \half \rho^{\prime \circ} \tilde{\eth}' h^{1 \circ}_{mm} + \frac{5}{2} \tau^\circ \tilde{\thorn}' h^{1 \circ}_{m\bar{m}} + \frac{5}{2} \bar{\tau}^\circ \tilde{\thorn}' h^{1 \circ}_{mm} - \frac{3}{2} \tilde{\thorn}' h^{2 \circ}_{nm}.
\end{split}
\end{equation}
The $mm$ NP component of the EE gives, at leading orders, 
\begin{equation}
\begin{split}
T^{3 \circ}_{mm} =& \ - \tilde{\thorn}' \left(h^{2 \circ}_{mm} + \Omega^\circ h^{1 \circ}_{mm} \right),\\
T^{4 \circ}_{mm} =& \ - 2 \rho^{\prime \circ} h^{2 \circ}_{mm} + \frac{3}{2} \Psi^\circ h^{1 \circ}_{mm} - \half \bar{\Psi}^\circ h^{1 \circ}_{mm} - \rho^{\prime \circ} \Omega^\circ h^{1 \circ}_{mm} - \bar{\tau}^\circ \tilde{\eth} h^{1 \circ}_{mm} - \tau^\circ \tilde{\eth}' h^{1 \circ}_{mm}\\
&+ \tilde{\eth} h^{2 \circ}_{nm} + \Omega^\circ \tilde{\eth} h^{1 \circ}_{nm} - 2 \tilde{\thorn}' \left(h^{3 \circ}_{mm} + \half \Omega^\circ h^{2 \circ}_{mm} - \half \Omega^{\circ 2} h^{1 \circ}_{mm} \right).
\end{split}
\end{equation}
Finally, the $m\bar m$ NP component of the EE gives, at leading orders, 
\begin{equation}
\begin{split}
T^{3 \circ}_{m\bar{m}} =& \ \tilde{\thorn}' \left( h^{2 \circ}_{m\bar{m}} + \half \Omega^\circ h^{1 \circ}_{m\bar{m}} \right),\\
T^{4 \circ}_{m\bar{m}} =&\  \Re \left( \frac{}{} - 4 \bar{\tau}^\circ h^{1 \circ}_{nm} - \Psi^\circ h^{1 \circ}_{m\bar{m}} + 2 \tau^\circ \tilde{\eth}' h^{1 \circ}_{m\bar{m}} - 2 \Omega^\circ \tilde{\eth}' h^{1 \circ}_{nm} - \tilde{\eth}' h^{2 \circ}_{nm} \right) + h^{2 \circ}_{nn} + \half \Omega^\circ h^{1 \circ}_{nn}\\
&+ 2 \rho^{\prime \circ} h^{2 \circ}_{m\bar{m}} + \Omega^\circ \rho^{\prime \circ} h^{1 \circ}_{m\bar{m}} + \tilde{\thorn}' \left( 2 h^{3 \circ}_{m\bar{m}} + \half \Omega^\circ h^{2 \circ}_{m\bar{m}} - \half \Omega^{\circ 2} h^{1 \circ}_{m\bar{m}} \right).
\end{split}
\end{equation}
Let us now assume that the stress tensor satisfies even
\begin{equation}
T_{nn}=O(\rho^2), \quad
T_{ln} = O(\rho^3), \quad
T_{lm},
T_{ll} = O(\rho^5),
\end{equation}
and let us assume that the leading order (in $\rho$) $m\bar m$ NP components of the $m\bar m$ metric perturbation vanishes as we take the retarded KN coordinate $u \to -\infty$.

Then we can draw the following conclusions:
\begin{enumerate}
\item Using the $ln$ NP component of the linearized EE we get $\tilde{\thorn}' h^{1 \circ}_{m\bar m} = 0$. 
Then, since $\tilde{\thorn}'=\partial_u$ in retarded KN coordinates and the Kinnersley frame, we conclude that $h^{1 \circ}_{m\bar m} = 0$, as this quantity is constant in $u$ and vanishes for $u \to -\infty$. Using the $ll$ NP component of the linearized EE we get $h^{2 \circ}_{m\bar m} = 0$.

\item Using the $nn$ NP component of the linearized EE we then get
\begin{equation}
\label{eq:bondinn}
T^{2 \circ}_{nn} = -\tilde \thorn' \left(  \tilde{\eth}' h^{1 \circ}_{nm} + \tilde{\eth} \, \bar{h}^{1 \circ}_{nm} +  h^{1 \circ}_{nn} \right).
\end{equation}
This equation can be seen as the NP analog of the flux-balance law for the Bondi mass aspect $\mu^\circ \circeq \GHPw{-3}{-3}$ ($=-$ the expression in $(\dots)$ on the right side) i.e.,
\begin{equation}
\label{eq:bondimu}
\mu^\circ  =   \tilde{\eth}' h^{1 \circ}_{nm} + \tilde{\eth} \, \bar{h}^{1 \circ}_{nm} +  h^{1 \circ}_{nn} .
\end{equation}

\item Using the $lm$ NP component of the linearized EE  at order $\rho^3$ and expanding all quantities in terms of Held's operators, it follows
\begin{equation}
0=  \half \tilde \eth' h^{1 \circ}_{mm} + h^{1 \circ}_{nm} .
\end{equation}

\item Using the $nm$ and $nn$ NP component of the linearized EE at order $\rho^3$ and $\rho^2$, using the previous items, and developing all quantities in terms of Held's operators, it follows
\begin{equation}\label{eq:Jorigin}
-2 \tilde{\thorn}' T_{nm}^{3 \circ} + \tilde{\eth} T_{nn}^{2 \circ}
= \tilde{\thorn}' \left[-2 \tilde{\eth} \left(\tilde{\eth}' h^{1 \circ}_{nm} + h^{1 \circ}_{nn} \right) + 3 \tilde{\thorn}' \left(h^{2 \circ}_{nm} -\bar\tau^\circ {h}^{1 \circ}_{mm} + \Omega^\circ h^{1 \circ}_{nm} \right)   \right].
\end{equation}
By combining \eqref{eq:Jorigin} with \eqref{eq:bondinn}, one can furthermore obtain
\begin{equation}
\label{eq:fluxbalance}
\begin{split}
&\tilde{\thorn}' \left( 2\tau^\circ T^{2 \circ}_{nn} - 4 \rho^{\prime \circ} T^{3 \circ}_{nm} \right) \\
&= \left( \tilde{\eth} + \dfrac{\tau^\circ}{2 \rho^{\prime \circ}}  \tilde{\thorn}' \right)  \left[
 - 2 \rho^{\prime \circ} T^{2\circ}_{nn} - 4 \rho^{\prime \circ} \left(\tilde{\eth}' h^{1 \circ}_{nm} + h^{1 \circ}_{nn} \right) \right]\\
&+ \tilde{\thorn}' \tilde{\thorn}' \left[ 6 \rho^{\prime \circ} \left( h^{2 \circ}_{nm} -\bar\tau^\circ {h}^{1 \circ}_{mm} + \Omega^\circ h^{1 \circ}_{nm} \right) - \tau^\circ \left(\tilde{\eth}'h^{1 \circ}_{nm} + 3 \tilde{\eth} \bar{h}^{1 \circ}_{nm} + h_{nn}^{1 \circ} \right) \right]
\end{split}
\end{equation}
This equation gives the NP analog of the flux-balance law for the Bondi angular aspect as follows. 
The operator $\tilde{\eth} + \tfrac{\tau^\circ}{2 \rho^{\prime \circ}}  \tilde{\thorn}'$ is a spin-lowering operator on 
a 2-sphere $\mathbb{S}^2$ of constant $u,r$ in retarded KN coordinates and the Kinnersley frame, see table \ref{tab:2} in \ref{sec:Held}.
Therefore, taking the inner product with an appropriate spin $s=1$, $\ell = 1$ spin-weighted harmonic, 
the term in the second line vanishes when integrated over $\mathbb{S}^2$ after an integration by parts.
Then, we can basically remove one $\tilde \thorn'$ from both sides, as this is $\partial_u$ in retarded KN coordinates and the Kinnersley frame, provided that the terms fall of sufficiently say, as $u \to -\infty$. The left side then gives the flux of the Bondi angular momentum component
associated with this spin-weighted harmonic, and the right side i.e., the expression in $[\dots]$ in the term in the third line, gives the NP form of the 
Bondi angular moment aspect $j^\circ \circeq \GHPw{-6}{-4}$ i.e.,  
\begin{equation}
\label{eq:Bondij}
j^\circ= 6 \rho^{\prime \circ} \left( h^{2 \circ}_{nm} -\bar\tau^\circ {h}^{1 \circ}_{mm} + \Omega^\circ h^{1 \circ}_{nm} \right) - 
\tau^\circ \left(\tilde{\eth}'h^{1 \circ}_{nm} + 3 \tilde{\eth} \bar{h}^{1 \circ}_{nm} + h_{nn}^{1 \circ} \right).
\end{equation}
\item As a consequence of the previous items, we get
\begin{equation}
T_{ln} = O(\rho^4), \quad
T_{nm} = O(\rho^3).
\end{equation}

\end{enumerate}

\section{Gauge transformations preserving the TIRG \cite{Price,Price2}}
\label{sec:residualg}



In the main part of the paper, we require a characterization 
of the infinitesimal gauge transformations preserving the TIRG $h_{ll} = h_{ln} = h_{lm} = h_{m\bar{m}}=0$. Such an analysis has been carried out by \cite{Price,Price2}. It is recalled here---with minor modifications---for convenience, since some formulas of this analysis are required in the main text. The analysis is just as involved for a general non-accelerating type D spacetime as it is for Kerr, and the formulas below apply to this more general class of backgrounds. One requires all the NP components of the Lie derivative of such a metric with respect to a vector field $\xi^a$, given by
\begin{subequations}
\begin{align}
\label{Liexigll}
\left( \lie{\xi} g \right)_{ll} &= \ 2 \thorn \xi_l \\
\label{Liexigln}
\left( \lie{\xi} g \right)_{ln} &= \  \thorn \xi_n + \thorn' \xi_l + \left( \tau' + \bar{\tau} \right) \xi_m + \left( \tau + \bar{\tau}' \right) \xi_{\bar{m}}\\
\label{Liexiglm}
\left( \lie{\xi} g \right)_{lm} &=  \ \left( \eth + \bar{\tau}' \right) \xi_l + \left( \thorn + \bar{\rho} \right) \xi_{\bar m}\\
\label{Liexignn}
\left( \lie{\xi} g \right)_{nn} &= \ 2 \thorn' \xi_n\\
\label{Liexignm}
\left( \lie{\xi} g \right)_{nm} &= \ \left( \eth + \tau \right) \xi_n + \left( \thorn' + \rho' \right) \xi_m \\
\label{Liexigmm}
\left( \lie{\xi} g \right)_{mm} &= \ 2 \eth \xi_m \\
\label{Liexigmmb}
\left( \lie{\xi} g \right)_{m\bar{m}} &= \ \left( \rho'+ \bar{\rho}' \right) \xi_l + \left( \rho + \bar{\rho} \right) \xi_n + \eth' \xi_m - \eth \xi_{\bar{m}}.
\end{align}
\end{subequations}
For a residual vector field, we clearly must have $\left( \mathcal L_\xi g \right)_{ll} = \left( \lie{\xi} g \right)_{lm} = \left( \lie{\xi} g \right)_{ln} = \left( \lie{\xi} g \right)_{m\bar{m}} = 0$.
To analyze these conditions, it is very convenient to re-express the GHP quantities in terms of Held's, see \ref{sec:Held}. One finds
\begin{subequations}
\begin{align}
\label{Liexigll1}
\left( \lie{\xi} g \right)_{ll} 
=& \ 2 \thorn \xi_l,\\
\label{Liexigln1}
\left( \lie{\xi} g \right)_{ln} 
=& \ \Big[ \tilde{\thorn}' + \bar{\tau}^\circ \rho \bar{\rho} \tilde{\eth} + \tau^\circ \rho \bar{\rho} \tilde{\eth}' - \tau^\circ \bar{\tau}^\circ \rho \bar{\rho} \left( \rho + \bar{\rho} \right) - \half \left( \Psi^\circ \rho^2 + \bar{\Psi}^\circ \bar{\rho}^2 \right) \Big] \xi_l \nonumber \\
&+ \thorn \xi_n - \left( \bar{\tau}^\circ \rho^2 \bar{\rho} \Omega^\circ \right) \xi_m + \left( \tau^\circ \rho \bar{\rho}^2 \Omega^\circ \right) \xi_{\bar{m}},\\
\label{Liexiglm1}
\left( \lie{\xi} g \right)_{lm} 
=& \ \bar{\rho} \left( \tilde{\eth} - 2 \tau^\circ \bar{\rho} \right) \xi_l + \left( \thorn + \bar{\rho} \right) \xi_m,\\
\label{Liexignn1}
\left( \lie{\xi} g \right)_{nn} 
=& \ 2 \Big[ \tilde{\thorn}' + \bar{\tau}^\circ \rho \bar{\rho} \tilde{\eth} + \tau^\circ \rho \bar{\rho} \tilde{\eth}' + \tau^\circ \bar{\tau}^\circ \rho \bar{\rho} \left( \rho + \bar{\rho} \right) + \half \left( \Psi^\circ \rho^2 + \bar{\Psi}^\circ \bar{\rho}^2 \right) \Big] \xi_n,\\
\label{Liexignm1}
\left( \lie{\xi} g \right)_{nm} 
=& \ \bar{\rho} \Big[ \tilde{\eth} + \tau^\circ (\rho + \bar{\rho}) \Big] \xi_n\\
&+ \Big[ \tilde{\thorn}' + \bar{\tau}^\circ \rho \bar{\rho} \tilde{\eth} + \tau^\circ \rho \bar{\rho} \tilde{\eth}' + \rho \bar{\rho}^2 \tau^\circ \bar{\tau}^\circ - 2 \rho^2 \bar{\rho} \tau^\circ \bar{\tau}^\circ + \half \bar{\Psi}^\circ \bar{\rho}^2 + {\rho'}^\circ \bar{\rho}\nonumber \\
&- \Psi^\circ \rho^2 + \half \rho \bar{\rho} \left( \rho^{\prime \circ} + \bar{\rho}^{\prime \circ} \right) \Omega^\circ - \half  \rho \bar{\rho} \bar{\Psi}^\circ \Big] \xi_m,\nonumber
\\
\label{Liexigmm1}
\left( \lie{\xi} g \right)_{mm} 
=& \ 2 \bar{\rho} \left( \tilde{\eth} + \bar{\rho} \tau^\circ \right) \xi_m,\\
\label{Liexigmmb1}
\left( \lie{\xi} g \right)_{m\bar{m}} 
=& \ \left[ {\rho'}^\circ \bar{\rho} + \bar{\rho}^{\prime \circ} \rho - \half \left( \Psi^\circ \rho^2 + \bar{\Psi}^\circ \bar{\rho}^2 \right) - \half \left( \Psi^\circ + \bar{\Psi}^\circ \right) \rho \bar{\rho} - \tau^\circ \bar{\tau}^\circ \rho \bar{\rho} ( \rho + \bar{\rho} ) \right] \xi_l\nonumber\\
&+ \left( \rho + \bar{\rho} \right) \xi_n + \rho \left( \tilde{\eth}' - \bar{\tau}^\circ \rho \right) \xi_m + \bar{\rho} \left( \tilde{\eth} - \tau^\circ \bar{\rho} \right) \xi_{\bar{m}}. 
\end{align}
\end{subequations}
Given these expressions, if we demand now that $\left( \lie{\xi} g_{ab} \right) l^a = 0$, as is necessary for a residual gauge vector field, this leads to \cite{Price}:
\begin{equation}\label{eq:residual1}
\begin{split}
\xi_l =&\  \xi^\circ_l,\\
\xi_n =&\  \xi^\circ_n + \half \left( \frac{1}{\rho} + \frac{1}{\bar{\rho}} \right) \tilde{\thorn}' \xi^\circ_l + \rho \bar{\rho} \tau^\circ \bar{\tau}^\circ \xi^\circ_l + \half \left( \Psi^\circ \rho + \bar{\Psi}^\circ \bar{\rho} \right) \xi^\circ_l + \rho \bar{\tau}^\circ \Omega^\circ \xi^\circ_m\\
-& \bar{\rho} \tau^\circ \Omega^\circ \xi^\circ_{\bar{m}} - \rho \bar{\tau}^\circ \tilde{\eth} \xi^\circ_l - \bar{\rho} \tau^\circ \tilde{\eth}' \xi^\circ_l,\\
\xi_m =& \ \frac{1}{\bar{\rho}} \xi^\circ_m + \bar{\rho} \tau^\circ \xi^\circ_l - \tilde{\eth} \xi^\circ_l,
\end{split}
\end{equation}
for certain $\xi^\circ_l \circeq \GHPw{1}{1}$, $\xi^\circ_n \circeq \GHPw{-1}{-1}$ and $\xi^\circ_m \circeq \GHPw{2}{0}$, where we recall that a circle 
over a GHP scalar indicates that it is annihilated by $\th$. 

A residual gauge vector field must additionally satisfy $\left( \lie{\xi} g_{ab} \right) m^a \bar{m}^b = 0$. To analyze this condition, we 
use the formula for the Lie-derivate of $g_{ab}$ with respect to a general vector field $\xi^a$ of the form \eqref{eq:residual1}. This leads to 
\begin{equation}\label{HeldLiexigmmb}
\begin{split}
\left( \lie{\xi} g \right)_{m\bar{m}} =& \ \half \left( \frac{1}{\rho} + \frac{1}{\bar{\rho}} \right)^2 \rho \bar{\rho} \tilde{\thorn}' \xi^\circ_l + \left( {\rho'}^\circ \bar{\rho} + \bar{\rho}^{\prime \circ} \rho \right) \xi^\circ_l - \bar{\rho} \tilde{\eth} \tilde{\eth}' \xi^\circ_l - \rho \tilde{\eth}' \tilde{\eth} \xi^\circ_l\\
&+ \left( \rho + \bar{\rho} \right) \xi^\circ_n + \frac{\rho}{\bar{\rho}} \tilde{\eth}' \xi^\circ_m - \left( \rho + \bar{\rho} \right) \bar{\tau}^\circ \xi^\circ_m + \frac{\bar{\rho}}{\rho} \tilde{\eth} \xi^\circ_{\bar{m}} - \left( \rho + \bar{\rho} \right) \tau^\circ \xi^\circ_{\bar{m}} = 0.
\end{split}
\end{equation}
Expanding this equation in powers of $\rho, \bar \rho$ using \ref{sec:Held}, and then equating the coefficients annihilated by $\th$ to zero, one finds \cite{Price,Price2}:

\begin{theorem}
\label{thm:resgauge}
A gauge vector field $\xi^a$ on a non-accelerating type D spacetime preserves the TIRG if and only if $\xi^a$ can be written in the form \eqref{HeldLiexigmmb} with  GHP scalars $\xi^\circ_l \circeq \GHPw{1}{1}$, $\xi^\circ_n \circeq \GHPw{-1}{-1}$ and $\xi^\circ_m \circeq \GHPw{2}{0}$ satisfying
\begin{equation}
\label{eq:residcond}
\begin{split}
\tilde{\thorn}' \xi_l^\circ &= -\half \left(\tilde{\eth}' \xi_m^\circ + \tilde{\eth} \; \xi_{\bar{m}}^\circ \right),\\
\xi_n^\circ &= \half \left(\tilde{\eth}' \tilde{\eth} + \tilde{\eth} \;\tilde{\eth}' - \rho^{\prime\circ} -  \bar{\rho}^{\prime\circ} \right)\xi^\circ_l - \half \Omega^\circ \left(
\tilde{\eth}' \xi_m^\circ - \tilde{\eth} \;\xi_{\bar{m}}^\circ \right) . 
\end{split}
\end{equation}    
\end{theorem}

\section{Linearized Einstein operator in GHP \cite{Price,Green:2019nam}}
\label{app:LinEinGHP}

The linearized Einstein operator $\E$ in a Petrov type II spacetime was given in GHP form
by \cite{Price}, and by \cite{Green:2019nam}, correcting some typos. Specializing to non-accelerating type D, 
we get the following expressions.
\begin{equation}
\begin{split}
(\E h)_{ll} &= \Big[(\eth'-\tau')(\eth-\bar{\tau}') + \rho(\thorn'+\rho'-\bar{\rho}') -
(\thorn-\rho)\rho' + \Psi_2\Big]h_{ll}\\
&\pheq + \Big[-(\rho+\bar{\rho})(\thorn+\rho+\bar{\rho}) +4\rho\bar{\rho}\Big]h_{ln}\\
&\pheq + \Big[-(\thorn-3\bar{\rho})(\eth'-\tau'+\bar{\tau}) + \bar{\tau}\thorn-\bar{\rho}\eth'\Big]h_{lm}\\
&\pheq + \Big[-(\thorn-3\rho)(\eth+\tau-\bar{\tau}') + \tau\thorn-\rho\eth\Big]h_{l\bar{m}}\\
&\pheq + \Big[\thorn(\thorn-\rho-\bar{\rho})+2\rho\bar{\rho}\Big]h_{m\bar{m}},
\end{split}
\label{eqn:Ell_app}
\end{equation}

\begin{equation}
\begin{split}
(\E h)_{nn} 
&= \pheq  \Big[(\eth'-\bar{\tau})(\eth-\tau) +\bar{\rho}'(\thorn-\rho+\bar{\rho}) -
(\thorn'-\bar{\rho}')\bar{\rho} + {\bar\Psi}_2\Big]h_{nn}\\
&\pheq + \Big[-(\rho'+\bar{\rho}')(\thorn'+\rho'+\bar{\rho}') + 4\rho'\bar{\rho}'
 \Big]h_{ln}\\
&\pheq + \Big[-(\thorn'-3\rho')(\eth'+\tau'-\bar{\tau}) + \tau'\thorn' - \rho'\eth'\Big]h_{nm}\\
&\pheq + \Big[-(\thorn'-3\bar{\rho}')(\eth+\bar{\tau}'-\tau) + \bar{\tau}'\thorn' - \bar{\rho}'\eth\Big]h_{n\bar{m}}\\
&\pheq + \Big[\thorn'(\thorn'-\rho'-\bar{\rho}') + 2\rho'\bar{\rho}'\Big]h_{m\bar{m}},
\end{split}
\end{equation}

\begin{equation}
\begin{split}
(\E h)_{ln} &= \half\Big[\rho'(\thorn'-\rho') + \bar{\rho}'(\thorn'-\bar{\rho}')\Big]h_{ll}\\
&\pheq + \half\Big[\rho(\thorn-\rho) + \bar{\rho}(\thorn-\bar{\rho})\Big]h_{nn}\\
&\pheq + \half\Big[-(\eth'+\tau'+\bar{\tau})(\eth-\tau-\bar{\tau}') -
(\eth'\eth+3\tau\tau'+3\bar{\tau}\bar{\tau}') + 2(\bar{\tau}+\tau')\eth\\
&\pheq\pheq+ (\thorn-2\bar{\rho})\rho' +(\thorn'-2\rho')\bar{\rho} - \bar{\rho}'(\thorn+\rho) -
\rho(\thorn'+\bar{\rho}') -\Psi_2 -{\bar\Psi}_2\Big]h_{ln}\\
&\pheq + \half\Big[(\thorn'-2\bar{\rho}')(\eth'-\tau') + \bar{\tau}(\thorn'+\rho'+\bar{\rho}')
-\tau'(\thorn'-\rho')\\
&\pheq \pheq-(2\eth'-\bar{\tau})\bar{\rho}'\Big]h_{lm}\\
&\pheq + \half\Big[(\thorn'-2\rho')(\eth-\bar{\tau}') + \tau(\thorn'+\bar{\rho}'+\rho')
-\bar{\tau}'(\thorn'-\bar{\rho}')\\
&\pheq \pheq-(2\eth-\tau)\rho'\Big]h_{l\bar{m}}\\
&\pheq + \half\Big[(\thorn-2\rho)(\eth'-\bar{\tau}) + (\tau'+\bar{\tau})(\thorn+\bar{\rho})
-2(\eth'-\tau')\rho-2\bar{\tau}\thorn\Big]h_{nm}\\
&\pheq + \half\Big[(\thorn-2\bar{\rho})(\eth-\tau) + (\bar{\tau}'+\tau)(\thorn+\rho)
-2(\eth-\bar{\tau}')\bar{\rho}-2\tau\thorn\Big]h_{n\bar{m}}\\
&\pheq + \half\Big[-(\eth'-\bar{\tau})(\eth'-\tau') +
\bar{\tau}(\bar{\tau}-\tau')\Big]h_{mm}\\
&\pheq + \half\Big[-(\eth-\tau)(\eth-\bar{\tau}') +
\tau(\tau-\bar{\tau}')\Big]h_{\bar{m}\bar{m}}\\
&\pheq + \half\Big[(\eth'+\tau'-\bar{\tau})(\eth-\tau+\bar{\tau}') +
(\eth'\eth-\tau\tau'-\bar{\tau}\bar{\tau}'+\tau\bar{\tau}) - (\Psi_2+{\bar\Psi}_2)\\
&\pheq \pheq+(\thorn'-2\rho')\bar{\rho} + (\thorn-2\bar{\rho})\rho' +\rho(3\thorn'-2\bar{\rho}')
+\bar{\rho}'(3\thorn-2\rho)\\
&\pheq\pheq -2\thorn'\thorn + 2\rho\bar{\rho}' +2\eth'(\tau)-\tau\bar{\tau}\Big]h_{m\bar{m}},
\end{split}
\label{eqn:Eln_app}
\end{equation}

\begin{equation}
\begin{split}
(\E h)_{lm} &= \half\Big[(\thorn'-\rho')(\eth-\bar{\tau}')
+(\eth-\tau-2\bar{\tau}')\bar{\rho}' -(\eth-\tau)\rho' +\tau(\thorn'+\rho')\Big]h_{ll}\\
&\pheq +\half\Big[-(\thorn-\rho+\bar{\rho})(\eth+\tau-\bar{\tau}') -
(\eth-3\tau+\bar{\tau}')\bar{\rho} - 2\rho\bar{\tau}'\Big]h_{ln}\\
&\pheq +\half\Big[-(\thorn'+\bar{\rho}')(\thorn-2\bar{\rho}) + \rho(\thorn'+2\rho'-2\bar{\rho}') -
4\rho'\bar{\rho} +2\Psi_2\\
&\pheq\pheq + (\eth'+\bar{\tau})(\eth-2\bar{\tau}') - \tau(\eth'+\tau'-2\bar{\tau})
-\tau'(\tau-4\bar{\tau}')\Big]h_{lm}\\
&\pheq + \half\Big[-\eth(\eth-2\tau)  -
2\bar{\tau}'(\tau-\bar{\tau}')\Big]h_{l\bar{m}}\\
&\pheq + \half\Big[\thorn(\thorn-2\rho) + 2\bar{\rho}(\rho-\bar{\rho})\Big]h_{nm}\\
&\pheq + \half\Big[-(\thorn-\bar{\rho})(\eth'-\tau'+\bar{\tau}) + 2\bar{\tau}\bar{\rho}\Big]h_{mm}\\
&\pheq + \half\Big[(\thorn+\rho-\bar{\rho})(\eth+\bar{\tau}'-\tau) + 2\bar{\tau}'(\thorn-2\rho) -
(\eth-\tau-\bar{\tau}')\bar{\rho} +2\rho\tau\Big]h_{m\bar{m}},
\end{split}
\label{eqn:Elm_app}
\end{equation}

\begin{equation}
\begin{split}
(\E h)_{n\bar{m}} &= 
\pheq  \half\Big[(\thorn-\rho+\bar{\rho})(\eth'-\bar{\tau}) - (\eth'-2\tau'+\bar{\tau})\rho +
\tau'(\thorn-\bar{\rho})\Big]h_{nn}\\
&\pheq + \half\Big[(-(\thorn'-\rho'+\bar{\rho}')(\eth'+\tau'-\bar{\tau}) -
(\eth'-3\tau'+\bar{\tau})\bar{\rho}' -2\rho'\bar{\tau}\Big]h_{ln}\\
&\pheq +\half\Big[(\thorn'(\thorn'-2\rho') +2\bar{\rho}'(\rho'-\bar{\rho}')\Big]h_{l\bar{m}}\\
&\pheq +\half\Big[-\eth'(\eth'-2\tau') -
2\bar{\tau}(\tau'-\bar{\tau})\Big]h_{nm}\\
&\pheq +\half\Big[-(\thorn+\bar{\rho})(\thorn'-2\bar{\rho}') + \rho'(\thorn+2\rho-2\bar{\rho}) -
4\rho\bar{\rho}' + 2\Psi_2\\
&\pheq\pheq + (\eth+\bar{\tau}')(\eth'-2\bar{\tau}) - \tau'(\eth+\tau-2\bar{\tau}') -
\tau(\tau'-4\bar{\tau})\Big]h_{n\bar{m}}\\
&\pheq +\half\Big[-(\thorn'-\bar{\rho}')(\eth-\tau+\bar{\tau}') +2\bar{\tau}'\bar{\rho}'\Big]h_{\bar{m}\bar{m}}\\
&\pheq +\half\Big[(\thorn'+\rho'-\bar{\rho}')(\eth'-\tau'+\bar{\tau}) + 2\bar{\tau}(\thorn'-2\rho')
-(\eth'-\tau'-\bar{\tau})\bar{\rho}'+2\rho'\tau'\Big]h_{m\bar{m}} \\
\end{split}
\end{equation}

\begin{equation}
\begin{split}
(\E h)_{mm} &= 
\pheq \Big[-\eth(\eth-\tau-\bar{\tau}') -2\tau\bar{\tau}'\Big]h_{ln}\\
&\pheq +\Big[(\thorn'-\rho')(\eth-\bar{\tau}') - (\eth-\tau-\bar{\tau}')\rho' +
\tau(\thorn'+\rho'-\bar{\rho}') -\bar{\tau}'(\thorn+\bar{\rho}')\Big]h_{lm}\\
&\pheq +\Big[(\thorn-\bar{\rho})(\eth-\tau) - (\eth-\tau-\bar{\tau}')\bar{\rho} -\tau(\thorn+\rho)
+ \bar{\tau}'(\thorn-\rho+\bar{\rho})\Big]h_{nm}\\
&\pheq +\Big[-(\thorn'-\rho')(\thorn-\bar{\rho}) + (\eth-\tau)\tau' -
\tau(\eth'+\tau'-\bar{\tau}) + \Psi_2\Big]h_{mm}\\
&\pheq +\Big[(\tau+\bar{\tau}')\eth +
(\tau-\bar{\tau}')^2\Big]h_{m\bar{m}},
\end{split}
\end{equation}

\begin{equation}
\begin{split}
(\E h)_{m\bar{m}} &= \half\Big[\thorn'(\thorn'-\rho'-\bar{\rho}') + 2\rho'\bar{\rho}' \Big]h_{ll}\\
&\pheq +\half\Big[\thorn(\thorn-\rho-\bar{\rho}) + 2\rho\bar{\rho}\Big]h_{nn}\\
&\pheq +\half\Big[-(\thorn'+\rho'-\bar{\rho}')(\thorn-\rho+\bar{\rho}) -\thorn'(\thorn+\rho)
+\rho(\thorn'+\rho'-\bar{\rho}') -{\bar\Psi}_2\\
&\pheq\pheq+(\eth'-\bar{\tau})(\eth-\tau-\bar{\tau}') + \eth'\eth -
(\eth-2\bar{\tau}')\tau' -
\bar{\tau}(2\eth+\bar{\tau}')\\
&\pheq\pheq-2\tau(\eth'-\bar{\tau})+2\tau'\bar{\tau}'+\bar{\rho}\bar{\rho}'\Big]h_{ln}\\
&\pheq +\half\Big[-(\thorn'-2\rho')(\eth'-2\bar{\tau}) + \bar{\tau}(\thorn'+2\rho'-2\bar{\rho}')-2\tau'\bar{\rho}'\Big]h_{lm}\\
&\pheq +\half\Big[-(\thorn'-2\bar{\rho}')(\eth-2\tau) + \tau(\thorn'+2\bar{\rho}'-2\rho') -2\bar{\tau}'\rho'\Big]h_{l\bar{m}}\\
&\pheq +\half\Big[-(\thorn-2\bar{\rho})(\eth'-2\tau') + \tau'(\thorn-2\rho-2\bar{\rho}) -
2\rho\bar{\tau}+4\tau'\bar{\rho}\Big]h_{nm}\\
&\pheq +\half\Big[-(\thorn-2\rho)(\eth-2\bar{\tau}') + \bar{\tau}'(\thorn-2\bar{\rho}-2\rho) -
2\bar{\rho}\tau+4\bar{\tau}'\rho\Big]h_{n\bar{m}}\\
&\pheq
+\half\Big[-\bar{\tau}(\eth'-\bar{\tau})-\tau'(\eth'-\tau')
\Big]h_{mm}\\
&\pheq +\half\Big[-\tau(\eth-\tau)-\bar{\tau}'(\eth-\bar{\tau}')
\Big]h_{\bar{m}\bar{m}}\\
&\pheq + \half\Big[2\thorn'\thorn -(\thorn'-\bar{\rho}')\bar{\rho} - (\thorn-\rho)\rho'
-\rho(\thorn'-\rho'+\bar{\rho}') - \bar{\rho}'(\thorn+\rho-\bar{\rho})\\
&\pheq\pheq -(\eth'-2\tau')\bar{\tau}' + \tau(\eth'+2\bar{\tau}) -
\tau'(\eth-\bar{\tau}') + \bar{\tau}(\eth+\tau) -\eth'(\tau)\\
&\pheq\pheq -\Psi_2-{\bar\Psi}_2\Big]h_{m\bar{m}}.
\end{split}
\end{equation}

\section{Teukolsky operators in GHP \cite{Price,StewardWalker}}
\label{sec:appO}

\noindent
{\bf Spin-2:}
$\O$ takes a $\GHPw{4}{0}$ GHP scalar to another $\GHPw{4}{0}$ GHP scalar and is 
\begin{equation}\label{eq:O}
  \O \eta = 2 \Big[ (\thorn - 4\rho - \bar \rho)(\thorn'-\rho') - (\eth - 4\tau - \bar \tau')(\eth' - \tau') -3\Psi_2 \Big] \eta.
 \end{equation}
  The adjoint Teukolsky operator $\O^\dagger$ takes a $\GHPw{-4}{0}$ GHP scalar to another $\GHPw{-4}{0}$ GHP scalar and is 
\begin{equation}
  \label{eq:Odag}
  \O^\dagger \Phi =  2\Big[ (\thorn'-\bar \rho')(\thorn + 3 \rho) - (\eth' - \bar \tau)(\eth +3\tau) -3\Psi_2 \Big] \Phi.
\end{equation}

\noindent
{\bf Spin-1:}
$\O$ takes a $\GHPw{2}{0}$ GHP scalar to another $\GHPw{2}{0}$ GHP scalar and is 
\begin{subequations}
\begin{align}
\label{Odef}
\mathcal{O} \eta =&2  \Big[ \left( \thorn - \bar{\rho} - 2\rho \right) \left( \thorn' -  \rho' \right) - \left( \eth - \bar{\tau}' - 2\tau \right) \left( \eth' - \tau' \right) \Big] \eta.
\end{align}
\end{subequations}
  The adjoint Teukolsky operator $\O^\dagger$ takes a $\GHPw{-2}{0}$ GHP scalar to another $\GHPw{-2}{0}$ GHP scalar and is 
\begin{subequations}
\begin{align}
\label{Odaggerdef}
\mathcal{O}^\dagger \Phi =&2  \Big[ \left( \thorn' - \bar{\rho}' \right) \left( \thorn + \rho \right) - \left( \eth' - \bar{\tau} \right) \left( \eth + \tau \right) \Big] \Phi.
\end{align}
\end{subequations}

\section{Reconstruction operators in GHP \cite{Price,StewardWalker}}
\label{app:ST}
{\bf Spin-2:}
The operators $\S, \T$ appearing in Teukolsky's master equation for spin-2 perturbations act on symmetric
rank-2 covariant tensor fields and give GHP scalars of weights
$\GHPw{4}{0}$,
\begin{subequations}
  \begin{align}\label{eq:S}
    \S T ={}& (\eth - \bar \tau' - 4\tau)\Big[(\thorn - 2\bar \rho) T_{lm} - (\eth - \bar \tau') T_{ll} \Big] \nonumber \\
            & + (\thorn - \bar \rho - 4\rho)\Big[(\eth - 2\bar \tau') T_{lm} - (\thorn - \bar \rho) T_{mm} \Big], \\
    \label{eq:T}\T h ={}& \half(\eth-\bar \tau')(\eth-\bar \tau') h_{ll} + \half (\thorn - \bar \rho)(\thorn - \bar \rho) h_{mm} \nonumber \\
            &-\half\Big[(\thorn- \bar \rho)(\eth-2\bar \tau') + (\eth -\bar \tau')(\thorn -2\bar \rho)\Big] h_{(lm)}.
  \end{align}
\end{subequations}
The adjoint operators $\T^\dagger, \S^\dagger$ act on GHP scalars of weights $\GHPw{-4}{0}$ and produce a contravariant symmetric 
rank-2 tensor, 
\begin{subequations}
\begin{align}
\label{eq:Tdag}
  \left(\T^\dagger \eta\right)^{ab} ={}
  & \half l^a l^b(\eth-\tau)(\eth-\tau) \eta + \half m^a m^b (\thorn-\rho)(\thorn-\rho) \eta \nonumber\\
  & - \half l^{(a} m^{b)}\Big[(\eth + \bar{\tau}' -\tau)(\thorn - \rho) + (\thorn - \rho + \bar \rho)(\eth-\tau) \Big] \eta, \\
\label{eq:Sdag}
  \left(\S^\dagger \Phi \right)^{ab} ={}
  & - l^a l^b(\eth-\tau)(\eth+3\tau) \Phi - m^a m^b (\thorn - \rho)(\thorn + 3\rho) \Phi \nonumber \\
  & + l^{(a} m^{b)}\Big[(\thorn - \rho + \bar \rho)(\eth + 3\tau) + (\eth-\tau+\bar \tau')(\thorn +3\rho) \Big] \Phi . 
\end{align}
\end{subequations}

\noindent
{\bf Spin-1:}
The operators $\S, \T$ appearing in Teukolsky's master equation for spin-1 perturbations act on symmetric
rank-1 covariant tensor fields and give GHP scalars of weights
$\GHPw{2}{0}$,
\begin{subequations}
\begin{align}
\label{Tdef}
\mathcal{T} A =& \left( \eth - \bar{\tau}' \right) A_l - \left( \thorn - \bar{\rho} \right) A_m ,\\
\label{Sdef}
\mathcal{S} J =& \half \Big[ \left( \eth - \bar{\tau}' - 2\tau \right) J_l - \left( \thorn - \bar{\rho} - 2\rho \right) J_m  \Big].
\end{align}
\end{subequations}
The adjoint operators $\T^\dagger, \S^\dagger$ act on GHP scalars of weights $\GHPw{-2}{0}$ and produce a contravariant 
rank-1 tensor,
\begin{subequations}
\begin{align}
\label{Tdaggerdef}
\left( \mathcal{T}^\dagger \eta \right)^a =& \Big[ m^a \left( \thorn - \rho \right) - l^a \left( \eth - \tau \right) \Big] \eta,\\
\label{Sdaggerdef}
\left(\mathcal{S}^\dagger \Phi \right)^a =& \half \Big[ m^a \left( \thorn + \rho \right) - l^a \left( \eth + \tau \right) \Big] \Phi.
\end{align}
\end{subequations}

\section{Calculus with $Z,X,B,C$}
\label{sec:bivec}

Here we give some formulas involving the self-dual 2-forms $Z,Z',X$ as defined in \eqref{eq:XZZp}.

\begin{enumerate}
\item For a 1-form $G$ we have the following formula for its exterior differential.
\begin{equation}
\begin{split}
\intd G =& \ Z' \bigg[ (\thorn - \bar{\rho}) G_{m} - (\eth - \bar{\tau}') G_l \bigg]\\
+& \bar{Z'} \bigg[ (\thorn - \rho) G_{\bar{m}} - (\eth' - \tau') G_l \bigg]\\
+& \half X \bigg[ (\thorn + \rho - \bar{\rho}) G_n - (\thorn' + \rho' - \bar{\rho}') G_l - (\eth + \tau - \bar{\tau}') G_{\bar{m}} + (\eth' + \tau' - \bar{\tau}) G_m \bigg]\\
+& \half \bar{X} \bigg[ (\thorn + \bar{\rho} - \rho) G_n - (\thorn' + \bar{\rho}' - \rho') G_l + (\eth + \bar{\tau}' - \tau) G_{\bar{m}} - (\eth' + \tau' - \bar{\tau}) G_m \bigg]\\
+& Z \bigg[ (\eth' - \bar{\tau}) G_n - (\thorn' - \bar{\rho}') G_{\bar{m}} \bigg]\\
+& \bar{Z} \bigg[ (\eth - \tau) G_n - (\thorn' - \rho') G_{m} \bigg].
\end{split}
\end{equation}

\item We have the following formulas for the Hodge-dual of the NP tetrad legs:
\begin{equation}
\begin{split}
\star l =&\  -i \bar{m} \wedge Z = -i l \wedge X,\\
\star n =& \ i m \wedge Z' = i n \wedge X,\\
\star m =& \ -i n \wedge Z = -i m \wedge X,\\
\star \bar{m} =& \ i l \wedge Z' = i \bar{m} \wedge X.
\end{split}
\end{equation}

\item Letting $\intd_\Theta$ be the twisted exterior differential on GHP weighted forms, defined by replacing $\nabla_a$ by $\Theta_a$ in its definition, one may show that
\begin{equation}
\begin{split}
\intd_\Theta l =& \ - \bar{\tau} Z +\half (\rho - \bar{\rho}) X + \text{c.c.},\\
\intd_\Theta n =& \ \bar{\tau}' Z' - \half (\rho' - \bar{\rho}') X + \text{c.c.},\\
\intd_\Theta m =& \ (\tau - \bar{\tau}') \frac{X + \bar{X}}{2} + \rho \bar{Z'} - \bar{\rho}' Z,\\
\intd_\Theta \bar{m} =& \ - (\tau' - \bar{\tau}) \frac{X + \bar{X}}{2} + \bar{\rho} Z' - \rho' \bar{Z},
\end{split}
\end{equation}
and that 
\begin{equation}
\begin{split}
\intd_\Theta^\dagger l =& \ \rho + \bar{\rho},\\
\intd_\Theta^\dagger n =& \ \rho' + \bar{\rho}',\\
\intd_\Theta^\dagger m =& \ \tau + \bar{\tau}',\\
\intd_\Theta^\dagger \bar{m} =& \ \tau' + \bar{\tau}.
\end{split}
\end{equation}

\item One has equations of structure
\begin{equation}
\begin{split}
\intd_\Theta Z' =& \ B' \wedge Z',\\
\intd_\Theta X =& \ 2 (B + B') \wedge X,\\
\intd_\Theta Z =& \ B \wedge Z,
\end{split}
\end{equation}
and 
\begin{equation}
\begin{split}
\intd_\Theta^\dagger Z' =&\ - {B'} \cdot Z',\\
\intd_\Theta^\dagger X =& \- 2 ({B + B'}) \cdot X,\\
\intd_\Theta^\dagger Z =& \ - B \cdot Z,
\end{split}
\end{equation}
where
\begin{equation}
\begin{split}
B =& \-\rho n + \tau \bar{m},\\
C =& \-\rho m + \tau l.
\end{split}
\end{equation}

\item One may check that
\begin{equation}
\begin{split}
\star B =&\ i B \wedge Y,\\
\star B' =&\ - i B' \wedge Y,
\end{split}
\end{equation}
\begin{equation}
\zeta^{-1} \intd \zeta = B + B',
\end{equation}
where $\zeta = \Psi_2^{-\tfrac{1}{3}}$. A closely related identity is $ \Theta_a \Psi_2 = -3 (B_a + B'_a) \Psi_2$.

\item
$B_a, C_a$ are complex null vectors, with the following further properties.
\begin{equation}
  B \cdot B' = \rho \rho' - \tau \tau' = - C \cdot C'.
\end{equation}
The twisted divergence is given by
\begin{subequations}
\begin{align}
  \intd_\Theta^\dagger B =&\ \rho \rho' - \tau\tau' + \Psi_2,\\
  \intd_\Theta^\dagger C =& \ 0.
\end{align}
\end{subequations}
The twisted exterior differential is given by
\begin{subequations}\label{eq:dB}
\begin{align}
  \intd_\Theta B =& \
  \bar{Z'}(\eth'\rho) + \bar{Z}(\thorn'\tau) + \bar{X} \left(\eth' \tau - \rho \bar\rho' - \half \Psi_2\right) - \half\Psi_2 X + 2 C \wedge C',\\
\intd_\Theta C =& \Psi_2 Z.
 \end{align}
\end{subequations}
Contractions with $Z,Z',X$:
\begin{IEEEeqnarray}{rClCrCl}
  Z^{\prime ab} B_b &=& 0, &\qquad & Z^{\prime ab} C_b &=& B^a, \\
  X^{ab} B_b &=& B^a, &\qquad & X^{ab} C_b &=& - C^a, \\
  Z^{ab} B_b &=& -C^a, &\qquad & Z^{ab} C_b &=& 0.
\end{IEEEeqnarray}
\end{enumerate}


\section{Construction of residual gauge vector field $\zeta^a$}
\label{sec:zetaproof}

Here we construct a solution $(d^\circ, e^\circ, \zeta^\circ_l, \zeta^\circ_n, \zeta^\circ_m)$ to Eqs. \eqref{eq:residcond2}, \eqref{eq:c2}, \eqref{eq:ecircdet}, \eqref{eq:dcircdet}. The GHP scalars $(\zeta^\circ_l, \zeta^\circ_n, \zeta^\circ_m)$ combined then give the residual gauge vector field $\zeta^a$ via  \eqref{eq:residual1}.

Recall that a degree mark $\circ$ on any quantity means that it is a GHP scalar annihilated by $\thorn$, so such a quantity 
does not depend upon $r$ in retarded KN coordinates and the Kinnersley frame. 
Eqs. \eqref{eq:residcond2}, \eqref{eq:c2}, \eqref{eq:ecircdet}, \eqref{eq:dcircdet} are coupled and, as a system of PDEs in the remaining retarded KN coordinates $(u,\theta,\varphi_*)$, not of any particular canonical PDE-type, except in Schwarzschild spacetime, where they can be reduced to solving a single, fourth order, parabolic equation \cite{Green:2019nam}. Thus, it is far from obvious that a solution should exist in Kerr, and even if it does, how to construct it.\footnote{However, one can make the crude observation that, merely counting the number of functions to be determined and the number of equations, there is a certain freedom left to fix basically one free real valued GHP scalar, which basically corresponds to $f^\circ$ in lemma \ref{lem:9}.}  

Fortunately it turns out that there is a non-obvious simplifying structure behind these equations which we uncover in lemma \ref{lem:9}, and which is ultimately based on the non-accelerating type D geometry of Kerr. This structure, featuring a new auxiliary GHP scalar, $f^\circ \circeq \GHPw{4}{0}$, will then be combined with the TS identities to obtain the desired solution in proposition \ref{prop:13}.

\begin{lemma}
\label{lem:9}
Suppose $(d^\circ, e^\circ, f^\circ, \zeta^\circ_l, \zeta^\circ_n, \zeta^\circ_m)$ are smooth, 
$(e^\circ, f^\circ, \zeta^\circ_l, \zeta^\circ_n, \zeta^\circ_m)$ and $\tilde{\thorn}' d^\circ$ tend to zero as $u \to -\infty$ with their
Held derivatives $\tilde{\thorn}', \tilde{\eth}, \tilde{\eth}'$, and satisfy the following equations:
\begin{subequations}

\begin{equation}\label{eq:f}
\tilde{\eth}'\tilde{\eth}'\tilde{\eth}'\tilde{\eth}' f^\circ  - 3 \Psi^\circ \tilde{\thorn}' \bar f^\circ = 
\tilde{\thorn}' \bar h^{1\circ}_{mm},
\end{equation}

\begin{equation}\label{eq:zm}
\zeta^\circ_m = \tilde{\eth}' f^\circ,
\end{equation}

\begin{equation}\label{eq:zl}
\tilde \thorn' \zeta_l^\circ = -\Re \left(\tilde{\eth}' \zeta_m^\circ \right),
\end{equation}

\begin{equation}\label{eq:zn}
\zeta_n^\circ = \Re \Big[ \left(\tilde{\eth}' \tilde{\eth} - \rho^{\prime\circ} \right)\zeta^\circ_l - \Omega^\circ \tilde{\eth}' \zeta_m^\circ + 2 \bar{\tau}^\circ \zeta^\circ_m \Big],
\end{equation}

\begin{equation}\label{eq:e}
e^\circ = -2\tilde{\thorn}' \bar f^\circ,
\end{equation}

\begin{equation}\label{eq:d}
\tilde{\thorn}' d^\circ = \frac{1}{3} \left( 2  \tilde{\eth}' \tilde{\eth}' \tilde{\eth}' \tilde{\eth} \,\zeta^\circ_l -2 \Omega^\circ \tilde{\eth}' \tilde{\eth}' \tilde{\eth}' \zeta^\circ_m +6 \bar{\tau}^\circ \tilde{\eth}' \tilde{\eth}' \zeta^\circ_m -6 \bar{\tau}^\circ \tilde{\eth}' \tilde{\eth} \zeta^\circ_{\bar{m}} + \tilde{\eth}' \tilde{\eth} \bar{h}^{1 \circ}_{mm} +2 \rho^{\prime \circ} \bar{h}^{1 \circ}_{mm}  \right) \ .
\end{equation}

\end{subequations}
%
Then we can adjust $d^\circ \to d^\circ +\delta^\circ$, with $\tilde{\thorn}' \delta^\circ = 0$, so that 
$(d^\circ, e^\circ, \zeta^\circ_l, \zeta^\circ_n, \zeta^\circ_m)$ solve Eqs. \eqref{eq:residcond2}, \eqref{eq:c2}, \eqref{eq:ecircdet}, \eqref{eq:dcircdet}.
\end{lemma}
\begin{proof}
The proof of the lemma hinges on the special properties of two auxiliary GHP scalars, $p^\circ \circeq \GHPw{-5}{-1}$, $k^\circ \circeq \GHPw{-4}{0}$, 
defined by 
\begin{equation}
\label{Kcircdef}
k^\circ := - 4 \tilde{\eth}' \tilde{\eth}' \tilde{\eth}' \zeta^\circ_{m} + 4 \tilde{\thorn}' \bar{h}^{1 \circ}_{mm} - 6 \Psi^\circ e^\circ,
\end{equation}
and
\begin{multline}
\label{Pcircdef}
p^\circ := - 4 \tilde{\eth}' \tilde{\eth}' \tilde{\eth}' \tilde{\eth} \, \zeta^\circ_l + 4 \Omega^\circ \tilde{\eth}' \tilde{\eth}' \tilde{\eth}' \zeta^\circ_m - 12 \bar{\tau}^\circ \tilde{\eth}' \tilde{\eth}' \zeta^\circ_m + 12 \bar{\tau}^\circ \tilde{\eth}' \tilde{\eth} \zeta^\circ_{\bar{m}} - 2 \tilde{\eth}' \tilde{\eth} \bar{h}^{1 \circ}_{mm}\\
- 4 \rho^{\prime \circ} \bar{h}^{1 \circ}_{mm} + 6 \tilde{\thorn}' d^\circ.
\end{multline}
Obviously, $p^\circ=0$ is equivalent to \eqref{eq:d}. Furthermore, substituting equations \eqref{eq:zm}, \eqref{eq:e} into the definition of $k^\circ$, 
we get $k^\circ = 0$ from  \eqref{eq:f}. 

Next, we define the GHP scalar $x^\circ$ by
\begin{equation}
\label{eq:x}
\begin{split}
x^\circ := \ & \tilde{\thorn}' c^\circ -  \Re \bigg\{ \half \left(\tilde{\eth}' \tilde{\eth}' \tilde{\eth} \; \tilde{\eth} + \tilde{\eth} \; \tilde{\eth} \; \tilde{\eth}' \tilde{\eth}' \right)  \tilde{\eth}'
\zeta^\circ_m \\
&
+ \bigg[  \Omega^\circ \left(\tilde{\eth}' \tilde{\eth} + \tilde{\eth} \, \tilde{\eth}' \right) + 2 \tau^\circ \tilde{\eth}'  - 2 \bar{\tau}^\circ \tilde{\eth}  + \half \Psi^\circ - \frac{7}{2} \bar{\Psi}^\circ \bigg]\tilde{\thorn}' \tilde{\eth}'
\zeta^\circ_m \bigg\}.
\end{split}
\end{equation}
Using lemma \ref{lem:11}, using the type D identities as encoded in Held's formalism (recalled in \ref{sec:Held}), and using the asymptotic expansion of the linearized EE in the form $({\mathcal E}h)_{ab} = S_{ab}$ (see Sec. \ref{sec:xabdef}), we get after a very lengthy calculation:
\begin{equation}\label{eq: ReEthEthK0}
\Re \left( \tilde{\eth} \tilde{\eth} k^\circ \right) = x^\circ + 2 \tilde{\thorn}' \Re \left( \tilde{\eth}' \tilde{\eth}' h^{1 \circ}_{mm} - h^{1 \circ}_{nn} \right) = x^\circ + 2 S^{2 \circ}_{nn} - 4 \Re \left( \tilde{\eth}' S^{2 \circ}_{nm} \right).
\end{equation}
By lemma \ref{lem:S}, we know that $S^{2 \circ}_{nn}=S^{2 \circ}_{nm}=0$, so $k^\circ = 0$ gives $x^\circ = 0$ in view of \eqref{eq: ReEthEthK0}. Substituting \eqref{eq:zl} into $x^\circ$ in \eqref{eq:x}, we can pull out $\tilde{\thorn}'$ from all expressions and get
\begin{equation}
\label{eq:x=0}
\begin{split}
0 = & \ \tilde{\thorn}' \bigg\{ c^\circ -  \Re \bigg[ -\half \left(\tilde{\eth}' \tilde{\eth}' \tilde{\eth} \; \tilde{\eth} + \tilde{\eth} \; \tilde{\eth} \; \tilde{\eth}' \tilde{\eth}' \right) \zeta^\circ_l \\
&+ \bigg(  \Omega^\circ \left(\tilde{\eth}' \tilde{\eth} + \tilde{\eth} \, \tilde{\eth}' \right) + 2 \tau^\circ \tilde{\eth}'  - 2 \bar{\tau}^\circ \tilde{\eth}  + \half \Psi^\circ - \frac{7}{2} \bar{\Psi}^\circ \bigg)  \tilde{\eth}' \zeta^\circ_m \bigg] \bigg\}.
\end{split}
\end{equation}
By assumption, all quantities in the curly brackets acted upon by $\tilde{\thorn}'$ are decaying for $u \to -\infty$. Thus, we may effectively remove $\tilde{\thorn}'$ (which is $= \partial_u$ in KN coordinates and the Kinnersley frame) form the entire equation. This gives us precisely \eqref{eq:c2}.

Next, we define the auxiliary GHP scalar $y^\circ$ by
\begin{equation}
\label{eq:ycirc1}
\begin{split}
y^\circ :=& \, - \tilde{\eth} \, d^\circ  - \bar{h}^{2 \circ}_{nm} + \tau^\circ \bar{h}^{1 \circ}_{mm} + \half \Omega^\circ \tilde{\eth} \bar{h}^{1 \circ}_{mm} + 2 \Omega^\circ \tilde{\eth}' \tilde{\eth}' \tilde{\eth} \; \zeta^\circ_l \nonumber - 2 \Omega^{2 \circ} \tilde{\eth}' \tilde{\eth}' \zeta^\circ_m \\
& \, - 2 \left(2 \Omega^\circ \bar{\rho}^{\prime \circ} - \bar{\Psi}^\circ \right) \tilde{\eth}' \zeta^\circ_l 
- 2 \Omega^\circ \bar{\tau}^\circ \tilde{\eth} \; \zeta^\circ_{\bar{m}} + 2 \Omega^\circ \bar{\tau}^\circ \tilde{\eth}' \zeta^\circ_m + 2 \Omega^\circ \Psi^\circ \zeta^\circ_{\bar{m}}.
\end{split}
\end{equation}
Using lemma \ref{lem:11}, using the type D identities as encoded in Held's formalism, and using the asymptotic expansion of the linearized EE in the form $(\mathcal{E} h)_{ab} = S_{ab}$, we get after a very lengthy calculation:
\begin{equation}
\label{eq: EthP0-Omega0EthK0}
\begin{split}
\tilde{\eth} p^\circ - \Omega^\circ \tilde{\eth} k^\circ =& \, \tilde{\eth}' \tilde{\eth}' \tilde{\eth}' h^{1 \circ}_{mm} -  \tilde{\eth}' \tilde{\eth} \, \tilde{\eth} \, \bar{h}^{1 \circ}_{mm} - 3 \Omega^\circ \tilde{\eth} \tilde{\thorn}' \bar{h}^{1 \circ}_{mm} + 6 \tau^\circ \tilde{\thorn}' \bar{h}^{1 \circ}_{mm} - 6 \tilde{\thorn}' \bar{h}^{2 \circ}_{nm}\\
& + 2 \tilde{\eth}' h^{1 \circ}_{nn} - 6\tilde{\thorn}' y^\circ \\
=& \, 4 \bar{S}^{3 \circ}_{nm} - 2 \tilde{\eth} \bar{S}^{3 \circ}_{mm} - 2 \tilde{\eth}' \tilde{\eth}' S^{3 \circ}_{lm} + 2 \tilde{\eth} \, \tilde{\eth}' \bar{S}^{3 \circ}_{lm} - 6\tilde{\thorn}' y^\circ.
\end{split}
\end{equation}
By lemma \ref{lem:S}, we know that $S^{3 \circ}_{nm}=S^{3 \circ}_{mm} = S^{3 \circ}_{lm} =0$, so $k^\circ = 0$ (by \eqref{eq:f}) and $p^\circ = 0$ (by \eqref{eq:d}) gives $\tilde{\thorn}' y^\circ = 0$ in view of \eqref{eq: EthP0-Omega0EthK0}. 

By assumption, all terms in $y^\circ$ in \eqref{eq:ycirc1} go to zero for $u \to -\infty$, except possibly for $-\tilde{\eth} \, d^\circ$. 
However, we are still free to adjust $d^\circ \to d^\circ +\delta^\circ$, with $\tilde{\thorn}' \delta^\circ = 0$. This adjustment will set $y^\circ = 0$ provided that $\delta^\circ$
satisfies 
\begin{equation}
\label{yeq}
y^\circ = \tilde{\eth} \, \delta^\circ = \left( \tilde{\eth} + \frac{\tau^\circ}{2 \rho^{\prime \circ}}  \tilde{\thorn}' \right) \delta^\circ , 
\end{equation}
where we used $\tilde{\thorn}' \delta^\circ = 0$ in the second step. The operator $\tilde{\eth} + \tfrac{\tau^\circ}{2 \rho^{\prime \circ}}  \tilde{\thorn}'$ is a spin-lowering operator on 
each cut $\mathbb{S}^2_u$ of constant $u$ of $\sI^+$, in retarded KN coordinates and the Kinnersley frame, see table \ref{tab:2} in \ref{sec:Held}. By well-known properties of these operators, , see e.g. \cite{Chandrasekhar:1984siy}, \eqref{yeq} hence has a solution on such a  cut $\mathbb{S}^2_u$---hence on all such cuts, since $y^\circ, \delta^\circ$ do not depend upon $u$---iff 
$y^\circ$ is $L^2$-orthogonal to the $s=-1,\ell=1$ spin-weighted spherical harmonics, see e.g. \cite{Chandrasekhar:1984siy}. 
To see that this is the case, we take an $L^2$ inner product of $y^\circ$ as in \eqref{eq:ycirc1} with such a harmonic
on a cut of constant $u$ and we let $u \to -\infty$. Then the contribution of all terms in $y^\circ$ as in \eqref{eq:ycirc1} other than $-\tilde{\eth} \, d^\circ$ goes to zero, by assumption.
Furthermore, since $\tilde{\thorn}' d^\circ \to 0$ as $u \to -\infty$, we may replace this term, asymptotically, by $-(\tilde{\eth} + \tfrac{\tau^\circ}{2 \rho^{\prime \circ}}  \tilde{\thorn}')d^\circ$, 
which is orthogonal to the $s=-1,\ell=1$ spin-weighted spherical harmonics. Thus, replacing $d^\circ \to d^\circ +\delta^\circ$, we can set $y^\circ = 0$.
This gives us precisely \eqref{eq:dcircdet} after using the formulas of lemma \ref{lem:11} to appropriately rewrite $y^\circ = 0$.

The remaining equations hold trivially: Eq. \eqref{eq:ecircdet} follows by applying $\tilde{\eth}$ to \eqref{eq:e} and then substituting  \eqref{eq:zm}. Eqs. \eqref{eq:zl}, \eqref{eq:zn} are mere repetitions of \eqref{eq:residcond2}.
\end{proof}

\begin{lemma}
\label{lem:11}
Integration of  \eqref{eq:hmbmb1} for $\Phi$ and subsequent improvement \eqref{Phireddef} yields 
\begin{equation}\label{eq:phiexp}
\Phi = \Phi^{0\circ} + \frac{\Phi^{1\circ}}{\rho} + \frac{\Phi^{2\circ}}{\rho^2} + \frac{\Phi^{3\circ}}{\rho^3} + O(\rho),
\end{equation}
where $\Phi^{i\circ} \circeq \GHPw{-4+i}{i}$ are GHP scalars annihilated by $\thorn$, given by
\begin{subequations}\label{eq:phivalues}
\begin{align}
\Phi^{0\circ} =& \ d^\circ, \\
\Phi^{1\circ} =& \ 2 \left(\tilde{\eth}'\tilde{\eth}' \zeta_l^\circ - 2 \bar{\tau}^\circ \zeta^\circ_{\bar{m}} \right) + \bar{h}^{1 \circ}_{mm},\\
\Phi^{2\circ} =& \ -2 \tilde{\eth}' \zeta_{\bar{m}}^\circ,\\
\label{eq:phivaluesd}
\Phi^{3\circ} =& \ e^\circ.
\end{align}
\end{subequations}
The GHP scalars $a^\circ \circeq \GHPw{-4}{-2}$, $b^\circ \circeq \GHPw{-1}{1}$, $c^\circ \circeq \GHPw{-3}{-3}$, defined by \eqref{eq:yzdef} and lemma \ref{lem:10}, are
\begin{equation}\label{eq:abc}
\begin{split}
a^\circ &= \ \half \bar{h}^{2 \circ}_{nm} - \half \tau^\circ \bar{h}^{1 \circ}_{mm} - \frac{1}{4} \Omega^\circ \tilde{\eth} \bar{h}^{1 \circ}_{mm},\\
b^\circ &= \ 0,\\
c^\circ &= \ \half \Re \left( h^{1 \circ}_{nn} + \tilde{\eth}' \tilde{\eth}' h^{1 \circ}_{mm} 
\right).
\end{split}
\end{equation}
\end{lemma}

\begin{proof}
Integration of  \eqref{eq:hmbmb1} using Held's coordinate free method (see \ref{sec:Held}) shows that the solution can be given as an asymptotic power series in terms of $\rho^j, j =-1,0,1,2,\dots$ and $\log \rho$ times coefficients annihilated by $\thorn$. 
It follows that integration of equation \eqref{eq:hmbmb1} for $\Phi$ yields a solution with the asymptotic behavior
\begin{equation}
\Phi = \Phi^{1\circ} \frac{1}{\rho} + \Phi^{0\circ} + \Phi^{{\rm log}\circ} \log(\rho) + o(\rho). 
\end{equation}
Since $\Phi^{0\circ} \circeq \GHPw{-4}{0}$ is a solution to the homogeneous version of the equation \eqref{eq:hmbmb1} for $\Phi$, we may set it to zero if we wish. Furthermore, we find
\begin{equation}
\Phi^{1\circ} = \bar{h}_{mm}^{1 \circ},
\quad
\phi_{\rm log}^\circ = -\frac{2}{3} \left(\bar h_{mm}^{2 \circ} - \Omega^\circ \bar{h}_{mm}^{1 \circ} \right). 
\end{equation}
Since by lemma \ref{lem:S} we may assume that $S_{mm}$ is of order $O(\rho^4)$, i.e., $S^{j \circ}_{mm} = 0$ for $j \le 3$, the $mm$ NP component of the linearized EE gives
\begin{equation}
\tilde{\thorn}' \left( h^{2 \circ}_{mm} + \Omega^\circ h^{1 \circ}_{mm} \right) = 0,
\end{equation}
see \ref{BondiIRG}.
Since we can assume that $h^{2 \circ}_{mm} + \Omega^\circ h^{1 \circ}_{mm} = 0$ as $u \to -\infty$ and since $\tilde{\thorn}' = \partial_u$ in retarded KN coordinates and the Kinnersley frame, we conclude $\Phi^{{\rm log}\circ}=0$.

The claimed formulas for the GHP scalars $a^\circ, b^\circ, c^\circ$ follow by a substitution of the 
asymptotic expansion of $h_{ab}$ and $\Phi$ from the formulas in lemma \ref{lem:10} and \eqref{eq:yzdef}.  
\end{proof}

We finally complete the construction by showing:

\begin{proposition}
\label{prop:13}
Under the assumptions of theorem \ref{thm:2} there exist $(e^\circ, f^\circ,  \zeta_l^\circ, \zeta^\circ_n, \zeta^\circ_m)$ with the 
properties required in lemma \ref{lem:9}.
\end{proposition}

\begin{proof}
Throughout this proof, we work in retarded KN coordinates $(u,\theta,\varphi_*)$ and the Kinnersley frame. 
The equations posed by lemma \ref{lem:9}
will be solved in the following sequence:

\begin{enumerate}
\item We show that \eqref{eq:f} has a suitable solution $f^\circ$. 
\item We view  \eqref{eq:zm} as the definition 
for $\zeta^\circ_m$ and also of $\zeta^\circ_{\bar m} = \bar \zeta^\circ_m$,
since $\zeta^a$ is real.
\item We view  \eqref{eq:e} as the definition of $e^\circ$.
\item We integrate  \eqref{eq:zl} from $u=-\infty$ to obtain the unique solution $\zeta^\circ_l$ decaying for $u\to -\infty$, noting that $\tilde{\thorn}' = \partial_u$. 
\item We view  \eqref{eq:zn} as a definition of $\zeta_n^\circ$. 
\item At this stage, the right side of 
 \eqref{eq:d} 
for $d^\circ$ is known, and we integrate \eqref{eq:d} from some freely chosen $u=u_0$ noting again that $\tilde{\thorn}' = \partial_u$. Having performed this integration, we perform the unique adjustment $d^\circ \to d^\circ +\delta^\circ$ for the $\delta^\circ$ satisfying $\tilde{\thorn}' \delta^\circ = 0$ whose construction is described in the proof of lemma \ref{lem:9}.
\end{enumerate}

The most non-trivial step is (i). Using the formulas for the Bondi news tensor in \ref{BondiIRG}, we can write \eqref{eq:f}
\begin{equation}\label{eq:f1}
\tilde{\eth}'\tilde{\eth}'\tilde{\eth}'\tilde{\eth}' f^\circ  - 3 \Psi^\circ \tilde{\thorn}' \bar f^\circ = 
-N^\circ,
\end{equation}
where $N^\circ = -\tilde{\thorn}' h^{1\circ}_{mm} \circeq \GHPw{0}{-4}$ is related to $N_{AB}$ by \eqref{NAB}. Consequently, the assumption about $N_{AB}$ 
in theorem \ref{thm:2} gives us that $N^\circ$ and any Held-type derivative thereof is $= O(|u|^{-3/2-\epsilon})$ when $u \to \pm \infty$, for some $\epsilon>0$.

It will be convenient to analyze \eqref{eq:f1}  in Fourier space, for which an $L^2$-framework is most suited here. We define the $L^2$-norm of GHP scalars
$X^\circ \circeq \GHPw{p}{q}$ to be 
\begin{equation}
\label{eq:L2}
\| X^\circ \|^2_{L^2} \equiv \int_{\mathbb{R} \times \mathbb{S}^2} |\Psi^\circ|^{\frac{p+q}{3}+1} \, |X^\circ|^2 \, \eta^\circ ,
\end{equation}
where $\eta^\circ_{abc} \circeq \GHPw{3}{3}$ is an integration element\footnote{In retarded KN coordinates and the Kinnersley 
frame, it would be given by $\eta^\circ = \sin \theta \, \intd u \wedge \intd \theta \wedge\intd \varphi_*$.} on $\sI^+ \cong \mathbb{R} \times \mathbb{S}^2$
defined by $l^d \epsilon_{dabc} = \eta^\circ_{abc}/\rho^2 + O(1/\rho)$, 
and $\Psi^\circ \circeq \GHPw{-3}{-3}$ is the GHP scalar defined in \ref{sec:Held}. It is included to make the definition GHP invariant.

Using the Fourier inversion theorem and the fact \cite{Breuer} that the spin-weighted spheroidal harmonics form a strongly complete set for $\omega \in \mathbb{R}$,  
it is possible to expand any $L^2$-GHP scalar ${}_sX^\circ$ of spin $s=(p-q)/2$ annihilated by $\thorn$ into modes as 
\begin{equation}
\label{eq:Fourier}
    {}_sX^\circ = \sum_{\ell,m}  \int_{-\infty}^\infty \intd \omega \, e^{-i\omega u} {}_sS_{\ell m}(\theta,\varphi_*;a\omega) {}_sX_{\ell m}^\circ(\omega) ,
\end{equation}
with square summable/integrable mode amplitudes ${}_sX_{\ell m}^\circ(\omega)$.
[${}_sS_{\ell m}(\theta,\varphi_*;a\omega)$ are the spin-weighted spheroidal harmonics with the 
harmonic $e^{im\varphi_*}$ included in their definition, see Sec. \ref{sec:modesoln}]. The summation sign is a shorthand for 
\begin{equation}
    \sum_{\ell,m} 
    =  \sum_{\ell = |s|}^\infty \sum_{m=-\ell}^\ell
\end{equation}
In the case of interest in this proof, $(1+|u|)N^\circ$ is smooth and 
$L^2$ (see \eqref{eq:L2}) together with all of its Held-type derivatives
$\tilde \thorn', \tilde{\eth}, \tilde{\eth'}$. Consequently, the Fourier modes of 
$N^\circ$ are such that $N_{\ell m}^\circ(\omega)$ or $\dd N_{\ell m}^\circ(\omega)/\dd \omega$
times any power of $|m|,|\omega|$, and ${}_{\pm 2} E_{\ell m}(a\omega)$ [the angular eigenvalue, see 
Sec. \ref{sec:modesoln}] are square summable/integrable.


We now consider the equivalent form of \eqref{eq:f1} in terms of the mode decomposition \eqref{eq:Fourier}. 
Using the angular TS identity \eqref{S:TS}, the symmetries \eqref{S:sym}, 
and the relationship between the operators ${}_s \mathcal{L}^\pm(a\omega)$ and $\tilde{\eth}, \tilde{\eth}'$ in \ref{sec:Held}, table \ref{tab:2}, Eq. \eqref{eq:f1} 
is found to be equivalent to
\begin{equation}
\label{eq:fN1}
    {}_2B_{\ell m}(a\omega) f^\circ_{\ell m}(\omega) +3iM\omega(-1)^m \, \overline{ f^\circ_{\ell -m}(-\omega)} = - \, 
    \overline {N_{\ell -m}^\circ(-\omega)} . 
\end{equation}
Here ${}_2B_{\ell -m}(-a\omega) = 
{}_2B_{\ell m}(a\omega)$ is the angular TS constant for spin $s=2$ which is real an non-negative for real $\omega$ by results quoted in Sec. \ref{sec:modesoln}. 
Eq. \eqref{eq:fN1} implies that
\begin{equation}
\label{eq:fN2}
\begin{split}
-{}_2C_{\ell m}(\omega) f^\circ_{\ell m}(\omega) \equiv& \ -\bigg[ {}_2B_{\ell m}(a\omega)^2+9M^2\omega^2 \bigg]  f^\circ_{\ell m}(\omega) \\
=& \  {}_2B_{\ell m}(a\omega) \overline {N_{\ell -m}^\circ(-\omega)}
-3iM\omega (-1)^m N^\circ_{\ell m}(\omega).
\end{split}
\end{equation}
The ``radial TS constant'' ${}_2C_{\ell m}(\omega)$
is non-negative for real $\omega$, and may vanish at most for $\omega=0$. However, by results quoted in Sec. \ref{sec:modesoln},
we have ${}_2B_{\ell m}(a\omega) = (\ell-1)\ell(\ell+1)(\ell+2)/4 + O(a\omega)$ uniformly in $m,\ell$,
showing that ${}_2C_{\ell m}(\omega)$ can vanish at most for $\ell \in \{0,1\}$. These values are excluded because $N^\circ$ has spin $s=2$. 
It follows that ${}_2C_{\ell m}(\omega) \ge c$ for some constant $c>0$ independent of $\omega,m,\ell$, and we can safely divide by ${}_2C_{\ell m}(\omega)$ in \eqref{eq:fN2}. 
The mode amplitudes $f^\circ_{\ell m}(\omega)$ obtained in this way from 
$N^\circ_{\ell m}(\omega)$ define a GHP scalar $f^\circ \circeq \GHPw{4}{0}$ via \eqref{eq:Fourier}, which is a weak
 $L^2$-solution of \eqref{eq:f1}.

We finally investigate the regularity and decay properties of $f^\circ$. Using lemma \ref{lem:u}, it follows that 
$f_{\ell m}^\circ(\omega)$ or $\dd f_{\ell m}^\circ(\omega)/\dd \omega$
times any power of $|m|,|\omega|$, and ${}_{\pm 2} E_{\ell m}(a\omega)$, are square summable/integrable.
Now, $-i\dd/\dd \omega, -i\omega, im$, and ${}_{\pm 2} E_{\ell m}(a\omega)$ correspond to multiplication by $u$,
$\mathcal{L}_t, \mathcal{L}_\phi$ and the application of the angular TS operator, respectively (which is an elliptic operator), see 
Sec. \ref{sec:modesoln}. Consequently, it follows that $(1+|u|)f^\circ$ is smooth and 
$L^2$ [see \eqref{eq:L2}] together with all of its Held-type derivatives
$\tilde \thorn', \tilde{\eth}, \tilde{\eth'}$. We have thus completed step (i).

Eqs. \eqref{eq:zm}, \eqref{eq:zn}, \eqref{eq:e} in steps (ii), (iii), (v) are mere definitions, and
give $e^\circ, \zeta_m^\circ, \zeta_n^\circ$.  

In step (iv), we must solve an equation of the form $\tilde{\thorn}' \zeta^\circ_l = X^\circ$ 
such that that $(1+|u|)X^\circ$ is smooth and 
$L^2$ together with all of its Held-type derivatives
$\tilde \thorn', \tilde{\eth}, \tilde{\eth'}$. Because $\tilde \thorn' =\partial_u$, we 
may integrate this equation from $u=-\infty$ and obtain a solution $\zeta^\circ_l$
which is smooth and going to zero as $u \to -\infty$ with all of its Held-type derivatives, though it may asymptote to a constant 
for $u \to +\infty$.

In step (vi)  we must solve an equation of the form $\tilde{\thorn}' d^\circ = X^\circ$ 
such that that $X^\circ$ is smooth and goes to zero as $u \to -\infty$ with all of its Held-type derivatives. Contrary to 
step (iv) his equation may therefore not 
necessarily be integrated from $u=-\infty$, but we may still integrate from some $u_0$. 
In this way we obtain a solution that may be growing for both $u \to \pm \infty$ but such that 
$\tilde{\thorn}' d^\circ$ goes to zero as $u \to -\infty$ with all of its Held-type derivatives.
\end{proof}

\section{Off-shell TS identities and AAB identity
}
\label{sec:TScov}

It is well-known that the TS identities can be written in GHP covariant form, see e.g., \cite{Fayos,Price}. For easier reference  
we give these forms here, keeping also those terms that are identically zero for solutions to the homogeneous EE. Such ``off-shell''
identities have not been given in the literature to our knowledge and are needed in the main text.
The TS identities in covariant form can be seen as consequences of operator identities involving combinations of  
$\mathcal{S}$, $\mathcal{T}$, see \ref{app:ST}, and their adjoint, overbared, and primed versions. 
One finds the following relations, where 
here and in the following $\lie{\xi}$ means the GHP covariant Lie-derivative \cite{Edgar}.

\begin{proposition}
\label{Prop:TS}
In any non-accelerating type D spacetime the following operator identities hold:
\begin{equation}
\label{TSfirstform}
\begin{split}
\mathcal{T} \mathcal{S}^\dagger \eta =\ & 0, \\
\mathcal{T} \overline{\mathcal{S}^\dagger \eta} =\ & - \half \thorn^4 \bar{\eta}, \\
\mathcal{T}' \mathcal{S}^\dagger \eta =\ & - \frac{1}{8} (\mathcal{O}^\dagger + 8 B^{\prime} \cdot \Theta + 16 \Psi_2) \mathcal{O}^\dagger \eta + \frac{3}{2} \Psi_2^{\frac{4}{3}} \lie{\xi} \eta\\
\mathcal{T}' \overline{\mathcal{S}^\dagger \eta} =\ & - \half \eth^{\prime 4} \bar{\eta}, \\
\end{split}
\end{equation}
for any $\eta \circeq \GHPw{-4}{0}$, where $\xi_a = \Psi_2^{-\frac{1}{3}} \left( B'_a - B_a \right)$ is the Killing vector field that exists in all such spacetimes ($\equiv M^{-1/3} t^a$ in Kerr). 
Under the same hypothesis, we also have:
\begin{equation}
\begin{split}
\S \T^\dagger \eta = \ & 0 \\ 
\S \overline{\T^\dagger \eta}=\ & -\half
\left[\thorn^2 - 4 (\rho + \bar{\rho}) \thorn + 12 \rho \bar{\rho} \right] \thorn^2 \bar{\eta}\\
\S' \T^\dagger \eta 
=\ & - \frac{1}{8} \mathcal{O}' (\mathcal{O}' + 8 B \cdot \Theta + 16 \Psi_2) \eta + \frac{3}{2} \Psi_2^{\frac{4}{3}} \lie{\xi} \eta \\
\S' \overline{\T^\dagger \eta}=& - \half \left[ \eth^{\prime 2} - 4 (\bar{\tau} + \tau') \eth' + 12 \bar{\tau} \tau' \right] \; \eth^{\prime 2} \bar{\eta}
\end{split}
\end{equation}
\end{proposition}

We also have the GHP primed- and overbared versions of these equations since the GHP formalism is symmetric under these operations in the case of type D spacetimes. Using the GHP-identity \cite{Price}
\begin{equation}
\thorn^4 \Psi_2^{-\frac{4}{3}} \eth^{\prime 4} =  \eth^{\prime 4} \Psi_2^{-\frac{4}{3}} \thorn^4,
\end{equation}
the consistency of the relations in proposition \ref{Prop:TS} implies:
\begin{proposition}
\label{prop:3}
For any $\eta \circeq \GHPw{-4}{0}$
we have
\begin{equation}
\begin{split}
\thorn^4 \Psi_2^{-\frac{4}{3}} \mathcal{T}' \Re(\S^\dagger \eta)  =& \eth^{\prime 4} \Psi_2^{-\frac{4}{3}} \mathcal{T} \Re(\S^\dagger \eta) - 3\lie{\xi} \bar{\mathcal{T}}
\Re(\S^\dagger \eta) \\
&- \frac{1}{16} \thorn^4 \Psi_2^{-\frac{4}{3}}(\O^\dagger + 8 B' \cdot \Theta + 16 \Psi_2)\O^\dagger \eta.
\end{split}
\end{equation}
\end{proposition}
Again, also  the primed and overbared analogues of this statement hold true. For a metric $h_{ab} = \Re(\S^\dagger \eta)_{ab}$ of reconstructed form such that $\O^\dagger \eta = 0$
holds (implying the EE $(\E h)_{ab} = 0$), the identity of proposition \ref{prop:3} may be written in the more familiar form \cite{Fayos,Price}
\begin{equation}
\label{TSonshell1}
\thorn^4 \Psi_2^{-\frac{4}{3}} \psi_4  = \eth^{\prime 4} \Psi_2^{-\frac{4}{3}} \psi_0 - 3\lie{\xi} \bar \psi_0.
\end{equation} 

Given that, up to global assumptions and up to terms not contributing to $\psi_0, \psi_4$, any perturbation $h_{ab}$  satisfying the linearized EE $h_{ab}$ can be written in reconstructed from $h_{ab} = \Re(\S^\dagger \eta)_{ab}$ such that $\O^\dagger \eta = 0$ by corollary \ref{GHZcor}, it is plausible that an ``off-shell'' version of
the ``on-shell'' TS identity \eqref{TSonshell} ought to exist that is valid for any $h_{ab}$ (not necessarily satisfying the EE). Such an identity would be expected to 
have additional terms explicitly containing $(\E h)_{ab}$. 
In fact, the existence of such an identity follows directly from a general principle expressed by theorem \ref{lem:null}, see proposition \ref{prop:3'}. 

Actually, that proposition is also a consequence of an even more general identity, due to \cite{Aksteiner:2016pjt}.
Using the operators $\S,\T$, 
their identity may be stated as:

\begin{theorem} (AAB decomposition \cite{Aksteiner:2016pjt}) For any symmetric tensor $h_{ab}$, 
\begin{equation}
\label{AAB}
(\lie{\xi} h)_{ab} = \frac{1}{3} \Re\left[ \S^\dagger_{ab}\left( \Psi_2^{-\frac{4}{3}} \T'h \right) - \S^{\prime \dagger}_{ab} \left(\Psi_2^{-\frac{4}{3}} \T h\right) \right] + \nabla_{(a} (\mathcal{A} h)_{b)} + 
(\mathcal{M} \E h)_{ab},
\end{equation}
where $\mathcal{M}, \A$ are partial differential operators
built from the GHP directional derivative operators and rational functions of the background GHP scalars $\tau, \rho, \tau', \rho', \Psi_2$ and the complex conjugates.
\end{theorem}

\medskip
\noindent
{\bf Remark 6.} The AAB decomposition has certain analogies to the decomposition in the spin-1 case claimed in theorem \ref{thm:MaxwellHertz} \cite{Dolan,Wardell2}. Indeed, let us consider two potentials $\Phi_0 \circeq \GHPw{4}{0}$ and $\Phi_4 \circeq \GHPw{-4}{0}$ that are anti-time-derivatives of the Weyl scalars $\psi_0$ and $\psi_4$ [compare \eqref{eq:Hertzspin-1}],
\begin{equation}
\label{eq:Hertzspin-2}
{\mathcal L}_\xi \Phi_0 = \psi_0, \quad 
{\mathcal L}_\xi \Phi_4 = -\psi_4, 
\end{equation}
and let us further consider an anti-time-derivative, $S_{ab}$, 
of the stress tensor:
\begin{equation}
\label{}
{\mathcal L}_\xi S_{ab} = T_{ab}. 
\end{equation}
Then we set [compare \eqref{Hab}]
\begin{equation}
\label{Habcd}
H_{abcd} := -\Phi_0 Z_{abcd}^{\prime} - \Phi_4 Z^{}_{abcd}, 
\end{equation}
and [compare \eqref{eq:Aansatz}]
\begin{equation}
    h_{ab} := 
    \frac{1}{3} \nabla^c \Re\left( \zeta^4 \nabla^d H_{d(ab)c}  \right)
    +(\mathcal{M} S)_{ab},
\end{equation}
where $\zeta = \Psi_2^{-\frac{1}{3}}$. Then \eqref{AAB} implies that this $h_{ab}$ satisfies the linearized EE $\E h_{ab} = T_{ab}$ up to a term annihilated by $\mathcal{L}_\xi$, analogously to theorem \ref{thm:MaxwellHertz}. One may further note the structural similarities between this formula for $h_{ab}$ and \eqref{eq:Aansatz} in the spin-1 case, but also the following crucial differences: 1) the terms in \eqref{Habcd} have a relative sign difference compared to \eqref{Hab}, 2) the powers of $\zeta$ are not analogous. 

\medskip
\noindent
{\bf Remark 7.}
The authors of \cite{Aksteiner:2016pjt} used a spinor formalism to express \eqref{AAB} in a somewhat different looking form. The operators $\mathcal{M}, \mathcal{A}$ may be extracted from their spinor formulas \cite[Eqs. (56), (57a,b)]{Aksteiner:2016pjt} with some labor\footnote{
Since we are taking the real part in \eqref{AAB}, our operator $\mathcal{M}$ is related to the real part of the operator $\mathcal{N}$ that appears in \cite[Eq. (57b)]{Aksteiner:2016pjt}.}. We now give a different proof showing clearly that the existence of 
the identity \eqref{AAB} can be traced to the TS identities. 
Our proof is more conceptual than \cite{Aksteiner:2016pjt}, but does not, by itself, yield the explicit forms of $\mathcal{M}, \A$ as provided by \cite{Aksteiner:2016pjt}.

\begin{proof}  
First we argue that \eqref{AAB} holds if $(\E h)_{ab}=0$: Using the 
identities $\T\lie{X} = \T'\lie{X} = 0$,  \eqref{TSonshell}, and \eqref{TSfirstform}, we see that the difference between the left and right side of \eqref{AAB} 
is annihilated by both $\T, \T'$. By Wald's theorem \cite{Waldthm} (see e.g., \cite{StewardWalker} for a proof in GHP language), this difference hence has to be the sum of a pure gauge perturbation-and a zero mode perturbation, see \ref{app:D}. 

The argument by \cite{Waldthm} shows that 
the pure gauge transformation is constructed out of $h_{ab}$ and its derivatives, and thus of the form $\nabla_{(a} ({\mathcal A}h)_{b)}$ for some partial differential operator $\A$ as in \eqref{AAB}. Furthermore, a zero mode perturbation actually cannot be present in the difference: Interpolating $h_{ab}$ to zero on the causal complement of 
an open neighborhood of some point, $x$ while keeping $h_{ab}$ the same---as is always possible by the gluing theorems by 
\cite{CorvinoSchoen}---on a yet smaller neighborhood of $x$, 
the difference is of compact support on this causal complement, and the same near $x$. But a zero mode cannot vanish on the causal complement of 
any compact set in any gauge. Hence, the difference does not contain a zero mode. 

Thus, \eqref{AAB} without the term involving $\mathcal{M}$ holds whenever $(\E h)_{ab}=0$. Let $(\mathcal{D}h)_{ab}$ be 
difference of the left side of \eqref{AAB}, minus the terms on the right side not involving $\mathcal{M}$. We need to show that is an operator from which we can factor $\E$. 
This problem is solved by applying theorem \ref{lem:null} to all possible NP components of the operator $\mathcal{D}$. 
\end{proof}

We now apply $\T$ to both sides of \eqref{AAB}, and we use that (i) $\T$ commutes with the GHP Lie-derivative $\lie{\xi}$, (ii) $\O^{\prime \dagger} \Psi_2^{-\frac{3}{4}} = 
\Psi_2^{-\frac{3}{4}} \O$, (iii) Eqs. \eqref{TSfirstform} and their GHP-primed versions, and (iv)
the SEOT identity $\S\E=\O\T$. Then we obtain:
\begin{proposition}
\label{prop:3'} (Off-shell TS identity)
We have
\begin{equation}
\label{TSonshell}
\thorn^4 \Psi_2^{-\frac{4}{3}} \T'  = \eth^{\prime 4} \Psi_2^{-\frac{4}{3}} \T - 3\lie{\xi} \bar\T + \bar{\mathcal{Q}} \E ,
\end{equation}
where\footnote{One may show that $\mathcal{Q}$ is proportional $\T \mathcal{N}$, where $\mathcal N$ is the operator defined in Eq. (57b) by \cite{Aksteiner:2016pjt}. We thank Barry Wardell for pointing this fact out to us.}
\begin{equation}
\mathcal{Q} = 3\left[ 2\T\mathcal{M} + \frac{1}{4} \Psi_2^{-\frac{4}{3}} \left(\O + 8B \cdot (\Theta + 4B') + 16 \Psi_2 \right) \S
\right].
\end{equation}
Here $\mathcal{M}$ is the operator appearing in \eqref{AAB}.
\end{proposition}

\section{A ``Nullstellensatz'' for the linearized EE}
\label{sec:NSsatz}

The following theorem 
is based on a combination of a general Nullstellensatz for systems of PDE \cite{Shankar} and specific structural properties of the linearized Einstein operator, $\E$.

\begin{theorem}
\label{lem:null}
(Nullstellensatz for the linearized EE)
Let $\mathcal{D}$ be a partial differential operator on Kerr, from symmetric covariant tensors to some GHP bundle $\sL_{p,q}$, built from rational functions of the GHP scalars $\rho, \tau, \rho', \tau', \Psi_2$, 
the GHP operators $\thorn, \thorn', \eth,$ and the NP tetrad legs $l^a, n^a, m^a$, and their complex conjugates.
Suppose $\mathcal{D}$ has the property that $\mathcal{D}h = 0$ whenever a smooth $h_{ab}$ satisfies 
$(\E h)_{ab} = 0$. Then there is a partial differential operator $\mathcal{M}$, with the same general properties as $\mathcal{D}$, such that 
$\mathcal{D} = \mathcal{M}\E$.
\end{theorem}
\begin{proof}
Let $x_0 \in \sM$ be an arbitrary but fixed point and consider a high frequency solution $h_{ab}(\nu, x) = A_{ab}(x) e^{i\nu S} + O(\nu^{-1})$ in traceless Lorenz gauge
to the linearized EE near $x_0$ with momentum $p_a=\nabla_a S(x_0)$ at $x_0$, as constructed e.g., in \cite{Green:2015kur}. 
Taking the frequency $\nu \to \infty$ in $\mathcal{D}h(\nu)=0$, it follows that 
$\sigma_{\mathcal D}(x_0, p)^{ab} A_{ab} = 0$, where $A_{ab}$ is any symmetric covariant rank two tensor in the tangent space at $x_0$ such that
$A_{ab}p^b=0=g^{ab}A_{ab}$. (Here, $\sigma$ is the principal symbol \cite{hormander1} of a partial differential operator obtained by replacing $\partial_a \to -ip_a$ in the 
highest derivative terms.) 

Now we introduce an orthonormal frame $e_\mu^a, \mu=0,1,2,3$ in $T_{x_0} \sM$ such that $\frac{1}{\sqrt{2}}(e^a_0 \pm e^a_1)$, 
$\frac{1}{\sqrt{2}}(e^a_2 \pm i e^a_3)$ give $l^a, n^a, m^a, \bar m^a$ at $x_0$ in this ordering. We use this tetrad 
to identify tensors  over this tangent space with tensors in Minkowski spacetime $(\mathbb{R}^4, \eta_{\mu\nu})$, 
as e.g., in $p^a = p^\mu e_\mu^a$, $A_{ab} =  A_{\mu\nu} e^\mu_a e^\nu_b$ etc., and we use it as a trivialization of 
$\sL_{p,q}$ at $x_0$. 
Then under this identification, $\mathcal{D}_0 := \sigma_{\mathcal D}(x_0, \partial_\mu), \mathcal{E}_0 := \sigma_{\mathcal E}(x_0, \partial_\mu)$ define 
matrix partial differential operators with constant complex coefficients on Minkowski spacetime. $\E_0$ is, of course, 
the linearized Einstein operator in Minkowski spacetime. By linearity, any smooth tensor of the form 
\begin{equation}
\label{Fourierrep}
h_{\mu\nu}(x) = \int \dd^3 \textbf{p} \, \left[ A_{\mu\nu}^+(\textbf{p}) e^{+i|\textbf{p}|x^0+ i\textbf{px}} +
A_{\mu\nu}^-(\textbf{p}) e^{-i|\textbf{p}|x^0+ i\textbf{px}}
\right]
\end{equation}
 such that $A_{\mu\nu}^\pm(\textbf{p})$ is a Schwarz class function of $\textbf{p}$ with (here $p^\mu = (|\textbf{p}|,\textbf{p}), x^\mu = (x^0,\textbf{x})$)
$A_{\mu\nu}^\pm(\textbf{p})p^\nu = A_{\mu\nu}^\pm(\textbf{p})\eta^{\mu\nu} = 0$, is a solution in traceless Lorenz gauge, 
to the linearized EE in Minkowski spacetime, 
$\E_0 h=0$, such that also $\mathcal{D}_0 h = 0$. As is well-known and described, e.g., in Sec. \ref{sec:RedMax}, any smooth solution to the EE
with spatially compact support\footnote{I.e., the restriction to any Cauchy surface has compact support.} 
can be put into traceless Lorenz gauge by a gauge transformation of spatially compact support. Furthermore, by an application of the ``Ehrenpreis–Palamodov Fundamental Principle'' (see e.g., \cite[sec. 7.7]{hormander3})
any smooth solution to the EE with spatially compact support in Minkowski spacetime in traceless Lorenz gauge obviously can be put into the form \eqref{Fourierrep}
with an $A_{\mu\nu}(\textbf{p})$ having the indicated properties. Since $\mathcal{D}_0 \mathcal{L}_X \eta = 0$ (which follows 
by testing\footnote{This equation holds due to the assumption of the lemma, 
because $\mathcal{E} \mathcal{L}_X g = 0$, which is a general structural property of the linearized EE
on a background satisfying $R_{ab}=0$.} $\mathcal{D} \mathcal{L}_X g = 0$ with a high frequency gauge transformation $X^a$ near $x_0$), it follows that 
$\mathcal{D}_0 h = 0$ for every smooth solution $\E_0 h = 0$ to the EE with spatially compact support in Minkowski spacetime. 

By the gluing theorems in \cite{CorvinoSchoen}, given an open set $\mathscr{U}$ of any Cauchy surface of 
Minkowski spacetime and smooth Cauchy data for the linearized EE in this set (necessarily satisfying the linearized constraint equations), we can find 
an extension of these Cauchy data to the entire Cauchy surface having compact support in a somewhat larger set (necessarily satisfying the linearized constraint equations). 
It follows easily from this fact and the usual propagation properties of the linearized EE that $\mathcal{D}_0 h = 0$ for every smooth solution $\E_0 h = 0$, not necessarily of spatially compact support. 

At this stage, we can apply a well known Nullstellensatz for systems of PDE with constant coefficients, 
see ``Oberst's theorem'' for $C^\infty(\mathbb{R}^4)$ (see e.g., \cite{Shankar}), implying that there is a matrix partial differential operator $\mathcal{N}_0$ with constant coefficients
such that $\mathcal{D}_0 = \mathcal{N}_0 \mathcal{E}_0$. 
Actually, we have some flexibility how with think about the coefficients of $\E_0, \mathcal{D}_0, \mathcal{N}_0$. 

The best result is obtained if we think about the coefficients as elements of a field of fractions\footnote{We thank Anna-Laura Sattelberger for suggesting the use of other coefficient fields in this argument.}, $\mathbb{K}$, constructed as follows. First, we consider the ring $\mathbb{C}[X]$, where the indeterminates $X$ stand for $\tau(x_0), \rho(x_0), \Psi_2(x_0)$, as well as their GHP primes and conjugates, see \ref{sec:GHPformulas}. However, rather than considering them as GHP weighted scalars, we consider them as formal indeterminates at this stage. Next, we consider the integral domain $\mathbb{I}$ obtained by factoring $\mathbb{C}[X]$ by the relations \eqref{eq:GHPrelations} (at $x_0$), which generate a prime ideal of $\mathbb{C}[X]$. Finally, we let 
$\mathbb{K}$ be the quotient field of $\mathbb{I}$. Since 
Oberst's theorem applies to all fields of characteristic zero, 
we can think of $\E_0, \mathcal{D}_0, \mathcal{N}_0 \in \mathbb{K}[\partial_0, \partial_1, \partial_2, \partial_3]$ as having coefficients in $\mathbb{K}$, i.e. rational functions of GHP background scalars modulo the relations \eqref{eq:GHPrelations}.

On $\mathbb{K}[\partial_0, \partial_1, \partial_2, \partial_3]$, there is an action of the multiplicative group $\mathbb{C}^\times$ of non-zero complex numbers. It acts as 
\eqref{GHPtrafo} on the GHP background scalars at $x_0$ (hence thereby on $\mathbb{K}$), and as $(\partial_0+\partial_1) \to \lambda \bar{\lambda} (\partial_0+\partial_1)$, $(\partial_0-\partial_1) \to \lambda^{-1} \bar{\lambda}^{-1} (\partial_0-\partial_1)$, and $(\partial_2+i\partial_3) \to \lambda \bar{\lambda}^{-1} (\partial_2+i\partial_3)$. Let us denote this 
group action by $\alpha_\lambda$, and consider 
\begin{equation}
\label{K0defn}
\mathcal{K}_0:=\left(
\lambda \frac{\partial}{\partial \lambda} +
\bar \lambda \frac{\partial}{\partial \bar \lambda}-q-p
\right) \alpha_\lambda \mathcal{N}_0 \bigg|_{\lambda = 0}.
\end{equation}
$\mathcal{K}_0$ measures the failure of $\mathcal{N}_0$ to be locally GHP covariant at $x_0$. Since $\E_0, \mathcal{D}_0$ have the property of being locally GHP invariant, we get 
$0 = \mathcal{K}_0 \mathcal{E}_0$ by applying $\alpha_\lambda$ to the relation $\mathcal{D}_0 = \mathcal{N}_0 \mathcal{E}_0$
and applying the derivative operator in \eqref{K0defn}. Now we take the adjoint of this operator relation, giving 
\begin{equation}
    0 =  \mathcal{E}_0 \mathcal{K}_0^\dagger.
\end{equation}
This implies that $\mathcal{K}_{0\mu\nu}^\dagger$ must be of the 
form $\mathcal{K}_{0\mu\nu}^\dagger = \partial_{(\mu} \mathcal{A}_{0\nu)}^\dagger$, i.e. it has the form of a gauge-transformation, for some $\mathcal{A}_{0\mu}^\dagger \in \mathbb{K}[\partial_0, \partial_1, \partial_2, \partial_3]$. We now redefine 
\begin{equation}
 \mathcal{N}_0^{\mu\nu} \to   \mathcal{N}_0^{\mu\nu} -  \mathcal{A}^{(\mu}_0 \partial^{\nu)},   
\end{equation}
and then the redefined operator still satisfies $\mathcal{D}_0 = \mathcal{N}_0 \mathcal{E}_0$ (since the Einstein operator necessarily yields a divergence-free rank-2 tensor), and furthermore, we now have $\mathcal{K}_0=0$ in \eqref{K0defn}. 
Thus, the new $\mathcal{N}_0$ is GHP covariant at $x_0$. 

In the next step, we promote $\mathcal{N}_0$ to an operator 
on all of $\sM$ by going through the following steps. 
1) We pass from the basis $\partial_0, \partial_1, \partial_2, \partial_3$ to $\frac{1}{\sqrt{2}}(\partial_0 \pm \partial_1), 
\frac{1}{\sqrt{2}}(\partial_2 \pm i\partial_3)$. 
2) We replace these operators, in this order, by $\thorn, \thorn', \eth, \eth'$. The GHP directional derivatives obviously do not commute, so we pick any ordering that we please. 3) We replace the formal indeterminates in the ring $\mathbb{K}$ by the rational functions of the GHP background scalars that they correspond to. Then we get a new operator for each $x_0$, and since $x_0$ was arbitrary, this defines a partial differential operator $\mathcal{N}$ on $\sM$ with the general properties stated in the theorem. 

At this stage, because the GHP derivatives $\thorn, \thorn', \eth, \eth'$ turn rational functions of $\rho, \rho', \tau, \tau', \Psi_2$ and their complex conjugates into such functions
by table \ref{tab:4}, it follows that $\mathcal{D} - \mathcal{N}\E$ is a new 
partial differential operator on $\sM$ which is satisfying the same assumptions of the theorem as $\mathcal{D}$, but is of lower differential order. By lowering the order in this way a 
finite number of times, we reach an operator $\mathcal{D}$ of zeroth order, for which the theorem trivially holds. 
\end{proof}


\section*{References} 
\begin{thebibliography}{99}

\bibitem{Abbott:1981ff}
L.~F.~Abbott and S.~Deser,
``Stability of Gravity with a Cosmological Constant,''
Nucl. Phys. B \textbf{195}, 76-96 (1982),
doi:10.1016/0550-3213(82)90049-9

\bibitem{Aksteiner:2014zyp}
S.~Aksteiner,
\emph{Geometry and analysis on black hole spacetimes}, Ph.D. thesis (2014),
doi:10.15488/8214

\bibitem{Aksteiner:2016pjt}
S.~Aksteiner, L.~Andersson and T.~B\"ackdahl,
``New identities for linearized gravity on the Kerr spacetime,''
Phys. Rev. D \textbf{99}, 044043  (2019),
doi:10.1103/PhysRevD.99.044043

\bibitem{Andersson:2021eqc}
L.~Andersson, T.~B\"ackdahl, P.~Blue and S.~Ma,
``Nonlinear Radiation Gauge for Near Kerr Spacetimes,''
Commun. Math. Phys. \textbf{396}, 45-90 (2022),
doi:10.1007/s00220-022-04461-3

\bibitem{Andersson:2019dwi}
L.~Andersson, T.~B\"ackdahl, P.~Blue and S.~Ma,
``Stability for linearized gravity on the Kerr spacetime,''
[arXiv:1903.03859 [math.AP]].

\bibitem{Andersson}
L.~Andersson, S. Ma, C. Paganini, and B.~F. Whiting, ``Mode stability on the real axis,'' 
J. Math. Phys. {\bf 58}, 072501 (2017)

\bibitem{Angelopoulos}
Y. Angelopoulos, S. Aretakis and D. Gajic, ``Late-time tails and mode coupling of linear waves on Kerr spacetimes,'' Adv. Mathematics {\bf 417}, 
108939 (2023)

\bibitem{Araneda}
B. Araneda, ``Symmetry operators and decoupled equations for linear fields on black hole spacetimes,'' Class. Quant. Grav. {\bf 34}, 035002 (2016),
doi: 10.1088/1361-6382/aa51ff


\bibitem{Barack}
L. Barack, ``Lectures on black-hole perturbation theory,''\\
https://www.ctc.cam.ac.uk/activities/rise/slides/barack.pdf

\bibitem{Barack:2017oir}
L.~Barack and P.~Giudice,
``Time-domain metric reconstruction for self-force applications,''
Phys. Rev. D \textbf{95}, 104033 (2017),
doi:10.1103/PhysRevD.95.104033

\bibitem{Benomio:2022tfe}
G.~Benomio,
``A new gauge for gravitational perturbations of Kerr spacetimes I: The linearised theory,''
[arXiv:2211.00602 [gr-qc]].

\bibitem{Bourg:2024vre}
P.~Bourg, B.~Leather, M.~Casals, A.~Pound and B.~Wardell,
``Implementation of a GHZ-Teukolsky puncture scheme for gravitational self-force calculations,''
[arXiv:2403.12634 [gr-qc]].

\bibitem{Edgar}
E. S. Brian and G. Ludwig. ``Integration in the GHP formalism IV: A new Lie derivative operator leading to an efficient treatment of Killing vectors.'' Gen. Rel. Grav. {\bf 32},  637-671, (2000)

\bibitem{Campanelli}
M. Campanelli and C. O. Lousto, ``Second order gauge invariant gravitational perturbations of a Kerr black hole.'' Phys. Rev. D 59, 124022 (1999)

\bibitem{Casals:2004zq}
M.~Casals and A.~C.~Ottewill,
``High frequency asymptotics for the spin-weighted spheroidal equation,''
Phys. Rev. D \textbf{71}, 064025  (2005),
doi:10.1103/PhysRevD.71.064025


\bibitem{Casals:2024ynr}
M.~Casals, S.~Hollands, A.~Pound and V.~Toomani,
``Spin-2 Green's Functions on Kerr in Radiation Gauge,''
[arXiv:2402.15468 [gr-qc]].

\bibitem{Casals:2018cgx}
M.~Casals, A.~C.~Ottewill and N.~Warburton,
``High-order asymptotics for the Spin-Weighted Spheroidal Equation at large real frequency,''
Proc. Roy. Soc. Lond. A \textbf{475}, 20180701 (2019),
doi:10.1098/rspa.2018.0701

\bibitem{Casals:2020fsb}
M.~Casals and R.~T.~da Costa,
``The Teukolsky\textendash{}Starobinsky constants: facts and fictions,''
Class. Quant. Grav. \textbf{38}, 165016  (2020),
doi:10.1088/1361-6382/ac11a8

\bibitem{Chandrasekhar:1984siy}
S.~Chandrasekhar,
``The Mathematical Theory of Black Holes,''
Fundam. Theor. Phys. \textbf{9},  5-26 (1984),
doi:10.1007/978-94-009-6469-3\_2


\bibitem{Cheung:2022rbm}
M.~H.~Y.~Cheung, V.~Baibhav, E.~Berti, V.~Cardoso, G.~Carullo, R.~Cotesta, W.~Del Pozzo, F.~Duque, T.~Helfer and E.~Shukla, \textit{et al.}
``Nonlinear Effects in Black Hole Ringdown,''
Phys. Rev. Lett. \textbf{130}, 8 (2023),
doi:10.1103/PhysRevLett.130.081401

\bibitem{Chrzanowski:1975wv}
P.~L.~Chrzanowski,
``Vector Potential and Metric Perturbations of a Rotating Black Hole,''
Phys. Rev. D \textbf{11}, 2042-2062 (2020), 
doi:10.1103/PhysRevD.11.2042

\bibitem{CorvinoSchoen}
J. Corvino and R. M. Schoen, ``On the asymptotics for the vacuum Einstein constraint equations.'' J. Diff. Geom. 73, no. 2, 185-217 (2006)

\bibitem{Dafermos:2021cbw}
M.~Dafermos, G.~Holzegel, I.~Rodnianski and M.~Taylor,
``The non-linear stability of the Schwarzschild family of black holes,''
[arXiv:2104.08222 [gr-qc]].

\bibitem{DHR}
M. Dafermos, I. Rodnianski, and Y. Shlapentokh-Rothman, 
``Decay for solutions of the wave equation on Kerr exterior spacetimes III: The full subextremal case $|a|<M$,'' 
Annals of Mathematics {\bf 183}, 787-913 (2016),
doi: 10.4007/annals.2016.183.3.2


\bibitem{Dolan}
S.~R. Dolan, C. Kavanagh, and B.  Wardell, ``Gravitational perturbations of rotating black holes in Lorenz gauge,'' Phys. Rev. Lett. {\bf 128}, 151101 (2022), doi: 10.1103/PhysRevLett.128.151101

\bibitem{Dolan2}
S.~R. Dolan, L. Durkan, C. Kavanagh, and B. Wardell, ``Metric perturbations of Kerr spacetime in Lorenz gauge: Circular equatorial orbits,'' [arXiv:2306.16459 [gr-qc]] 

\bibitem{Ehlers}
J. Ehlers, ``The geometry of the (modified) GHP formalism,''
Comm. Math. Phys. {\bf 37}, 327-329 (1974)

\bibitem{Fayos}
F. Fayos, J. J. Ferrando, and X. Jaen, ``Electromagnetic and gravitational perturbation of type D space‐times'' J. Math. Phys. {\bf 31}, 410-415 (1990)

\bibitem{Finster:2015xma}
F.~Finster and J.~Smoller,
``A Spectral Representation for Spin-Weighted Spheroidal Wave Operators with Complex Aspherical Parameter,''
Methods Appl. Anal. \textbf{23}, 35-118  (2016),
doi:10.4310/MAA.2016.v23.n1.a2

\bibitem{Fletcher}
S. J. Fletcher, and A. W.-C. Lun, ``The Kerr spacetime in generalized Bondi–Sachs coordinates.'' Class. Quant. Grav. {\bf 20}, 4153 (2003)


\bibitem{GHP}
R.~P.~Geroch, A.~Held, and R.~Penrose,
``A space-time calculus based on pairs of null directions,''
J. Math. Phys. \textbf{14}, 874-881 (1973), doi:10.1063/1.1666410

\bibitem{GerochAS}
R. P. Geroch, ``Asymptotic structure of space-time.'' In Asymptotic structure of space-time, pp. 1-105. Boston, MA: Springer US, 1977.

\bibitem{GX}
R. P. Geroch and B. C. Xanthopoulos, ``Asymptotic simplicity is stable,'' J. Math. Phys. {\bf 19},  714-719 (1978)

\bibitem{Green:2019nam}
S.~R.~Green, S.~Hollands, and P.~Zimmerman,
``Teukolsky formalism for nonlinear Kerr perturbations,''
Class. Quant. Grav. \textbf{37}, 075001 (2020),
doi:10.1088/1361-6382/ab7075

\bibitem{Green2}
S. R. Green, ``Lorenz-gauge reconstruction for Teukolsky
solutions with sources in electromagnetism," in Presentation at the 24th Capra 
meeting on Radiation Reaction in General Relativity (2021), https://pirsa.org/21060044.

\bibitem{Green3}
S. R. Green, unpublished notes, March 2021

\bibitem{Green:2015kur}
S.~R.~Green, S.~Hollands, A.~Ishibashi and R.~M.~Wald,
``Superradiant instabilities of asymptotically anti-de Sitter black holes,''
Class. Quant. Grav. \textbf{33}, 125022 (2016)
doi:10.1088/0264-9381/33/12/125022

\bibitem{Green:2022htq}
S.~R.~Green, S.~Hollands, L.~Sberna, V.~Toomani and P.~Zimmerman,
``Conserved currents for a Kerr black hole and orthogonality of quasinormal modes,''
Phys. Rev. D \textbf{107} (2023) no.6, 064030
doi:10.1103/PhysRevD.107.064030

\bibitem{Hafner}
D. H\" afner, P. Hintz and A. Vasy, ``Linear Stability of slowly rotating Kerr black holes,'' Invent. Math. {\bf 223}, 1227–1406 (2021)

\bibitem{Held}
A. Held, ``A formalism for the investigation of algebraically special metrics. I.'' Communications in Mathematical Physics {\bf 37}, 311-326 (1974), doi: 10.1007/BF01645944


\bibitem{Hintz2}
P.~Hintz and A.~Vasy,
``The global non-linear stability of the Kerr-de Sitter family of black holes,'' doi:10.4310/acta.2018.v220.n1.a1

\bibitem{Hintz}
P. Hintz, ``A Sharp Version of Price’s Law for Wave Decay on Asymptotically Flat Spacetimes.'' Commun. Math. Phys. {\bf 389}, 491–542 (2022). 
https://doi.org/10.1007/s00220-021-04276-8

\bibitem{Holzegel}
P. Hintz and G. Holzegel, {\it Recent progress in general relativity,} Sections 9-11 ICM 2022, eds. D. Baliaev and S. Smirnov (2022), 
doi:10.4171/ICM2022/128

\bibitem{Hollands:2020vjg}
S.~Hollands and V.~Toomani,
``On the radiation gauge for spin-1 perturbations in Kerr-Newman spacetime,''
Class. Quant. Grav. \textbf{38}, 025013 (2021)
doi:10.1088/1361-6382/abc36f

\bibitem{HWcanonicalenergy}
S. Hollands and R. M. Wald, ``Stability of black holes and black branes.'' 
Commun. Math. Phys. {\bf 321}, 629-680 (2013), doi: 10.1007/s00220-012-1638-1

\bibitem{Hollands:2014eia}
S.~Hollands and R.~M.~Wald,
``Quantum fields in curved spacetime,''
Phys. Rept. \textbf{574}, 1-35 (2015), 
doi:10.1016/j.physrep.2015.02.001

\bibitem{Hollands:2024vbe}
S.~Hollands, R.~M.~Wald and V.~G.~Zhang,
``The Entropy of Dynamical Black Holes,''
[arXiv:2402.00818 [hep-th]]

\bibitem{Hollands:2016oma}
S.~Hollands, A.~Ishibashi and R.~M.~Wald,
``BMS Supertranslations and Memory in Four and Higher Dimensions,''
Class. Quant. Grav. \textbf{34}, 155005 (2017),
doi:10.1088/1361-6382/aa777a

\bibitem{hormander1}
L. H\"ormander,  {\it The analysis of linear partial differential operators. III}, 
Springer Berlin (2007), doi: 10.1007/978-3-540-49938-1 

\bibitem{hormander3}
L. H\"ormander, {\it An Introduction to Complex Analysis in Several Variables,} volume 7, North-Holland, Amsterdam, third edition (1990)



\bibitem{Kato}
T. Kato, {\it Perturbation Theory for Linear Operators,} Springer (1995)

\bibitem{Kavanagh:2016idg}
C.~Kavanagh, A.~C.~Ottewill and B.~Wardell,
``Analytical high-order post-Newtonian expansions for spinning extreme mass ratio binaries,''
Phys. Rev. D \textbf{93}, 124038 (2016),
doi:10.1103/PhysRevD.93.124038

\bibitem{Kegeles:1979an}
L.~S.~Kegeles and J.~M.~Cohen,
``Constructive procedure for perturbations of spacetimes,''
Phys. Rev. D \textbf{19}, 1641-1664 (1979), 
doi:10.1103/PhysRevD.19.1641

\bibitem{Keidl1}
T. S. Keidl, A. G. Shah, J. L. Friedman, D.-H. Kim, and L. R. Price, ``Gravitational self-force in a
radiation gauge,'' Phys. Rev. D {\bf 82}, 124012 (2010), Phys. Rev. D {\bf 90}, 109902 (erratum) (2014)

\bibitem{Keidl2} 
T. S. Keidl, A. G. Shah, J. L. Friedman, D.-H. Kim, and L. R. Price, ``Conservative, gravitational self- force for a particle in circular orbit around a Schwarzschild black hole in a radiation gauge,'' Phys. Rev. D {\bf 83}, 064018 (2011)

\bibitem{Kinnersley:1969zza}
W.~Kinnersley,
``Type D Vacuum Metrics,''
J. Math. Phys. \textbf{10}, 1195-1203 (1969), 
doi:10.1063/1.1664958

\bibitem{Klainerman:2021qzy}
S.~Klainerman and J.~Szeftel,
``Kerr stability for small angular momentum,''
Pure Appl. Math. Quart. \textbf{19}, 791-1678 (2023),
doi:10.4310/PAMQ.2023.v19.n3.a1

\bibitem{kobayashi}
S. Kobayashi and K. Nomizu. {\it Foundations of Differential Geometry,} Volume 2.  John Wiley \& Sons (1996)

\bibitem{Lagos:2022otp}
M.~Lagos and L.~Hui,
``Generation and propagation of nonlinear quasinormal modes of a Schwarzschild black hole,''
Phys. Rev. D \textbf{107}, 044040  (2023),
doi:10.1103/PhysRevD.107.044040


\bibitem{Loutrel:2020wbw}
N.~Loutrel, J.~L.~Ripley, E.~Giorgi and F.~Pretorius,
``Second Order Perturbations of Kerr Black Holes: Reconstruction of the Metric,''
Phys. Rev. D \textbf{103}, 104017  (2021),
doi:10.1103/PhysRevD.103.104017

\bibitem{Ma}
S. Ma and L. Zhang, ``Sharp Decay for Teukolsky Equation in Kerr Spacetimes.'' Commun. Math. Phys. {\bf 401}, 333–434 (2023),
doi:10.1007/s00220-023-04640-w

\bibitem{Ma:2024qcv}
S.~Ma and H.~Yang,
``Excitation of quadratic quasinormal modes for Kerr black holes,''
Phys. Rev. D \textbf{109} (2024) no.10, 104070
doi:10.1103/PhysRevD.109.104070
[arXiv:2401.15516 [gr-qc]].

\bibitem{xact1}
J. M. Mart\' in-Garc\' ia, ``xPerm: fast index canonicalization for tensor computer algebra,'' Comput. Phys. Commun. {\bf 179}, 597
(2008).

\bibitem{xact2}
J. M. Mart\' in-Garc\' ia, ``xAct: Efficient tensor computer algebra for
Mathematica,'' http://xact.es/.

\bibitem{Mathews}
J.~Mathews, A.~Pound and B.~Wardell,
``Self-force calculations with a spinning secondary,''
Phys. Rev. D \textbf{105}, 084031 (2022),
doi:10.1103/PhysRevD.105.084031

\bibitem{van2015metric}
M.~van de Meent and A.~G.~Shah, 
``Metric perturbations produced by eccentric equatorial orbits around a {K}err black hole,''
Phys. Rev. D \textbf{92}, 064025 (2015), doi: 10.1103/PhysRevD.92.064025

\bibitem{meent2}
M. van de Meent, ``Gravitational self-force on generic bound geodesics in Kerr spacetime,'' Phys.
Rev. D {\bf 97} 104033 (2018)



\bibitem{Mitman:2022qdl}
K.~Mitman, M.~Lagos, L.~C.~Stein, S.~Ma, L.~Hui, Y.~Chen, N.~Deppe, F.~H\'ebert, L.~E.~Kidder and J.~Moxon, \textit{et al.}
``Nonlinearities in Black Hole Ringdowns,''
Phys. Rev. Lett. \textbf{130}, 081402 (2023),
doi:10.1103/PhysRevLett.130.081402

\bibitem{Ori:2002uv}
A.~Ori,
``Reconstruction of inhomogeneous metric perturbations and electromagnetic four potential in Kerr space-time,''
Phys. Rev. D \textbf{67}, 124010 (2003), 
doi:10.1103/PhysRevD.67.124010

\bibitem{Poisson:2011nh}
E.~Poisson, A.~Pound, and I.~Vega,
``The motion of point particles in curved spacetime,''
Living Rev.~Rel. \textbf{14}, 7 (2011), 
doi: 10.12942/lrr-2004-6

\bibitem{Pound:Wardell}
A. Pound and B. Wardell, ``Black Hole Perturbation Theory and Gravitational Self-Force''. In: Bambi, C., Katsanevas, S., Kokkotas, K.D. (eds) Handbook of Gravitational Wave Astronomy, Springer, Singapore (2021),
doi: 10.1007/978-981-15-4702-7\_38-1


\bibitem{Pound:2013faa}
A.~Pound, C.~Merlin, and L.~Barack,
``Gravitational self-force from radiation-gauge metric perturbations,''
Phys. Rev. D \textbf{89}, 024009 (2014),
doi:10.1103/PhysRevD.89.024009

\bibitem{Pound:2012dk}
A.~Pound,
``Nonlinear gravitational self-force. I. Field outside a small body,''
Phys. Rev. D \textbf{86}, 084019 (2012), 
doi:10.1103/PhysRevD.86.084019

\bibitem{Prabhu:2018jvy}
K.~Prabhu and R.~M.~Wald,
``Canonical Energy and Hertz Potentials for Perturbations of Schwarzschild Spacetime,''
Class. Quant. Grav. \textbf{35}, 235004 (2018),
doi:10.1088/1361-6382/aae9ae

\bibitem{Price}
L. R. Price, \emph{Developments in the perturbation theory of algebraically special spacetimes}, PhD thesis U. Florida (2007)

\bibitem{Price2} 
L. R. Price, K. Shankar, and B. F. Whiting, ``On the existence of radiation gauges in Petrov type
II spacetimes,'' Class. Quant. Grav. {\bf 24}, 2367 (2007), doi: 10.1088/0264-9381/24/9/014

\bibitem{RPrice}
R. H. Price, ``Nonspherical perturbations of relativistic gravitational collapse. II. Integer-spin, zero-rest-mass fields,'' Phys. Rev. D {\bf 5} (10), 2439 (1972)

\bibitem{Ringstrom}
H. Ringstr\" om,  {\it The Cauchy problem in general relativity.} Vol. 6. European Mathematical Society (2009)

\bibitem{Sberna:2021eui}
L.~Sberna, P.~Bosch, W.~E.~East, S.~R.~Green and L.~Lehner,
``Nonlinear effects in the black hole ringdown: Absorption-induced mode excitation,''
Phys. Rev. D \textbf{105} (2022) no.6, 064046
doi:10.1103/PhysRevD.105.064046

\bibitem{Shankar}
S. Shankar, ``The Nullstellensatz for Systems of PDE,''
Adv. Appl. Math. {\bf 23}, 360-374 (1999),
https://doi.org/10.1006/aama.1999.0657

\bibitem{Shlapentokh-Rothman:2023bwo}
Y.~Shlapentokh-Rothman and R.~T.~da Costa,
``Boundedness and decay for the Teukolsky equation on Kerr in the full subextremal range $|a|<M$: physical space analysis,''
[arXiv:2302.08916 [gr-qc]].




\bibitem{Spiers:2023cip}
A. Spiers, A. Pound, and J. Moxon, ``Second-order Teukolsky formalism in Kerr spacetime: Formulation and nonlinear source,'' Phys. Rev. D {\bf 108}, 064002 (2023), doi: 10.1103/PhysRevD.108.064002

\bibitem{Spiers:2023mor}
A. Spiers, A. Pound, and B. Wardell, ``Second-order perturbations of the Schwarzschild spacetime: practical, covariant and gauge-invariant formalisms,'' arXiv:2306.17847 [gr-qc] 

\bibitem{Starobinsky1974}
A. A. Starobinsky and S. M. Churilov. “Amplification of electromagnetic and gravitational waves scattered by a rotating “black hole””. J. Exp. Theor. Phys. {\bf 38}, 1 (1974)

\bibitem{StewardWalker}
J. M. Steward and M. Walker, ``Perturbations of space-times in general relativity,'' Proc. R. Soc. Lond. A. {\bf 341}, 49-74 (1974)

\bibitem{Breuer}
J. M. Steward, ``On the Stability of Kerr's Space-Time,''
Proc. R. Soc. London A. {\bf 344}, 65-79 (1975)

\bibitem{TdC20}
R. Teixeira~da~Costa, ``Mode stability for the Teukolsky equation on extremal and subextremal Kerr spacetimes,'' Commun. Math. Phys. {\bf 378},  705–78
(2020),
doi: 10.1007/s00220-020-03796-z

\bibitem{Teukolsky:1973ha}
S.~A.~Teukolsky,
``Perturbations of a rotating black hole. 1. Fundamental equations for gravitational electromagnetic and neutrino field perturbations,''
Astrophys. J. \textbf{185}, 635-647 (1973), 
doi:10.1086/152444

\bibitem{Teukolsky:1972my}
S.~A.~Teukolsky,
``Rotating black holes - separable wave equations for gravitational and electromagnetic perturbations,''
Phys. Rev. Lett. \textbf{29},  1114-1118 (1972),
doi: 10.1103/PhysRevLett.29.1114

\bibitem{Teukolsky1974}
S. A. Teukolsky and W. H. Press. ``Perturbations of a rotating black hole. III-Interaction of the hole with gravitational and electromagnetic radiation.'' APJ, 
{\bf 193}, 443-461  (1974)

\bibitem{HT2}
V.~Toomani, P.~Zimmerman, A.~Spiers, S.~Hollands, A.~Pound, and S.~R.~Green,
``New metric reconstruction scheme for gravitational self-force calculations,''
Class. Quant. Grav. \textbf{39}, 015019 (2022),
doi:10.1088/1361-6382/ac37a5

\bibitem{Wald}
R. M. Wald, ``Construction of Solutions of Gravitational, Electromagnetic, or Other Perturbation Equations from Solutions of Decoupled Equations,''
Phys. Rev. Lett. {\bf 41}, 203 (1979), doi: 10.1103/PhysRevLett.41.203

\bibitem{Waldthm}
R. M. Wald, ``On perturbations of a Kerr black hole.'' J. Math. Phys. {\bf 14}, 1453-1461 (1973)

\bibitem{Waldbook}
R. M. Wald, {\it General Relativity,} University of Chicago Press (1984)

\bibitem{Wardell:2021fyy}
B.~Wardell, A.~Pound, N.~Warburton, J.~Miller, L.~Durkan and A.~Le Tiec,
``Gravitational Waveforms for Compact Binaries from Second-Order Self-Force Theory,''
Phys. Rev. Lett. \textbf{130}, 241402 (2023),
doi:10.1103/PhysRevLett.130.241402

\bibitem{Wardell2}
B. Wardell, private communication.


\bibitem{Yang:2014tla}
H.~Yang, A.~Zimmerman and L.~Lehner,
``Turbulent Black Holes,''
Phys. Rev. Lett. \textbf{114}, 081101 (2015),
doi:10.1103/PhysRevLett.114.081101



\end{thebibliography}
\end{document}